\documentclass[12pt,a4paper]{article}

 \usepackage[top=2.0cm,bottom=2.1cm,left=2.0cm,right=2.0cm]{geometry}
\usepackage{latexsym,amsmath,amsfonts,amssymb,amsthm,amscd,mathrsfs}

\usepackage{accents}

 \usepackage[unicode,colorlinks,plainpages=false,hyperindex=true,bookmarksnumbered=true,bookmarksopen=false,pdfpagelabels]{hyperref}
 \hypersetup{urlcolor=cyan,linkcolor=blue,citecolor=red,colorlinks=true}

\usepackage[small,nohug]{diagrams}

\newcommand{\be}{\begin{equation}}
\newcommand{\ee}{\end{equation}}
\newcommand{\bea}{\begin{eqnarray}}
\newcommand{\eea}{\end{eqnarray}}

\newcommand{\ba}{\begin{aligned}}
\newcommand{\ea}{\end{aligned}}

\numberwithin{equation}{section}
\newcounter{thmcounter}
\numberwithin{thmcounter}{section}

\theoremstyle{definition}

\newtheorem{definition}[thmcounter]{Definition}
\newtheorem{remark}[thmcounter]{Remark}
\theoremstyle{plain}
\newtheorem{corollary}[thmcounter]{Corollary}
\newtheorem{lemma}[thmcounter]{Lemma}
\newtheorem{proposition}[thmcounter]{Proposition}
\newtheorem{theorem}[thmcounter]{Theorem}

\def\1{{\boldsymbol 1}}                     %
\def\cH{{\mathcal H}}                       %
\def\cP{{\mathcal P}}                       %
\def\tr{\mathrm{tr}}                        %
\def\diag{\mathrm{diag}}                    %
\def\ri{{\rm i}}                            %
\def\C{\mathbb{C}}                          %
\def\N{\mathbb{N}}                          %
\def\R{\mathbb{R}}                          %
\def\T{\mathbb{T}}                          %
\def\U{{\rm U}}                            %
\def\GL{{\rm GL}(n,\C)}                           %
\def\fH{\mathfrak{H}}                       %
\def\cF{{\mathcal F}}                       %
\def\fP{{\mathfrak{P}}}                     %
\def\reg{\mathrm{reg}}                      %
\def\red{\mathrm{red}}                      %
\def\cS{{\mathcal S}}                       %
\def\Ad{{\mathrm{Ad}}}                      %
\def\id{{\mathrm{id}}}                      %
\def\cA{{\mathcal A}}                       %
\def\dt {\left.\frac{d}{dt}\right|_{t=0}}   %
\def\u{\mathfrak{u}}                        %
\def\b{\mathfrak{b}}                        %
\def\B{\mathrm{B}}                          %
\def\gl{\mathfrak{gl}(n,\C)}                %
\def\cN{{\mathcal N}}                       %
\def\cU{{\mathcal U}}                       %
\def\cM{{\mathcal M}}                       %
\def\cL{{\mathcal L}}                       %
\def\sgn{\mathrm{sgn}}
\def\bb{{\mathbf{b}}}
\def\cG{{\mathcal G}}                       %
\def\fC{{\mathfrak{C}}}                   %
\def\cW{{\mathcal W}}                       %
\def\Dress{\mathrm{Dress}}
\def\dress{\mathrm{dress}}
\def\cV{{\mathcal V}}                       %
\def\cZ{{\mathcal Z}}                       %
\def\cB{{\mathcal B}}                       %
\def\bV{{\mathbf V}}
\def\cE{{\mathcal E}}
\def\cI{{\mathcal I}}
\newcommand\br[1]{\{ #1 \}}
\newcommand\brM[1]{\{ #1 \}_\cM}
\def\ic{{\rm i}}  

\def\ccM{\check{\mathcal M}}
\newcommand{\BE}{\ensuremath{\mathbf{e}}}
\newcommand\wt[1]{\ensuremath{\mathrm{wt}[#1]}}
\def\Z{\mathbb{Z}}
\newcommand{\Mat}{\operatorname{Mat}}
\newcommand{\Ical}{\ensuremath{\mathcal{I}_L}}

\def\dag{{\dagger}}
\def\half{{\textstyle{\frac12}}}
\def\quarter{\textstyle{\frac14}}
\def\k{\mathfrak{u}}
\def\mk{e_k}  \def\mr{e_r}  \def\mn{e_n}

\begin{document}

\begin{center}
 {\Large\bf
Trigonometric real form of the spin RS model of Krichever and Zabrodin
 }
\end{center}

\medskip
\begin{center}
M.~Fairon${}^{a}$, L.~Feh\'er${}^{b,c}$ and I.~Marshall${}^{d}$
\\

\bigskip
${}^a$School of Mathematics and Statistics, University of Glasgow\\
University Place, G12 8QQ Glasgow, United-Kingdom\\
e-mail: maxime.fairon@glasgow.ac.uk

\medskip
${}^b$Department of Theoretical Physics, University of Szeged\\
Tisza Lajos krt 84-86, H-6720 Szeged, Hungary\\
e-mail: lfeher@physx.u-szeged.hu

\medskip
${}^c$Institute for Particle and Nuclear Physics\\
Wigner Research Centre for Physics\\
 H-1525 Budapest, P.O.B.~49, Hungary

\medskip
${}^d$Faculty of Mathematics, Higher School of Economics\\
 National Research University\\
 Usacheva 6, Moscow, Russia\\
 e-mail: imarshall@hse.ru

\end{center}

 \setcounter{tocdepth}{2}

\medskip
\begin{abstract}
We investigate the trigonometric real form of the spin Ruijsenaars--Schneider system introduced,
at the level of equations of motion,
 by Krichever and Zabrodin in 1995.
This pioneering work and all earlier studies of the Hamiltonian interpretation of the system
were performed in complex holomorphic settings; understanding the real forms is a non-trivial problem.
We explain that the trigonometric real form emerges from Hamiltonian reduction of an obviously integrable
`free' system carried by a spin extension of the Heisenberg double of the $\U(n)$ Poisson--Lie group.
The  Poisson structure on the unreduced real phase space $\GL \times \C^{nd}$  is the direct product
of that of the Heisenberg double and $d\geq 2$ copies of a $\U(n)$ covariant  Poisson structure
on $\C^n \simeq \R^{2n}$ found by Zakrzewski, also in 1995.
 We reduce by fixing a group valued moment map to a
 multiple of the identity, and analyze the resulting reduced system in detail.
 In particular, we derive on the reduced phase space the Hamiltonian
 structure of the trigonometric spin Ruijsenaars--Schneider
 system and we prove its degenerate integrability.
\end{abstract}

\newpage

{\linespread{0.7}\tableofcontents}

\newpage
\section{Introduction}
\label{sec:I}

The unbroken interest in integrable many-body systems of Calogero--Moser--Sutherland
\cite{Cal,M,S}
and Ruijsenaars--Schneider (abbreviated RS) \cite{RS} types is due to their ubiquity in physical applications
and rich web of connections to important areas of mathematics \cite{A,E,N,RBanff,vDV}.
The same can be said about spin extensions of these models, which
currently attract attention
\cite{AO,AR,CF2,Fai,F1,F3,KLOZ,Li,P,Res2,Res3,ReSt,SS,Zo}.

Two kinds of spin
many-body models are studied in the literature. Those that feature only `collective spin variables'
belonging to some group theoretic phase space such as a coadjoint orbit, and
those that have `spin-vectors' embodying internal degrees of freedom
of the interacting particles. The former type of models arise rather
naturally in harmonic analysis and its classical mechanical
counterpart  \cite{E,EFK,FP2,FP3,LX,ReSt}.
The latter type of models, built on `individual spins', were
introduced at the non-relativistic level by  Gibbons and Hermsen \cite{GH},
 and their `relativistic' generalization was later put forward by  Krichever and Zabrodin \cite{KZ}.

In fact,
in 1995 Krichever and Zabrodin introduced a family of spin RS models
at the level of equations of motion and posed the question of their
Hamiltonian structure and integrability. These models have rational, trigonometric/hyperbolic
and elliptic versions and are usually studied in the holomorphic category.
The elliptic model encodes the dynamics of the poles of elliptic solutions of the
2D non-Abelian Toda lattice \cite{KZ}, and  a special hyperbolic degeneration is related
to affine Toda solitons \cite{BH}.
The existence of a Hamiltonian structure was established by Krichever \cite{Kri} in the general case
based on a universal construction that is hard to make explicit (see also \cite{Sol}).
The rational case
was treated via Hamiltonian reduction by Arutyunov and Frolov \cite{AF} in 1997, utilizing a `spin extension'
of the holomorphic cotangent bundle of $\GL$.
More than twenty years later, there appeared two different treatments of the holomorphic
trigonometric/hyperbolic  models: by Chalykh and Fairon \cite{CF2} based on double brackets and
quasi-Hamiltonian structures,
and by Arutyunov and Olivucci \cite{AO} based on Hamiltonian reduction of a spin extension
of the  Heisenberg double \cite{STS} of the standard factorizable Poisson--Lie group structure
on $\GL$. In the present paper,  we shall deal with the trigonometric real form
of the models of \cite{KZ}
utilizing the Heisenberg double of the Poisson--Lie group $\U(n)$, which is a natural
generalization of $T^* \U(n)$. The spin extension of the Heisenberg double that we consider is
based on a $\U(n)$ covariant Poisson structure on $\C^n$ introduced by Zakrzewski \cite{Z}.

Although the holomorphic systems are of great interest from several viewpoints, it is not easy
to extract from them the features of the dynamics of the real forms,
which should also be investigated.
For motivation, it perhaps suffices to recall that all pioneering papers of the subject \cite{Cal,M,RS,S}
are devoted to point particles moving along the \emph{real} line or circle.

The $\C$-valued dynamical variables of the Krichever--Zabrodin model are
`particle positions' $x_i$ ($i=1,\dots, n)$ together with $d$-component row vectors $c_i$ and
column vectors
$a_i$.
The composite spin variables $F_{ij}$ are built from these individual spins according to the
rule
\be
F_{ij}:= c_i \cdot a_j := \sum_{\alpha=1}^d c_i^\alpha a_j^\alpha,
\label{I1}\ee
and the equations of motion can be written in first order form as follows:
\be
\dot{x}_i = F_{ii},
\quad
\dot{a}_i^\alpha= \lambda_i a_i^\alpha + \sum_{k\neq i}  V(x_{ik}) a_k^\alpha F_{ki} ,
\quad
\dot{c}_j^\alpha= - \lambda_j c_j^\alpha - \sum_{k\neq j}  V(x_{kj})  c_k^\alpha F_{jk},
\label{I2}\ee
where $x_{ik}:= x_i - x_k$. In the elliptic case the `potential' is given by
$V(x)= \zeta(x) - \zeta(x+ \gamma)$ with the Weierstrass zeta-function and an arbitrary complex
`coupling constant' $\gamma\neq 0$. The model admits hyperbolic/trigonometric degenerations
for which one has $V^{\mathrm{hyp}}(x) = \coth(x) - \coth(x+\gamma)$ and
$V^{\mathrm{rat}}(x) = x^{-1} - (x+\gamma)^{-1}$.
The parameters $\lambda_i$ in \eqref{I2} are arbitrary.
This is a hallmark of  gauge invariance, and thus
it is natural to declare that the `physical observables'  are invariant with respect to arbitrary rescalings
\be
a_i \mapsto \Lambda_i^{-1} a_i, \quad c_i \mapsto \Lambda_i c_i,
\label{I3}\ee
where the $\Lambda_i$ may depend on the dynamical variables as well.
One way to deal with this ambiguity is to impose a gauge fixing condition.
Note also the interesting feature of the model that the spins $a_i, c_i$ are
not purely internal degrees of freedom, since they directly encode the velocities
through the equations of motion $\dot{x}_i = F_{ii}$.

In the trigonometric real form of our interest we put $x_j:= \frac{1}{2}q_j$, where
the $q_j$ are real and are regarded as angles.
In other words, we deal with particles located on the unit circle at the points
$Q_j:= \exp(\ri q_j)$. The spins $c_i$ and $a_i$ are complex conjugates of each other, and
we parametrize them as
\be
c_i^\alpha = v(\alpha)_i, \quad
a_i^\alpha =\overline{ v}(\alpha)_i,
\label{I4}\ee
where $v(-)_i$ is regarded as a $d$-component row vector.
For each $\alpha$, $v(\alpha)$ is also viewed as an $n$-component column vector, and thus
$F = \sum_\alpha v(\alpha) v(\alpha)^\dagger$ is an $n$ by $n$ Hermitian matrix.
The potential $V$ is now chosen to be
\be
V(x) := \cot(x) - \cot(x - \ri \gamma)
\label{I5}\ee
with a real, positive coupling constant $\gamma$. The gauge transformations are given by arbitrary
$\Lambda_i \in \U(1)$ and accordingly we have $\lambda_i \in \ri \R$.
It can be checked that these reality constraints are consistent with the
equations of motion \eqref{I2}.
We remark in passing that they imply the second order equation
\be
\frac{1}{2}\ddot{q}_i = \sum_{j\neq i} F_{ij}F_{ji}
\left[ V\left(\frac{q_{ij}}{2}\right) - V\left(\frac{q_{ji}}{2}\right)
\right]
= \sum_{j\neq i} \vert F_{ij} \vert^2 \frac{2\cot(\frac{q_{ij}}{2})}{1+\sinh^{-2}(\gamma)\,\sin^2(\frac{q_{ij}}{2})}.
\label{I6}\ee

The equations of motion as given above are local in the sense that one does not know on what
phase space their flow is complete, which is required for an integrable system.
Neglecting this issue, let us assume that $\sum_\alpha v(\alpha)_i \neq 0$ for all $i$,
which permits us to impose the conditions
\be
 \cU_i:= \sum_\alpha v(\alpha)_i >0.
 \label{I7}\ee
Note that \eqref{I7} is a gauge fixing in disguise, since it amounts to setting the phase of $\cU_i$ to $0$.
Then, consistency with the requirement $\Im(\cU_i) =0$ can
be used to uniquely determine the $\lambda_i$, and one finds the gauge fixed equations
of motion
\be
\frac{1}{2}\dot{q}_i =  F_{ii},
\qquad
\dot{v}(\alpha)_j = \ri \eta_j v(\alpha)_j
 - \sum_{\ell\neq j} F_{jl} v(\alpha)_l V\left(\frac{q_{lj}}{2}\right),
\label{I8}\ee
with
\be
\ri \eta_j =
\frac{1}{2}\sum_{\ell \neq j} \frac{\cU_l}{\cU_j} \left[ F_{jl} V\left(\frac{q_{lj}}{2}\right)
  + F_{lj} V\left(\frac{q_{jl}}{2}\right)\right].
  \label{I9}\ee

In this paper, we develop a Hamiltonian reduction approach to the real, trigonometric
spin RS model specified above. In particular, this yields a phase space on which
all flows of interest are complete. An open dense subset of the phase space will be associated
with the gauge fixing condition \eqref{I7}, and on this submanifold we shall determine
the explicit form of the Poisson brackets that generate the equations of motion \eqref{I8}
by means of the Hamiltonian
\be
\cH=  \sum_{i} F_{ii}.
\label{I10}\ee
We shall also prove the degenerate integrability of the model by displaying
 $(2nd -n)$ independent, real-analytic integrals of motion that form a polynomial Poisson algebra
  whose $n$-dimensional Poisson center contains $\cH$.
 These results will be derived by using the $q_i$ and gauge fixed versions of the
 `dressed spins' $v(\alpha)$ as coordinates on the reduced phase space obtained from
 Hamiltonian reduction.
 However, we will put forward another remarkable set of variables as well, which
consists of canonical pairs $q_i, p_i$ and `reduced primary spins' $w^\alpha$ that
decouple from $q$ and $p$ under the reduced Poisson bracket.
On the overlap of their dense domains,
the relation between the two sets of variables
can be given explicitly, but the formula is very involved.
The drawback of the variables $q, p, w^\alpha$ is that in terms of them  $\cH$ and
 the equations of motion become
complicated.

Notice that the Newton equations \eqref{I6} imply the conservation of the sum of
the velocities $\dot{q}_i$, which gives the Hamiltonian via \eqref{I8} and \eqref{I10}, and $\dot{q}_i$
is non-negative by \eqref{I8}.
The same features appear in the spinless chiral RS model \cite{RS} defined by the Hamiltonian
\be \label{I:Hplus}
 \cH_{\mathrm{RS}}^+=
\sum_i e^{2\theta_i} \prod_{j\neq i} \left[1+ \frac{\sinh^2\gamma}
{1 + \sin^2\frac{q_i- q_j}{2}}\right]^{\frac{1}{2}},
\ee
with Darboux coordinates\footnote{This form of the chiral RS Hamiltonian is the one found in \cite{RS}.
Different Darboux variables, which avoid the appearance of square roots in the
Hamiltonian, are also often used in the literature.}
$q_i, \theta_j$.
The second order equations of motion for $q_i$ generated by this Hamiltonian reproduce
the $d=1$ special case of \eqref{I6}. To see this, note that
$F_{ij} F_{ji} = F_{ii} F_{jj}$ if $d=1$, and substitute $\dot{q_i}=2 F_{ii}$ from \eqref{I8} into
\eqref{I6}.
In fact, the spinless RS model results from the $d=1$ special case of our Hamiltonian reduction:
in this case $w^1$ becomes gauge equivalent to a constant vector and
one derives the model utilizing also a canonical
 transformation between $q, p$ mentioned above and $q,\theta$ \cite{FK2}.
Thus, the spin RS systems of \cite{KZ}  are generalizations of the \emph{chiral} RS model.
We follow the general practice in dropping `chiral' from their name.

Our result on the degenerate integrability of the system is not surprising, since the
same property holds in the complex holomorphic case \cite{AO,CF2} and it also holds
generically for large families of related spin many-body models obtained
by Hamiltonian reduction \cite{Res1,Res2,Res3}.
Despite these earlier results, the degenerate integrability of our specific real system
can not be obtained directly. Therefore,  it requires a  separate treatment,
and we shall exhibit the desired integrals of motion in explicit form.

Here is an outline of this work and its main results.
In Section \ref{sec:T}, we present the master phase space $\cM$ which is an
extension of the Heisenberg double of $\U(n)$ by a space of primary spins.
The latter space is formed of $d\geq 2$  copies of $\C^n$ endowed with a $\U(n)$ covariant
Poisson bracket and a (Poisson--Lie) moment map, see Proposition \ref{Pr:Zak}.
We also introduce the `free'
degenerate integrable system on $\cM$ that will be reduced.
In Section \ref{sec:H}, we define the Hamiltonian reduction and progress towards the
description of the corresponding reduced phase space $\cM_{\red}$, which is a real-analytic
symplectic manifold of dimension $2nd$.
In particular, we exhibit two models of dense open
subsets of $\cM_{\red}$; the first one is used in the subsequent sections to
 derive the real form of the trigonometric spin RS system described above,
 while the second one allows  us to prove that $\cM_{\red}$ is connected and it
 leads to a concise formula for the reduced symplectic
 form (see Theorem \ref{Th:redsymp} and Corollary \ref{Cor:Dense}).
 The second model will also be used for recovering the Gibbons--Hermsen system through a scaling limit
 (see Remark \ref{Rem:scaling}).
In Section \ref{sec:N}, we characterize the projection of a family of free Hamiltonian vector fields of
$\cM$ onto $\cM_{\red}$, and show in Corollary \ref{Cor:vf}  that one of these
projections reproduces the equations of motion \eqref{I8}.
Then, in Section \ref{sec:RP}, we obtain the reduced Poisson bracket presented in  Theorem \ref{Thm:redPB}.
 This offers an alternative way to derive the equations of motion \eqref{I8}, and
  we also provide a formula
 for the Poisson bracket of the
 Lax matrix that generates the commuting reduced Hamiltonians, see Proposition \ref{Pr:RUform}.
In Section \ref{sec:DIS}, we demonstrate the degenerate integrability of the real trigonometric
spin RS system, with the final
result formulated as Theorem \ref{Th:last}.
Section \ref{sec:Con} concludes this work and gathers open questions. There are four appendices devoted to auxiliary results
and proofs.

\medskip
\noindent
\textbf{Note on conventions.} The sign function $\sgn$ is such that $\sgn(i-k)$ is $+1$ if $i>k$, $-1$ if $i<k$, and
$0$ for $i=k$.  Similarly to Kronecker's delta function,
 we define for any condition $\mathrm{c}$ the symbol $\delta_{\mathrm{c}}$ which equals $+1$
if $\mathrm{c}$ is satisfied, and $0$ otherwise. For example, $\delta_{(j<l\leq k)}$ equals $+1$ if $j<l$ and $l\leq k$,
while it is $0$ if one of those two conditions is not satisfied.

\section{Heisenberg double, primary spins and  `free' integrable system}
\label{sec:T}

Eventually, we shall obtain the real, trigonometric spin RS system by reduction of an
`obviously integrable' system on the phase space $\cM  :=  M \times \C^{n\times d}$,
where $M$ is the Heisenberg double of the Poisson--Lie group $\U(n)$ and $\C^{n\times d}$
is the space of the so-called primary spin variables.  In this section
we present a quick overview of these structures, to be used in the subsequent sections.
More details can be found in the references \cite{F3,FK2,Kli,KS,Lu1,STS} and in Appendix \ref{sec:A}.

\subsection{The Heisenberg double and its models}

Let us start with the real vector space direct sum
\be
\gl = \u(n) + \b(n),
\label{T1}\ee
where $\b(n)$ denotes the Lie algebra of upper triangular complex matrices having real
entries along the diagonal, and the unitary Lie algebra $\u(n)$ consists of the
skew-Hermitian matrices. These are isotropic subalgebras with respect to the non-degenerate,
invariant bilinear form of $\gl$ given by
\be
\langle X, Y \rangle := \Im \tr(XY),
\quad
\forall X,Y \in \gl,
\label{T2}\ee
which means that we have a Manin triple  at hand.
Then define
\be
R:= \frac{1}{2} \left( P_{\u(n)} - P_{\b(n)}\right),
\label{T3}\ee
using the projection operators with ranges $\u(n)$ and $\b(n)$,
associated with the decomposition \eqref{T1}. For any $X\in \gl$, we may write
\be
X= X_{\u(n)} + X_{\b(n)} \quad \hbox{with}\quad X_{\u(n)} = P_{\u(n)}(X),\, X_{\b(n)} = P_{\b(n)}(X).
\label{T4}\ee
As a manifold, $M$ is the real Lie group $\GL$, and for any smooth \emph{real} function
$f\in C^\infty(\GL)$ we introduce the $\gl$-valued derivatives $\nabla f$ and $\nabla'f$ by
\be
\langle \nabla f(K), X\rangle := \dt f(e^{tX} K),
\quad
\langle \nabla' f(K), X\rangle := \dt f(Ke^{tX}),\quad
\forall X\in \gl,
\label{T5}\ee
where $K$ denotes the variable running over $\GL$.
The commutative algebra of smooth real functions, $C^\infty(M)$, carries two natural Poisson brackets
provided by
\be
\{ f, h\}_{\pm} := \langle \nabla f, R \nabla h \rangle \pm  \langle \nabla' f, R \nabla' h \rangle.
\label{T6}\ee
The minus bracket makes $\GL$ into a real Poisson--Lie group, while the plus one
corresponds to a symplectic structure on $M$. The former is called the Drinfeld double Poisson bracket and the latter
the Heisenberg double Poisson bracket \cite{STS}.

The real Poisson brackets can be extended to complex functions by requiring complex
bilinearity. Then the real Poisson brackets can be recovered if we know all  Poisson brackets
between the matrix elements of $K$ and its complex conjugate $\overline{K}$.
In the case of the Drinfeld double, we have
\be
\{ K_{ij}, K_{kl}\}_- =\ri K_{kj} K_{il} \left[ \delta_{ik} + 2 \delta_{(i>k)} - \delta_{lj} - 2 \delta_{(l>j)}\right],
\label{T7}\ee
and
\be
\{ K_{ij}, \overline{K}_{kl} \}_- = \ri K_{ij} \overline{K}_{kl} [ \delta_{ik} - \delta_{jl}]
+ 2 \ri \left[ \delta_{ik} \sum_{\beta > i} K_{\beta j} \overline{K}_{\beta l} - \delta_{jl} \sum_{\alpha <j} K_{i\alpha}
\overline{K}_{k\alpha}\right].
\label{T8}\ee
Consider the subgroup $\B(n)< \GL$ of upper triangular matrices having positive entries along
the diagonal, and the unitary subgroup $\U(n) < \GL$. These subgroups correspond to the subalgebras in \eqref{T1}.
It is well-known that both $\U(n)$ and $\B(n)$ are Poisson submanifolds of the Drinfeld double
$(\GL, \{\ ,\ \}_-)$. We denote their inherited Poisson structures by $\{\ ,\ \}_U$ and $\{\ ,\ \}_B$,
which makes them Poisson--Lie groups.

The Poisson brackets on $C^\infty(\U(n))$ and on $C^\infty(\B(n))$ admit the following description.
For any real function $\phi \in C^\infty(\U(n))$  introduce the
$\b(n)$-valued derivatives $D\phi$ and $D'\phi$ by
\be
\langle D \phi(g), X\rangle := \dt \phi(e^{tX} g),
\quad
\langle D' \phi(g), X\rangle := \dt \phi(ge^{tX}),
\,\, \forall X\in \u(n),
\label{T9}\ee
and for any $\chi \in C^\infty(\B(n))$ similarly introduce the $\u(n)$-valued derivatives $D\chi$ and $D'\chi$.
Then we have
\be
\{\phi_1, \phi_2\}_U(g) = -\left\langle D'\phi_1(g), g^{-1} (D\phi_2(g)) g \right\rangle,
\quad \forall g\in \U(n),
\label{T10}\ee
where the conjugation takes place inside $\GL$. Similarly
\be
\{ \chi_1, \chi_2\}_B(b)= \left\langle D' \chi_1(b), b^{-1} (D\chi_2(b)) b \right\rangle,
\quad \forall b\in \B(n).
\label{T11}\ee
The opposite signs in the last two formulae are due to our conventions.

By the Gram-Schmidt process, every element $K\in \GL$
admits the unique decompositions
\be
K = b_L g_R^{-1} = g_L b_R^{-1}\quad\hbox{with}\quad b_L, b_R \in \B(n),\, g_L, g_R\in \U(n),
\label{T12}\ee
and $K$ can be recovered also from the pairs $(g_L, b_L)$ and $(g_R, b_R)$,
by utilizing the identity
\be
b_L^{-1} g_L = g_R^{-1} b_R.
\label{T13}\ee
These decompositions give rise to the maps $\Lambda_L, \Lambda_R$ into $\B(n)$ and $\Xi_L, \Xi_R$ into $\U(n)$,
\be
\Lambda_L(K):= b_L, \quad \Lambda_R(K):= b_R, \quad \Xi_L(K):= g_L,\quad\Xi_R(K):= g_R.
\label{T14}\ee
Then we obtain the maps from $\GL$ onto $\U(n)\times \B(n)$,
\be
(\Xi_L, \Lambda_R),\quad (\Xi_R, \Lambda_L),\quad (\Xi_L, \Lambda_L),\quad (\Xi_R, \Lambda_R),
\label{T15}\ee
which are all (real-analytic) diffeomorphisms.
In particular we shall use the diffeomorphism
\be
m_1:= (\Xi_R, \Lambda_R): \GL \to \U(n) \times \B(n)
\label{T16}\ee
to transfer the Heisenberg double Poisson bracket to $C^\infty(\U(n) \times \B(n))$.
The formula of the resulting Poisson bracket \cite{F3}, called $\{\ ,\ \}_+^1$, can be written as follows:
\bea
&&\{\cF, \cH\}_+^1(g,b) =\left\langle D_2' \cF, b^{-1} (D_2\cH) b \right\rangle
-\left\langle D'_1\cF, g^{-1} (D_1\cH) g\right\rangle
\nonumber\\
&&\qquad \qquad  +  \left\langle D_1\cF , D_2\cH \right\rangle
-\left\langle D_1 \cH , D_2\cF \right\rangle,
\label{+PB1}\eea
for any $\cF, \cH \in C^\infty(\U(n)\times \B(n))$.
The derivatives on the right-hand side
are   taken at $(g,b)\in \U(n)\times \B(n)$; the subscripts 1 and 2 refer to derivatives with respect to the
first and second arguments.
As an application, one can determine the Poisson brackets between the matrix elements of $(g,b):= (g_R, b_R)$
on the Heisenberg double,  which gives
\be
\{ g_{lm}, b_{jk} \}_+ = \ri \delta_{jl} g_{lm} b_{jk} + 2\ri \delta_{(j<l\leq k)}\, g_{jm} b_{lk},
\label{T18}\ee
and
\be
\{ g_{lm}, \overline{b}_{jk} \}_+ = \ri \delta_{jl} g_{lm} \overline{b}_{jk} +
2\ri \delta_{jl}  \sum_{j<\beta\leq k}\, g_{\beta m} \overline{b}_{\beta k}.
\label{T19}\ee
The same formulae are valid w.r.t. $\{\ ,\ \}_+^1$, and this was used for the computation.

Observe from the formula \eqref{+PB1} that both $\Xi_R$ and $\Lambda_R$ are Poisson maps w.r.t. the
Heisenberg double Poisson bracket and the Poisson brackets $\{\ ,\ \}_U$ and $\{\ ,\ \}_B$, respectively.
The same is true regarding the maps $\Xi_L$ and $\Lambda_L$. A further property that we use later is that
\be
\{ \Lambda_L^*(\chi_1), \Lambda_R^*(\chi_2)\}_+ = 0,
\quad
\forall \chi_1, \chi_2 \in C^\infty(\B(n)).
\label{T20}\ee
Note, incidentally, that  $\Xi_L$ and $\Xi_R$ enjoy the analogous identity.

The Heisenberg double admits another convenient model as well.
This relies on the diffeomorphism between $\B(n)$ and the manifold $\fP(n) = \exp(\ri \u(n))$
of positive definite Hermitian matrices, defined by $b \mapsto L:=bb^\dagger$. Then we have
the diffeomorphism
\be
m_2: \GL \to \U(n) \times \fP(n),
\quad m_2:= (\Xi_R, \Lambda_R \Lambda_R^\dagger).
\label{T21}\ee
For $\psi \in C^\infty(\fP(n))$, define the $\u(n)$-valued
derivative $d \psi$ by
\be
\langle d \psi(L), X\rangle := \dt \psi(L + t X),
\,\,\forall X\in \ri \u(n).
\label{T22}\ee
This definition makes sense since $(L + t X)\in \fP(n)$ for small $t$.
By using $m_2$, one can transfer the Poisson bracket $\{\ ,\ \}_+$ to $C^\infty(\U(n) \times \fP(n))$.
The resulting Poisson bracket is called $\{\ ,\ \}_+^2$, and is given \cite{F3} by the following explicit
formula:
\bea
&&\{F, H\}_+^2 (g,L)=4\left\langle L d_2F, \left(L d_2 H \right)_{\u(n)}\right\rangle
-\left\langle D'_1 F , g^{-1} (D_1 H) g  \right\rangle\nonumber \\
&&\qquad \qquad  +  2\left\langle D_1F, L d_2H \right\rangle
-2\left\langle D_1 H, L d_2F\right\rangle,
\label{+PB2}\eea
for any $F,H\in C^\infty(\U(n) \times \fP(n))$, where the derivatives
with respect to the first and second arguments are taken at $(g,L)\in \U(n)\times \fP(n)$.
The subscript $\u(n)$ refers to the decomposition defined in \eqref{T4}.

We end this review of the Heisenberg double by recalling  the symplectic form,
denoted $\Omega_M$,   that corresponds to the non-degenerate Poisson structure
 $\{\ ,\ \}_+$. It can  be displayed \cite{AM} as
\be
\Omega_M=\frac{1}{2}\Im\tr(d \Lambda_L \Lambda_L^{-1}\wedge  d \Xi_L \Xi_L^{-1})+
\frac{1}{2}\Im\tr( d\Lambda_R \Lambda_R^{-1}  \wedge d\Xi_R   \Xi_R^{-1}).
\label{T24}\ee
Here, $d\Lambda_L$ collects the exterior derivatives of the components of the matrix valued function $\Lambda_L$.
To be clear about our conventions, we remark that the wedge does not contain $\frac{1}{2}$, and
the Hamiltonian vector field $X_h$ of $h$ satisfies $d h = \Omega_M(\ , X_h)$ and $\{f,h\}_+ = df(X_h) = \Omega_M(X_h, X_f)$.

\subsection{The primary spin variables} \label{ssec:Prim}

We begin by recalling that the real, trigonometric spin Sutherland model of Gibbons--Hermsen \cite{GH} type
 can be derived via Hamiltonian reduction
of $T^*\U(n) \times \C^{n\times d}$, where $\C^{n\times d} \simeq \R^{2n \times d}$ carries its
canonical Poisson structure. In particular, if the elements of $\C^{n\times d}$ are represented as a collection
of $\C^n$ column vectors
\be
w^1, w^2, \dots, w^d,
\label{T25}\ee
then the $d$ different copies pairwise Poisson commute. The symmetry group underlying the reduction is $\U(n)$, which acts
on $\C^n$ in the obvious manner,
\be
\cA^{(n)}: \U(n) \times \C^n\to \C^n \quad\hbox{given by}\quad \cA^{(n)}(g,w):= g w.
\label{T26}\ee
For our generalization, it is natural to require this to be a Poisson action, i.e., $\cA^{(n)}$ should
be a Poisson map with respect to the Poisson structure \eqref{T10} on $\U(n)$ and a suitable
Poisson structure on $\C^n$. A further requirement is that the $\U(n)$-action should
be generated by a moment map.

Specialized to $\U(n)$ with the Poisson structure \eqref{T10}, the notion
of moment map that we use can be summarized as follows\footnote{In full generality, the concept of
Poisson--Lie moment map goes back to Lu \cite{Lu}; our  conventions are
slightly different from hers.}.
Suppose that we have a Poisson manifold $(\cP, \{\ ,\ \}_\cP)$
and a Poisson map $\Lambda: \cP \to \B(n)$, where $\B(n)$ is endowed with the Poisson structure \eqref{T11}.
Then, for any $X\in \u(n)$,  the following formula defines a vector field $X_{\cP}$ on $\cP$:
 \be
 \cL_{X_{\cP}}(\cF) \equiv X_{\cP}[\cF] :=\left\langle X,  \{\cF, \Lambda\}_\cP \Lambda^{-1} \right \rangle,
 \qquad \forall \cF\in C^\infty(\cP),
\label{T27}\ee
where the Poisson bracket is taken with every entry of the matrix $\Lambda$, and $\cL_{X_\cP}$ denotes
Lie derivative along $X_{\cP}$.
The map $X \mapsto X_{\cP}$ is automatically a Lie algebra anti-homomorphism, representing
an infinitesimal left action of $\U(n)$. If it integrates to a global action of $\U(n)$, then
the resulting action is Poisson, i.e., the action map $\cA: \U(n) \times \cP \to \cP$ is Poisson.
In the situation just outlined, $\Lambda$ is called the (Poisson--Lie) moment map
of the corresponding  Poisson action.

In the next proposition we collect the key properties of a Poisson structure on $\C^n$,
which is a special case of the $\U(n)$ covariant Poisson structures
found by Zakrzewski \cite{Z}.

\begin{proposition} \label{Pr:Zak}
 The following formula defines a Poisson structure on $\C^n \simeq \R^{2n}$:
\bea
  &&\{w_i,w_l\}=\ri\,  \sgn(i-l) w_i  w_l, \qquad \forall 1\leq i,l \leq n,
   \label{T29}\\
&&\{w_i,\overline{w}_l\}=\ri\,  \delta_{il} (2+|w|^2) +\ri w_i \overline{w}_l+ \ri\,  \delta_{il} \sum_{r=1}^n \sgn(r-i) |w_r|^2\,.
\label{T30}\eea
These formulae imply
\be
\{ \overline{w}_i, \overline{w}_l\} = \overline{\{w_i, w_l\}},
\label{T32}
\ee
which means that the Poisson bracket of real functions is real.
With respect to this Poisson bracket, the action \eqref{T26} of $\U(n)$ with \eqref{T10} is Poisson, and is generated
by the moment map $\bb: \C^n \to \B(n)$ given by
\be
\bb_{jj}(w) = \sqrt{\cG_j/\cG_{j+1}},\qquad \bb_{ij}(w) = \frac{w_i \overline{w}_j}{ \sqrt{\cG_j \cG_{j+1}}},
\qquad\forall 1\leq i<j\leq n
\label{bb}\ee
with
\be
\cG_j = 1 + \sum_{k=j}^n \vert w_k \vert^2,\quad
\label{T31}\cG_{n+1}:= 1,
\ee
The map $\bb$ satisfies the identity
\be
\bb(w) \bb(w)^\dagger = \1_n + w w^\dagger.
\label{bbdag}\ee
The Poisson structure is non-degenerate, and the corresponding symplectic form is given by
\be
\Omega_{\C^n} =
  \frac\ri2\sum_{k=1}^n\frac1{\cG_k} dw_k\wedge d\overline w_k
+ \frac\ri4\sum_{k=1}^{n-1}\frac1{\cG_k\cG_{k+1}}
d\cG_{k+1}
  \wedge
 \left(\overline w_k dw_k - w_k d\overline w_k\right).
 \label{Omzak1}\ee
\end{proposition}

\medskip

A variant of the factorization formula \eqref{bbdag}
(without connection to Zakrzewski's Poisson bracket)
 was found earlier by Klim\v c\'\i k, as
presented in an unpublished initial version of \cite{FK2}.
 For convenience, we give a self-contained proof of the proposition in Appendix \ref{sec:A}.

\begin{definition}
The pairwise Poisson commuting $w^1,\dots, w^d$ with each copy subject to the
Poisson brackets \eqref{T29}, \eqref{T30} are called \emph{primary spin  variables}.
The Poisson space obtained in this manner is denoted $\left(\C^{n\times d}, \{\ ,\ \}_\cW\right)$,
and we shall also use the notation
\be
W:= (w^1, \dots, w^d).
\label{Wdef}\ee
The corresponding symplectic form, $\Omega_\cW$, is the sum of $d$-copies of
$\Omega_{\C^n}$ \eqref{Omzak1}, one for each variable $w^\alpha$, $\alpha=1,\dots, d$.
\end{definition}

\subsection{The unreduced `free' integrable system} \label{ss:unredIS}

Let $\fH$ be an Abelian Poisson subalgebra of the Poisson algebra of
(smooth, real-analytic, etc.) functions
on a symplectic manifold $\cM$ of dimension $2N$, such that all elements of $\fH$ generate complete
Hamiltonian flows.  Assume that the functional dimension
of $\fH$ is $r$\footnote{This means that the exterior derivatives
of the elements of $\fH$ span an $r$-dimensional subspace of the cotangent space for generic points of $\cM$,
which form a dense open submanifold.},
and that there exists also a Poisson subalgebra $\fC$ of the functions on $\cM$
 whose functional dimensions is $(2N -r)$ and its center contains  $\fH$.
Then $\fH$ is a called an \emph{integrable system} with Hamiltonians $\fH$
 and algebra of constants
of motion $\fC$.  Liouville integrability is the $r=N$,  $\fC = \fH$, special case.
One calls the system \emph{degenerate integrable} (or non-commutative integrable, or superintegrable)   if $r<N$.
In the degenerate case, similarly to Liouville integrability, the flows of the Hamiltonians belonging to $\fH$
are linear in suitable coordinate systems on the joint level surfaces of $\fC$.
For further details of this notion and its variants, and for the generalization
of the Liouville--Arnold theorem, we refer to the papers \cite{J,MF,Nekh,Res2} and to
Section 11.8 of the book \cite{Rud}.

Consider the Heisenberg double $(M, \{\ ,\ \}_+)$ and the space of primary spins $(\C^{n\times d}, \{\ ,\ \}_{\cW})$
introduced in \S\ref{ssec:Prim}.
Define
\be
\cM:= M \times \C^{n\times d}
\ee
 and equip it with the product Poisson structure,
$\{\ ,\ \}_\cM$, which comes from the symplectic form
\be
\Omega_{\cM} = \Omega_M + \Omega_{\cW}.
\label{OmcM}\ee
Let $C^\infty(\B(n))^{\U(n)}$ denote those functions on $\B(n)$ that are invariant with respect to the dressing action of
$\U(n)$ on $\B(n)$, operating as
\be
\Dress_g(b):= \Lambda_L( g b), \qquad
\forall (g,b) \in \U(n) \times \B(n).
\label{Dress}\ee
It is well-known that these invariant functions form the center of the Poisson bracket $\{\ ,\ \}_B$ \eqref{T11}.
 Extend all maps displayed in \eqref{T14} to $\cM$ in the trivial manner, for example by setting
\be
\Lambda_R(K,W):= \Lambda_R(K)
\quad \hbox{with}\quad (K,W) \in M \times \C^{n\times d}.
\label{T33}\ee
Then
\be
\fH:= \Lambda_R^*\bigl( C^\infty\bigl(\B(n)\bigr)^{\U(n)}\bigr)
\label{T34}\ee
is an Abelian Poisson subalgebra of $C^\infty(\cM)$. We call the elements of $\fH$ `free Hamiltonians' since
their flows are easily written down explicitly.  Indeed, for $H= \Lambda_R^*(h)$ the flow
sends the initial value $(K(0), W(0))$ to
\be
(K(t), W(t)) = \left( K(0) \exp\left( - t Dh(b_R(0))\right), W(0)\right).
\label{T35}\ee
It follows that $b_R$ and $b_L$ are constants along the flow and we have the `free motion' on $\U(n)$
given by
\be
g_R(t) = \exp\left(t Dh(b_R(0))\right) g_R(0).
\label{T36}\ee
The functional dimension of $\fH$ is $n$, and for independent generators one may take
\be
H_k:= \Lambda_R^*(h_k) \quad \hbox{with}\quad h_k(b) = \frac{1}{2k}\tr (bb^\dagger)^k, \quad k=1,\dots, n.
\label{T37}\ee
The invariance of these functions follows from the useful identity
\be
 (\Dress_g(b)) (\Dress_g(b))^\dagger = g (bb^\dagger) g^{-1}.
 \label{Dressconj}\ee
The system is degenerate integrable, with $\fC$ taken to be the algebra of all
constants of motion, which are provided by arbitrary smooth functions depending on $b_L, b_R$ and $W$.
From the decomposition \eqref{T12}, we get
\be
b_R b_R^\dagger = g_R \left(b_L^{-1} (b_L^{-1})^\dagger\right) g_R^{-1},
\label{R38}
\ee
and this entails $n$ relations between the functions of $b_L$ and $b_R$. Thus the functional dimension of $\fC$ is
$2N - n$, with $N= n^2 + nd$, as required.
It is worth noting that the joint level surfaces of $\fC$ are compact, since they can be viewed
as closed subsets of $\U(n)$.

There are several ways to enlarge $\fH$ into an Abelian Poisson algebra of functional dimension $N$, i.e.,
to obtain Liouville integrability of the Hamiltonians in $\fH$. However, there is no canonical way to do so.
Degenerate integrability is a stronger property than Liouville integrability,
since it restricts the flows of the Hamiltonians to smaller level surfaces.
For these reasons, we shall not pay
attention to Liouville integrability in this paper.

\section{Defining the reduction and solving the moment map constraint}
\label{sec:H}

We first describe a Poisson action of $\U(n)$ on $\cM$ and use it for defining the reduction
of the free integrable system.
Then we shall deal with two parametrizations of the `constraint surface', which is obtained
by imposing the moment map constraint of equation \eqref{momconst} below.

\subsection{Definition of the reduction}
\label{ssec:H}

Let us start by introducing the following Poisson map $\Lambda:\cM \to \B(n)$,
\be
\Lambda(K,W) = \Lambda_L(K) \Lambda_R(K) \bb(w^1) \bb(w^2) \cdots \bb(w^d),
\label{H1}\ee
using the notations \eqref{T14}, \eqref{bb}.
The Poisson property of $\Lambda$ holds since all the factors are
separately Poisson maps, and their matrix elements mutually Poisson commute.
The infinitesimal action generated by $\Lambda$, via
 the formula \eqref{T27}, integrates to a global Poisson action of $\U(n)$ on $\cM$.
This action turns out to have a nice form in terms of the new variables on $\cM$
given  below.

\begin{definition} \label{Def:31}
  For $\alpha=1,\dots, d$, introduce
\be
b_\alpha:= \bb(w^\alpha), \qquad
B_\alpha:= b_R b_1 b_2 \cdots b_\alpha, \quad B_0:=b_R,
\label{H2}\ee
and define the
\emph{dressed spins} $v(\alpha)$ and the \emph{half-dressed spins} $v^\alpha$ by the equalities
\be
v(\alpha):= B_{\alpha -1} w^\alpha =: b_R v^\alpha.
\label{H3}\ee
\end{definition}

\begin{lemma} \label{L:NewVar}
The new variables on $\cM$ given by
\be
g_R, b_R, v(1),\dots, v(d)
\label{H4}\ee
are related by a diffeomorphism of $\U(n) \times \B(n) \times \C^{n\times d}$ to the variables
\be
g_R, b_R, w^1,\dots, w^d.
\label{H5}\ee
\end{lemma}

\begin{proof}
We have the relations,
\be
w^1 = b_R^{-1} v(1), \,\, w^2 = \bb(w^1)^{-1} b_R^{-1} v(2),\dots, w^d = \bb(w^{d-1})^{-1} \cdots \bb(w^1)^{-1} b_R^{-1} v(d),
\label{w}\ee
which can be used to reconstruct the variables \eqref{H5} from those in \eqref{H4}.
\end{proof}

The statement analogous to Lemma \ref{L:NewVar} for the variables $(g_R, b_R, v^1,\dots, v^d)$ also holds.
Later in the paper we shall use the following identities enjoyed by the half-dressed spins and the dressed spins,
which are direct consequences of \eqref{bbdag}. These identities and the subsequent proposition
actually motivated the introduction of these variables.

\begin{lemma} \label{L:id1vv}
With the above notations, one has the identities
\be
(b_1 b_2 \dots b_d) (b_1 b_2 \dots b_d)^\dagger = \1_n + \sum_{\alpha =1}^d v^\alpha (v^\alpha)^\dagger,
\quad
B_d B_d^\dagger = b_R b_R^\dagger + \sum_{\alpha=1}^d v(\alpha) v(\alpha)^\dagger.
\label{H7}\ee
\end{lemma}

\begin{proposition}
The moment map $\Lambda: \cM \to \B(n)$ given by \eqref{H1} generates the action
$\cA: \U(n) \times \cM \to \cM$  that operates
as follows:
\be
\cA_\eta: \left(g_R, b_R, v(1),\dots, v(d)\right) \mapsto \left(\tilde\eta g_R \tilde \eta^{-1},
\Dress_{\tilde \eta}(b_R),
\tilde\eta v(1), \dots, \tilde \eta v(d)\right), \quad \forall \eta\in \U(n),
\label{H6}\ee
where $\tilde \eta = \Xi_R(\eta b_L)^{-1}$ with $b_L = \Lambda_L \circ m_1^{-1}(g_R, b_R)$ using \eqref{T16}.
In other words, $b_L = \Lambda_L(K)$ with $K\in M$ parametrized by the pair $(g_R, b_R)$.
\end{proposition}

\begin{proof}
In order to avoid clumsy formulae and the introduction of further notations, in what follows we identify the variables
$g_R, b_R,  B_\alpha, w^\alpha$ and so on with
the associated evaluation functions on $M$, $\cM$ and $\C^n$. We shall also use the infinitesimal dressing action
corresponding to \eqref{Dress}, which has the form
\be
\dress_X(b) = b (b^{-1} X b)_{\b(n)},
\qquad \forall X\in \u(n),\,\, b\in \B(n).
\label{dress}\ee
For any $X\in \u(n)$, denote $X_\cM$, $X_M$, $X_{\C^n}$ the vector fields associated
with the moment maps $\Lambda: \cM\to \B(n)$, $\Lambda_L \Lambda_R: M \to \B(n)$ and $\bb: \C^n \to \B(n)$, respectively.
The formula of $X_M$ is known \cite{FK2,Kli} and $X_{\C^n}$ can be read off from \S\ref{ssec:Prim}.
In fact, we have
\be
\cL_{X_M} g_R = [(b_L^{-1} X b_L)_{\u(n)}, g_R],
\quad
\cL_{X_M} b_R = \dress_{(b_L^{-1} X b_L)_{\u(n)}} (b_R)
\label{H8}\ee
and
\be
\cL_{X_{\C^n}} w= X w, \quad \cL_{X_{\C^n}} \bb = \dress_{X}(\bb).
\ee
By using these, application of the definition \eqref{T27}  to the real and imaginary parts of the evaluation functions  gives
\be
\cL_{X_\cM} w^\alpha = \left((b_L B_{\alpha-1})^{-1} X b_L B_{\alpha-1}\right)_{\u(n)}\, w^\alpha,\quad
\cL_{X_\cM} B_\alpha= \dress_{ (b_L^{-1} X b_L)_{\u(n)}}(B_\alpha).
\label{H9}\ee
From the last two equalities, we obtain
\be
 \cL_{X_\cM} v(\alpha) = (b_L^{-1} X b_L)_{\u(n)}\, v(\alpha).
 \label{H10}\ee
In conclusion, we see that the vector field $X_\cM$ is encoded by the formula \eqref{H10} together with
\be
\cL_{X_\cM} g_R = [(b_L^{-1} X b_L)_{\u(n)}, g_R],
\qquad
\cL_{X_\cM} b_R = \dress_{(b_L^{-1} X b_L)_{\u(n)}} (b_R),
\label{H11}\ee
which follow from \eqref{H8} and the structure of $\Lambda$ \eqref{H1}. The completion of the
proof now requires checking
that the formula \eqref{H6} indeed gives a left-action of $\U(n)$ on $\cM$, whose infinitesimal
version reproduces the vector field $\cM$ found above. These last steps require some lines but are
fully straightforward, and thus we omit further details.
\end{proof}

\begin{remark} \label{R:act}
The action \eqref{H6} on $\cM$ is called (extended) quasi-adjoint action, since if we forget the
$v(\alpha)$ then it becomes the quasi-adjoint action on $M$ that goes back to \cite{Kli}.
At any fixed $(b_R, g_R)$, the map $\eta \mapsto \tilde \eta$ that appears in \eqref{H6} is a diffeomorphism
on $\U(n)$, and thus the quasi-adjoint action and the so-called \emph{obvious
action}  have the same orbits. The obvious action, denoted $A: \U(n) \times \cM \to \cM$,
operates as follows:
\be
 A_g(g_R,b_R,v) := (g g_R g^{-1}, \Dress_g(b_R), g v), \qquad \forall g\in \U(n),\, (g_R,b_R,v)\in \cM,
\label{obvA}\ee
where
\be
v:= (v(1), \dots, v(d)) \quad \hbox{and}\quad gv:= (g v(1),\dots, g v(d)).
\label{vdef}\ee
\end{remark}

We are interested in the reduction of $\cM$ defined by imposing the
moment map constraint
\be
\Lambda = e^\gamma \1_n \quad\hbox{with a fixed constant}\quad \gamma >0.
\label{momconst}\ee
The corresponding reduced phase space is
\be
\cM_{\red} = \Lambda^{-1}(e^\gamma \1_n)/\U(n).
\label{Mred}\ee
According to Remark \ref{R:act}, it does not matter whether we use the quasi-adjoint or the obvious
action for taking the quotient.

Denote $C^\infty(\cM)^{\U(n)}$ the $\U(n)$ invariant functions on $\cM$.
We may identify $C^\infty(\cM_\red)$ as the restriction of $C^\infty(\cM)^{\U(n)}$
to the `constraint surface' $\Lambda^{-1}(e^\gamma \1_n)$.
Then $C^\infty(\cM_\red)$ is naturally a Poisson algebra, with bracket denoted $\{\ ,\ \}_\red$.
This is obtained by using that the Poisson bracket of any two invariant functions is again invariant,
and its restriction to $\Lambda^{-1}(e^\gamma \1_n)$ depends only on the restrictions of the two
functions themselves.
One sees this relying on the  first class \cite{HT} character of the constraints
that appear in \eqref{momconst}.

Since the elements of $\fH$ \eqref{T34} are $\U(n)$ invariant,
they give rise to an Abelian Poisson subalgebra, $\fH_\red$, of $C^\infty(\cM_\red)$.
The flows of the elements of $\fH_\red$ on $\cM_\red$ result by projection of the free flows \eqref{T35}, see Section \ref{sec:N}.

\begin{remark} \label{Rem:EmbComplex}
 By using that \eqref{momconst} can be written equivalently
as  $\Lambda \Lambda^\dagger = e^{2 \gamma} \1_n$, it is not difficult to see that the triple $(g_R, b_R, v)$ belongs
to $\Lambda^{-1}(e^\gamma \1_n)$ if and only if it  satisfies
\be
e^{2\gamma} g_R^{-1} (b_R b_R^\dagger)  g_R - (b_R b_R^\dagger)  =
\sum_{\alpha=1}^d v(\alpha) v(\alpha)^\dagger.
\label{momconstdag}\ee
We notice that the set of the triples  $(g_R, b_Rb_R^\dagger, v)$ subject to \eqref{momconstdag} is a subset
of the set $\cM_{n,d,q}^\times$ defined in \cite{CF2}, which
contains the elements  $(X,Z, \cA_1,\ldots, \cA_d, \cB_1,\ldots, \cB_d)$ satisfying
\be
q^{-1} X Z X^{-1} - Z = \sum_{\alpha=1}^d \cA_\alpha \cB_\alpha
\label{CFconst}\ee
and the invertibility conditions
\be
\Bigl(Z + \sum_{\alpha=1}^k \cA_\alpha \cB_\alpha\Bigr)\in \GL, \quad \forall k=1,\dots, d,
\label{invcond}\ee
where $q$ is a nonzero complex constant, $X, Z \in \GL$, $\cA_\alpha \in \C^{n\times 1}$ and $\cB_\alpha\in \C^{1\times n}$,
for $\alpha=1,\dots, d$. (These are equations (4.3) and (4.4) in \cite{CF2}.)
It is known that if $q$ is not a root of unity, then the  action of $\GL$ on $\cM_{n,d,q}^\times$,  defined by
\be
g.(X,Z, \cA_\alpha, \cB_\alpha) := (g X g^{-1}, g Z g^{-1}, g \cA_\alpha, \cB_\alpha g^{-1}), \quad \forall g\in \GL,
\label{CFact}\ee
is  \emph{free}.
As explained in \cite{CF1,CF2}, this goes back to results
in representation theory \cite{CBS}.
Direct comparison of \eqref{momconstdag} and \eqref{CFconst}, and of the corresponding group actions,
shows that the $\U(n)$ action \eqref{obvA}  on our `constraint surface' $\Lambda^{-1}(e^\gamma \1_n)$ is  \emph{free}.
In our case the invertibility conditions \eqref{invcond} hold in consequence of the following identities that generalize \eqref{H7}:
\be
b_R b_R^\dagger + \sum_{\alpha=1}^k v(\alpha) v(\alpha)^\dagger  = \left(b_R \bb(w^1)\cdots \bb(w^k) \right)
\left(b_R \bb(w^1)\cdots \bb(w^k)\right)^\dagger.
\label{invid}\ee
Because $\U(n)$ acts freely on it,  \emph{ $\Lambda^{-1}(e^\gamma \1_n)$ is an embedded submanifold of $\cM$, and $\cM_\red$ \eqref{Mred}
is a smooth symplectic manifold, whose Poisson algebra coincides with
$\left(C^\infty(\cM_\red), \{\ ,\ \}_\red\right)$ presented in the preceding paragraph}.
Furthermore, since $(\cM, \Omega_{\cM})$ is actually a real-analytic symplectic manifold
and the formulae of the $\U(n)$ action and the moment map are all given by
real-analytic functions, $\cM_\red$ is also a \emph{real-analytic} symplectic manifold.
For the underlying general theory, the reader may consult \cite{OR}, and also appendix D in \cite{FG}.
\end{remark}

\begin{remark}
 We will eventually prove the degenerate integrability of the reduced system
by taking advantage of the following functions on $\cM$:
\begin{equation}
 I_{\alpha \beta}^k := \tr\left( v(\alpha)v(\beta)^\dagger L^k \right) = v(\beta)^\dagger L^k v(\alpha) \,, \quad 1\leq \alpha,\beta \leq d,\,\, k\geq 0\,,
\label{Eq:Int}\end{equation}
 where
\be
L := b_R b_R^\dagger.
\label{H17}\ee
 The identity \eqref{Dressconj} shows that $L$ transforms by conjugation,  and therefore
 these integrals of motion are invariant under the $\U(n)$ action \eqref{obvA} on $\cM$.
Their real and imaginary parts  descend to \emph{real-analytic} functions
on the reduced phase space.
\end{remark}

\subsection{Solution of the constraint  in terms of \texorpdfstring{$Q$}{Q} and dressed spins}
\label{ssec:Hbis}

Our fundamental task is to  describe  the set of $\U(n)$ orbits in
the `constraint surface' $\Lambda^{-1}(e^{\gamma} \1_n)$.
For this purpose, it will be convenient to label the points of $\cM$ by $g_R$, $L$  and
$v= (v(1), \dots, v(d))$ using that $L$ is given by \eqref{H17}.
In the various arguments we shall also employ alternative variables.

Since $g_R$ can be diagonalized by conjugation, we see from the form of the $\U(n)$ action \eqref{H6}
(or \eqref{obvA})
that every $\U(n)$ orbit lying in the constraint surface intersects
the set
\be
\cM_0 := \Lambda^{-1}( e^\gamma \1_n) \cap \Xi_R^{-1}(\T^n),
\label{H18}\ee
 where $\T^n$ is the subgroup of diagonal matrices in $\U(n)$.
Below,
\be
Q:= \diag(Q_1,\dots, Q_n)
\label{H19}\ee
stands for an element of  $\T^n$,
and  $\Ad_{Q^{-1}}$ denotes conjugation by $Q^{-1}$.
For any $\gamma \in \R^*$,
$\left(  e^{2\gamma} \Ad_{Q^{-1}} - \id \right)$ is an invertible linear operator on $\gl$,
 which
preserves the subspace of Hermitian matrices.
After this preparation, we present a useful characterization of  $\cM_0$.

\begin{proposition} \label{Pr:Lexp}
If $(Q,L,v) \in \cM_0$ \eqref{H18}, then $L$ can be expressed in terms of $Q$ and $v$ as follows:
\be
L = \left(  e^{2\gamma} \Ad_{Q^{-1}} - \id \right)^{-1} \left(\sum_{\alpha=1}^d v(\alpha) v(\alpha)^\dagger\right).
\label{H20}\ee
For the matrix elements of $L$, this gives
\be
L_{ij} = \frac{F_{ij}}{ e^{2\gamma} Q_j Q_i^{-1} - 1}
\quad\hbox{with}\quad
F:=\sum_{\alpha=1}^d v(\alpha) v(\alpha)^\dagger.
\label{H21}\ee
Conversely, if the Hermitian matrix $L$ given by the formula \eqref{H20} is positive definite,   then
$(Q,L,v)\in \cM_0$.
\end{proposition}

\begin{proof}
If $g_R= Q \in \T^n$, then  we have
\be
b_L Q^{-1} = Q^{-1} Q b_L Q^{-1} = Q^{-1} b_R^{-1},
\label{H22}\ee
showing that $b_L = Q^{-1} b_R^{-1} Q$.  Therefore, on $\cM_0$ the moment map $\Lambda$ \eqref{H1} reads
\be
\Lambda(Q, L, v) = Q^{-1} b_R^{-1} Q b_R \bb(w^1) \bb(w^2) \cdots \bb(w^d),
\label{H23}\ee
where $b_R$ and the $w^\alpha$ are viewed as functions of $L$ and $v$, given by the invertible relations of  equation
\eqref{H17} and Definition \ref{Def:31}.
We substitute this into the following equivalent form of the moment map constraint \eqref{momconst},
\be
\Lambda(Q, L, v) \Lambda(Q, L, v)^\dagger = e^{2\gamma} \1_n,
\label{H24}\ee
and thus obtain the  requirement
\be
(b_R b_1 b_2 \cdots b_d) ( b_Rb_1 b_2 \cdots b_d)^\dagger = e^{2\gamma} Q^{-1} L Q
\quad\hbox{with}\quad b_\alpha= \bb(w^\alpha).
\label{H25}\ee
By using  Lemma \ref{L:id1vv} and the definitions of $L$ and $F$,    this in turn is equivalent to
\be
 e^{2 \gamma} Q^{-1} L Q - L = F.
 \label{H26}\ee
 It follows that if $(Q,L,v)\in \cM_0$, then $L$ is given by the formula \eqref{H20}.

 To deal with the converse statement, notice that $L$ as given by the formula \eqref{H20}
 is Hermitian and automatically satisfies \eqref{H26}, but its positive definiteness is a non-trivial condition
 on the pair $(Q,v)$.
 Suppose that $L$ \eqref{H20} is positive definite. Then there exists a unique $b_R\in \B(n)$ for which
 $L = b_R b_R^\dagger$. Defining $v^\alpha := b_R^{-1} v(\alpha)$ (cf. \eqref{H3}), we can convert \eqref{H26} into
 \be
 e^{2\gamma} b_R^{-1} Q^{-1} b_R Q   (b_R^{-1} Q^{-1} b_R Q)^\dagger = \1_n + \sum_{\alpha=1}^d v^\alpha (v^\alpha)^\dagger.
 \label{H27}\ee
 Then there exists unique $w^1,\dots, w^d$ from $\C^n$ for which \eqref{H3} holds, and by \eqref{H2} and \eqref{H7}
 these variables satisfy
 \be
 \1_n + \sum_{\alpha=1}^d v^\alpha (v^\alpha)^\dagger = (\bb(w^1)\cdots \bb(w^d)) ( \bb(w^1) \cdots \bb(w^d))^\dagger.
 \label{H28}\ee
 Inserting into \eqref{H27}, and
 using \eqref{H23},
 we see that  \eqref{H27} implies the constraint equation \eqref{H24},
  whereby the proof is complete.
  \end{proof}

 We have the following consequence of Proposition \ref{Pr:Lexp}.

 \begin{corollary}
Let $L(Q,v)$ be given by \eqref{H20} and define the set
\be
\cP_0 := \{ (Q,v) \in \T^n \times \C^{n\times d}\mid L(Q,v) \,\hbox{is positive definite}\,\}.
\label{H29}\ee
The formula \eqref{H20} establishes a bijection between  $\cM_0$ \eqref{H18} and
$\cP_0$, which is an open subset of  $\T^n \times \C^{n\times d}$.
\end{corollary}

\medskip

Let us call $Q$  \emph{regular} if it belongs to
 \be
 \T^n_\reg:= \{ Q = \diag(Q_1,\dots, Q_n)\mid Q_i\in \U(1),\,  Q_i \neq Q_j \,\,\, \forall i\neq j\}.
\label{H30}\ee
Define
\be
\cM_0^\reg:= \{ (Q, L(Q,v), v)\in \cM_0 \mid Q \in \T^n_\reg \},
\label{H31}\ee
where $L(Q,v)$ is specified by \eqref{H20}.
Notice that any $g \in \U(n)$ for which $A_g$ \eqref{obvA} maps an element of $\cM_0^\reg$ to $\cM_0^\reg$ must
belong to the normalizer $\cN(n)$ of $\T^n$ inside $\U(n)$. Therefore, the quotient of the
regular part of the constraint surface by $\U(n)$, denoted $\cM^\reg_\red$,
can be identified as
\be
\cM^\reg_\red = \cM_0^\reg/\cN(n),
\label{H32}\ee
where
the quotient refers to the restriction of the obvious $\U(n)$ action \eqref{obvA} to the subgroup $\cN(n) < \U(n)$.

Let us call $Q\in \T^n$ \emph{admissible} if $Q\in \Xi_R^{-1}\left(\Lambda^{-1}(e^\gamma \1_n)\right)$.
In the next section,  we present an alternative procedure for solving the moment map constraint \eqref{momconst}, which will show
that all elements of $\T^n_\reg$ are admissible.

\begin{remark}
Equation \eqref{H26} implies $(e^{2\gamma} -  1) \tr(L) = \tr(F)$, and this can hold for a non-zero real $\gamma$ only
if $\gamma >0$ and $\tr(F)>0$, since $L$ must be positive definite and $\tr(F) \geq 0$ by the definition \eqref{H21}.
  This is why we assumed that $\gamma>0$.
It would be desirable to describe the elements of the set $\cP_0$ \eqref{H29} explicitly.
For $d=1$ the solution of this problem can be read off from \cite{FK2}.
On account of the next two observations, we expect that the structure of $\cP_0$ is  very different for $d<n$ and for $d\geq n$.
First, let us notice that $Q=\1_n$ is not admissible if $d<n$, since in this case the rank of $L$ given by the formula
\eqref{H20} is at most $d$, while the rank of  any positive definite $L$ is $n$.
Second, note that if $d\geq n$, then we can arrange to have $F=\1_n$ by suitable choice of $v$. Let $v_0$ be such a choice.
Then $L(\1_n, v_0)$ is a positive multiple of $\1_n$, and
therefore there is an open neighbourhood of $(\1_n, v_0)$ in $\T^n \times \C^{n\times d}$ that belongs to $\cP_0$.
\end{remark}

\subsection{Solution of the constraint  in terms of \texorpdfstring{$Q, p$}{Q,p} and primary spins}
\label{ssec:prim}

Now we return to using the variables
$g_R, b_R$ and $W$ \eqref{Wdef} for labeling the points of $\cM$.
We can uniquely decompose every element $b\in \B(n)$ as the product of a diagonal matrix, $b_0$, and
an upper triangular matrix, $b_+$, with unit diagonal. Applying this to $b= b_R$ we write
\be
b_R= b_0 b_+,
\label{H34}\ee
and introduce also
\be
S(W):= \bb(w^1) \bb(w^2) \cdots \bb(w^d)=: S_0(W) S_+(W).
\label{H35}\ee
The moment map constraint on $\cM_0$ \eqref{H18} reads
\be
\Lambda(Q,b_R, W) =  Q^{-1} b_R^{-1} Q b_R S(W) =  Q^{-1} b_+^{-1} Q b_+ S_0(W) S_+(W) = e^\gamma \1_n.
\label{H36}\ee
Since $b_0$ drops out from the formula of $\Lambda$, it is left arbitrary, and we parametrize it as
\be
b_0 = e^p
\quad\hbox{with}\quad p= \diag(p_1,\ldots, p_n), \quad p_i \in \R.
\label{H37}\ee
A crucial observation is that \eqref{H36} can be separated according to the diagonal
and strictly upper-triangular parts, since it is equivalent to the two requirements
\be
 S_0(W) = e^\gamma \1_n
\label{H38} \ee
 and
 \be
   Q^{-1} b_+^{-1} Q b_+  S_+(W) = \1_n.
   \label{H39}\ee
  The constraint \eqref{H38} is responsible for a reduction of the primary spin variables.
 Next we make a little detour and
  present a general analysis of such reductions.

Let us introduce the map $\phi: \C^{n\times d} \to \b(n)_0$ by writing
\be
S_0(W):= \exp(\phi(W)),
\label{H47}\ee
and notice from Remark \ref{Rem:Tact} that $\phi$  is the moment map for the ordinary Hamiltonian action
of $\T^n$ on the symplectic manifold $(\C^{n\times d}, \Omega_\cW)$ of the primary spins.
Here, the dual of the Lie algebra of the
torus $\T^n< \U(n)$ is identified with the space $\b(n)_0$ of real diagonal matrices.
The torus action in question is given by
\be
\tau \cdot (w^1,\dots, w^d) = (\tau w^1,\dots, \tau w^d), \qquad \forall \tau \in \T^n.
\label{H48}\ee
 Taking any  moment map value from the range of $\phi$,
\be
\Gamma:= \diag(\gamma_1,\dots, \gamma_n),
\label{H49}\ee
we define the reduced space of primary spins:
\be
\C^{n\times d}_\red(\Gamma):= \phi^{-1}(\Gamma)/\T^n.
\label{H50}\ee

\begin{proposition} \label{Prop:proper}
The moment map $\phi: \C^{n\times d} \to \b(n)_0$ defined by \eqref{H47} with \eqref{H35}
is proper, i.e., the inverse image of any
compact set is compact. Fixing any moment map value $\Gamma$ for which $\gamma_j>0$ for all $j$,
the reduced spin-space \eqref{H50} is a smooth, compact and connected symplectic manifold of dimension
$2n(d-1)$.
\end{proposition}
\begin{proof}
We first prove that the map $\phi$ is proper.
Since the compact sets of Euclidean  spaces are the  bounded and closed sets,
and since $\phi$ is continuous,
it is enough to show that the inverse image of any bounded subset of $\b(n)_0\simeq \R^n$ is a bounded
subset of $\C^{n\times d}$.
Due to the definition of $\phi$ and equation \eqref{bb},
the formula of $\phi=\diag(\phi_1,\dots, \phi_n)$ is determined by
the equality
\be
\exp(2 \phi_j(W)) = \prod_{\alpha=1}^d \frac{\cG_j(w^\alpha)}{\cG_{j+1}(w^\alpha)}
= \prod_{\alpha=1}^d \left[ 1 + \frac{|w_j^\alpha|^2}{\cG_{j+1}(w^\alpha)}\right].
\label{H51}\ee
The second equality shows that $\phi_j(W)\geq 0$. On the other hand,
using the first equality and that $\cG_{n+1}=1$, we see that
\be
W\in \phi^{-1}(\diag(\gamma_1,\dots, \gamma_n))
\label{H52}\ee
if and only if
\be
\prod_{\alpha=1}^d \cG_j(w^\alpha) = \exp( 2\sum_{k=j}^n \gamma_k),
\qquad \forall j=1,\dots, n.
\label{H53}\ee
Now, if $\diag(\gamma_1,\dots, \gamma_n)$ is from a bounded set, then
$ \sum_{k=1}^n \gamma_k \leq C$
with some constant $C$.
By using this and the $j=1$ special case of \eqref{H53}, we obtain
\be
1 + \sum_{\alpha=1}^d |w^\alpha|^2 \leq \prod_{\alpha=1}^d (1 + |w^\alpha|^2)=
\prod_{\alpha=1}^d \cG_1(w^\alpha)
 \leq e^{2C},
\label{H54}\ee
which implies  that the inverse image of any bounded set is bounded.

If $\gamma_j >0$ for all $j$, then we see from the formula \eqref{H51}
that for any $W\in \Phi^{-1}(\Gamma)$ and for each $1\leq j \leq n$ there
must exist an index $1\leq \alpha(j) \leq d$ such that
\be
w_j^{\alpha(j)}\neq 0.
\ee
This implies immediately that the action \eqref{H48} of $\T^n$ is \emph{free} on $\phi^{-1}(\Gamma)$,
and therefore $\C^{n\times d}_\red(\Gamma)$ \eqref{H50} is a smooth symplectic manifold
of dimension $2n(d-1)$.
Since the moment map $\phi$ is proper, all its fibers $\phi^{-1}(\Gamma)$ are compact,
and by Theorem 4.1 in \cite{Hi} they are also connected. Hence $\C^{n\times d}_\red(\Gamma)$
 is also compact and connected.
\end{proof}

\medskip
\noindent
\begin{remark}
If $\gamma_j >0$ for all $j$, then $\C^{n\times d}_\red(\Gamma)$ is actually a \emph{real-analytic} symplectic manifold.
To cover  it with charts,
for any map $\mu: \{1,\dots, n\} \to \{ 1,\dots, d\}$  we introduce the set
\be
X(\mu):= \{ W \in \phi^{-1}(\Gamma)\mid w_j^{\mu(j)} \neq 0,\,\forall j\}.
\ee
Then the reduced spin-space is the union of the open subsets
$Y(\mu):= X(\mu)/\T^n$, and
a model of $Y(\mu)$ is provided by
\be
Z(\mu):= \{ W \in \phi^{-1}(\Gamma)\mid w_j^{\mu(j)} >0,\,  \forall j\}.
\ee
One can specify coordinates on $Z(\mu)$  by solving the constraints \eqref{H53}
for the $w_j^{\mu(j)}$ in terms of the remaining free variables, the $w_j^\alpha$ with $\alpha \neq \mu(j)$,
which take their values in a certain open subset of $\C^{n(d-1)}$.
It is an interesting exercise to fill out the details, and to also write down the
reduced symplectic form by using these charts.

If $d=1$,  then the reduced spin-space consists of a single point. This is also true
in the trivial case for which $\gamma_j = 0$ for all $j$.
If some of the $\gamma_j$ are zero and the others are positive, then
the moment map constraint $\phi(W)=\Gamma$ leads to a stratified symplectic space.
Finally, note that for the case corresponding to equation \eqref{H38} $\gamma_j = \gamma >0$ for all $j$.
\end{remark}

Now returning to our main problem,
 it is useful to recast \eqref{H39} in the form
 \be
 b_+ S_+(W) = Q^{-1} b_+ Q.
 \label{H40}\ee
 By using the principal gradation of $n\times n$ matrices,
 this equation can be solved recursively for $b_+$ if $S_+(W)$ and $Q$ are given, with $Q$ regular.
 In fact, the following lemma is obtained by a word-by-word application
 of the arguments  of Section 5 in \cite{F1}; hence we omit the proof.

 \begin{lemma}\label{Lm:bexplicit}
Suppose that $S_+=S_+(W)$ and $Q$ are given, with $Q\in \T^n_\reg$. Then equation \eqref{H40} admits
 a unique solution for $b_+$, denoted $b_+(Q,W)$. Using the notation
 \be
\cI^{a, a+ j} = \frac{1}{Q_{a+j} Q_a^{-1} - 1}, \quad a=1,\dots, n-1,
\ee
 and placing the matrix indices in the upstairs position, we have
\be
b_+^{a,a+1} = \cI^{a, a+1} S_+^{a,a+1},
\ee
and for  $k=2,\dots, n-a$ we have
\be
b_+^{a, a+k} = \cI^{a, a+k} S_+^{a, a+k} +
\sum_{\substack{m=2,\dots, k\\ (i_1,\dots, i_m)\in \N^m\\ i_1 + \dots + i_m = k} }
\prod_{\alpha=1}^m
\cI^{a, a + i_1 +\dots + i_\alpha} S_+^{a + i_1 + \dots + i_{\alpha -1}, a+ i_1 + \dots + i_\alpha}.
\label{explicit}\ee
  \end{lemma}

Now we restrict ourselves to the regular part of $\cM_0$, stressing that it
is defined without reference to any particular parametrization:
\be
\cM_0^\reg \equiv \Lambda^{-1}( e^\gamma \1_n) \cap \Xi_R^{-1}(\T^n_\reg).
\label{H41}\ee
Any gauge transformation that maps an  element of
$\cM_0^\reg$ to $\cM_0^\reg$ is given by the obvious action \eqref{obvA} of the normalizer
$\cN(n)$ of $\T^n$ inside $\U(n)$.
The normalizer has the normal subgroup $\T^n$, and
the corresponding factor group is the permutation group
\be
S_n = \cN(n)/\T^n.
\ee
Consequently, we have
\be
\cM_\red^\reg = \cM_0^\reg/\cN(n) = (\cM_0^\reg/\T^n)/S_n.
\label{stages1}\ee
It is plain that $\cM_\red^\reg$ is a dense, open subset of the reduced phase space,
and the above consecutive quotients show that $\cM_0^\reg/\T^n$
is an $S_n$ covering space\footnote{More precisely,
 $\cM_0^\reg/\T^n$ is a principal fiber bundle
 with structure group $S_n$ over the base $\cM_\red^\reg$.}
of this dense open subset.

\begin{theorem}\label{Th:redsymp}
By  solving the moment map constraint for $b_R$ in the form
$b_R= e^p b_+(Q,W)$ as explained above,  the manifold $\cM_0^\reg$ \eqref{H41} can
be identified with the model space
\be
 \tilde \cP_0^\reg:=\T^n_\reg \times \b(n)_0 \times \phi^{-1}(\Gamma)=
 \{ (Q, p, W)\mid Q\in \T^n_\reg,\,\, p\in \b(n)_0,\,\,
 W\in \phi^{-1}(\Gamma) \},
 \ee
 where $\Gamma = \gamma \1_n$.
 Utilizing this model, the covering space
$\cM_0^\reg/\T^n$ of the regular part of the reduced phase space becomes identified with
the symplectic manifold
\be
T^*\T^n_\reg \times \C^{n\times d}_\red(\Gamma)
\ee
equipped with its natural product symplectic structure.
\end{theorem}
\begin{proof}
The restriction of the action \eqref{obvA} to $\T^n$ translates into the action
\be
A_\tau (Q,p, W) = (Q,p,\tau \cdot W), \qquad \tau \in \T^n,
\ee
on the model space  $\tilde \cP_0^\reg$, from which we obtain the
identification $\cM_0^\reg/\T^n \simeq T^* \T^n_\reg \times \C^{n\times d}_\red(\Gamma)$
at the level of manifolds.
Let $\xi_1: \tilde \cP_0^\reg \to \cM $ and $\xi_2: \phi^{-1}(\Gamma) \to \C^{n\times d}$
denote the natural inclusions, and write $Q_j = e^{\ri q_j}$.  Then a simple calculation gives
\be
\xi_1^* (\Omega_\cM) = \sum_{j=1}^n dp_j \wedge dq_j + \xi_2^* (\Omega_\cW),
\label{xisymp}\ee
which proves the claimed identification at the level of \emph{symplectic} manifolds.
\end{proof}

\begin{corollary} \label{Cor:Dense}
The dense open submanifold $\cM_\red^\reg\subseteq \cM_\red$ is connected, and consequently
$\cM_\red$ is also connected.
\end{corollary}
\begin{proof}
Since $\phi^{-1}(\Gamma)$ is connected by Proposition \ref{Prop:proper}, the connected components of
$\tilde \cP_0^\reg$ correspond to the connected components of $\T^n_\reg$.
It is well-known  (see e.g. the appendix in \cite{FA}) that any two connected components of
$\T^n_\reg$ are related by permutations. Thus $\cM_\red^\reg$
is the continuous image of a single connected component of $\tilde \cP_0^\reg$, implying
 its connectedness.
 The proof is finished by recalling that if a dense open subset of a topological
space is connected, then the space itself is connected.
\end{proof}

\begin{remark}\label{Rem:scaling}
In this long remark
we explain how the (trigonometric real form of the) spin Sutherland model of Gibbons and Hermsen \cite{GH}
can be obtained from our construction via a scaling limit. For this, we introduce a positive parameter $\epsilon$ and
replace the variables $p$ by $\epsilon\, p$ and $W$ by $\epsilon^{\frac{1}{2}} W$, while keeping $Q$ unchanged.
The formulae \eqref{bb} imply
\be
\bb( \epsilon^{\frac{1}{2}} w)_{kk}= 1 + \frac{1}{2} \epsilon |w_k|^2 + \mathrm{o}(\epsilon),\,\, \forall k
\quad \hbox{and}\quad
\bb(\epsilon^{\frac{1}{2}} w)_{ij} = \epsilon w_i \overline{w}_j + \mathrm{o}(\epsilon),
\,\,
\forall i<j.
\ee
By using this we see that the matrix $b=b_+$ in \eqref{H34}, given in explicit form by Lemma \ref{Lm:bexplicit},  has the expansion
\be
b(Q, \epsilon^{\frac{1}{2}} W)_{ij} = \epsilon (Q_j Q_i^{-1} -1)^{-1} \sum_{\alpha=1}^d w_i^\alpha \overline{w}_j^\alpha +  \mathrm{o}(\epsilon),
\quad \forall i<j.
\ee
Then, for $L= b_R b_R^\dagger$ with $b_R= \exp({\epsilon p}) b(Q, \epsilon^{\frac{1}{2}} W)$, we find
\be
\tr(L^{\pm 1}) = n \pm 2 \epsilon\, \tr (p) + 2 \epsilon^2 \tr ( p^2) +
\epsilon^2 \sum_{i<j} \frac{ \vert ( w_i^\bullet, w_j^\bullet)\vert^2 }{ \vert Q_j Q_i^{-1} - 1\vert^2}  + \mathrm{o}(\epsilon^2),
\label{Lexpand}\ee
where  $w_i^{\bullet}\in \C^d$ with components $w_i^\alpha$, and
$(w_i^\bullet, w_j^\bullet):= \sum_{\alpha=1}^d  w_i^\alpha \overline{w}_j^\alpha$.
Therefore, we obtain
\be
\lim_{\epsilon \to 0} \frac{1}{8 \epsilon^2}(\tr(L) + \tr(L^{-1}) -2n)
= \frac{1}{2}  \tr ( p^2) + \frac{1}{32}
 \sum_{i\neq j} \frac{\vert ( w_i^\bullet, w_j^\bullet)\vert^2 }{ \sin^2\frac{q_i-q_j}{2}},
\ee
which is just the standard Hamiltonian of the (real, trigonometric) Gibbons--Hermsen model.

Replacing $\gamma$ by $=\epsilon \gamma$ and taking the limit,  the
residual constraint \eqref{H38}  gives
$ (w_j^\bullet, w_j^\bullet) = 2 \gamma$.
Then, rescaling not only the variables but also the symplectic form \eqref{xisymp}, one gets
\be
\lim_{\epsilon\to 0} \epsilon^{-1} \left(\xi_1^*\Omega_{\cM}\right) =  \sum_{j=1}^n dp_j \wedge dq_j +
\frac{\ri}{2} \sum_{j=1}^n \sum_{\alpha=1}^d d w_j^\alpha \wedge d\overline{w}_j^\alpha,
\ee
which reproduces the symplectic form of the Gibbons--Hermsen model.

It is known \cite{FK2} that
 the standard spinless RS  Hamiltonian \cite{RS} can be derived as the reduction of  $\tr(L) + \tr(L^{-1})$  in the $d=1$ case.
 For $d\geq 2$,  we shall
 show (see Corollary \ref{Cor:vf} and Corollary \ref{Cor:vfredPB})  that  the Hamiltonian of the
 (real, trigonometric) spin RS model of Krichever of Zabrodin \cite{KZ} is the reduction of $\tr(L)$.
 As was already discussed in the Introduction,  the term \emph{chiral} spin RS model could have been a more fitting name for
 the model of \cite{KZ}, but we follow the literature in dropping `chiral' in this context.
\end{remark}

\medskip

In this subsection, we have derived an almost complete description of the
reduced system. We have established that $T^* \T^n_\reg \times \C^{n\times d}_\red(\Gamma)$
is an $S_n$ covering space of a dense, open subset of the reduced phase space,
and we can write down the Hamiltonians $\tr(L^k)$ by using the explicit formula $b_R = e^p b_+(Q,W)$.
Why is the paper not finished at this stage?
Well, one reason is that although $Q,p$ and $W$ are very nice variables for presenting
the reduced symplectic form,
 they are not the ones
that feature in the the Krichever--Zabrodin equations of motion, which we
wish to reproduce
in our setting.
In fact, the usage of the dressed spins $v(\alpha)$ \eqref{H3} will turn out indispensable for this purpose.
(Notationwise, we took this into account already in equation \eqref{I8}.)
Another, closely related, reason is that the action of permutations is practically intractable
in terms of the primary spins\footnote{This difficulty evaporates in the scaling limit discussed
in Remark \ref{Rem:scaling}.}. More precisely, the action on the components of $Q$ is the obvious one,
but on $p$ and $W$ it is known only in an implicit manner, via the realization
of these variables as functions of $L=b_R b_R^\dagger $ and the dressed spins, on which the full $\U(n)$ action,
and thus also the permutation action, is governed by the simple formula \eqref{obvA}.
In short, both the formula $b_+(Q,W)$ and the change of variables from $Q,p,W$
to $Q$ and dressed spins $v(\alpha)$  are complicated,
and for some purposes the latter will prove to be more convenient variables.

\section{The reduced equations of motion} \label{sec:N}

Below, we first present a characterization of the projection of the Hamiltonian vector fields
of arbitrary elements of $\fH$ \eqref{T34} to the dense open submanifold $\cM^\reg_\red$ \eqref{H32} of
 the  reduced phase space
$\cM_\red$ \eqref{Mred}.
Then we reproduce the trigonometric real form  of the equations of motion \eqref{I2} of \cite{KZ}
as the simplest special case of the reduced dynamics.

As in \S\ref{ssec:Hbis}, we parametrize the points of $\cM$ by $g_R$,  $v$ and
$L = b_R b_R^\dagger$.
For any $h\in C^\infty(\B(n))^{\U(n)}$ we  put
\be
\cV(L):= D h(b_R).
\label{N1}\ee
Denoting the Hamiltonian vector field of $H= \Lambda_R^*(h)$
by $X_H$ and viewing $g_R$, $v(\alpha)$ and $L$ as evaluation functions,
in correspondence to \eqref{T35}, we have
\be
X_H[g_R] = \cV(L) g_R,\quad
X_H[v(\alpha)]=0,
\quad X_H[L]=0.
\label{N2}\ee
It is clear that $X_H$ admits a well-defined projection on $\cM_\red$, which encodes the reduced dynamics.
Of course, one may add any infinitesimal gauge transformation to the vector field $X_H$ without modifying
its projection on $\cM_\red$, i.e., instead of $X_H$ one may equally well consider any $Y_H$ of the form
\bea
&&Y_H[g_R] = \cV(L) g_R + [Z(g_R, L, v), g_R],\nonumber\\
&& Y_H[v(\alpha)] = Z(g_R, L, v) v(\alpha),
\label{N3} \\
&& Y_H[L] = [Z(g_R, L, v), L],\nonumber
\eea
with arbitrary $Z(g_R, L, v)\in \u(n)$.
It is also clear that
one may use the restriction of $Y_H$ to  $\cM_0$ \eqref{H18} for determining the projection,
and $Z$ can be chosen in such a manner to guarantee the tangency of the restricted vector field to $\cM_0$.

Let us consider the vector space decomposition
\be
\u(n)= \u(n)_0 + \u(n)_\perp,
\label{N4}\ee
where $\u(n)_0$ and $\u(n)_\perp$ consist of diagonal and off-diagonal matrices, respectively.
Accordingly, for any $T \in \u(n)$ we have
\be
T = T_0 + T_\perp,
\qquad T= \u(n)_0,\,\, T_\perp \in \u(n)_\perp.
\label{N5}\ee
Using $\Ad_Q(T)= Q T Q^{-1}$,  the restriction of the
operator $(\Ad_Q - \id)$ to $\u(n)_\perp$ is invertible, and we define
\be
K(Q,L):= \left((\Ad_Q - \id)\vert_{\u(n)_\perp}\right)^{-1} \cV(L)_\perp.
\label{N6}\ee
More explicitly, setting $Q = \exp(\ri q)$, we have $K_{kk}=0$ and
\be
K(Q,L)_{kl} = - \frac{1}{2} \cV(L)_{kl} - \frac{\ri}{2} \cV(L)_{kl} \,\cot\left(\frac{q_k - q_l}{2}\right) ,
\quad \forall k\neq l.
\label{N7}\ee

\begin{proposition}
For any $H\in \fH$ \eqref{T34}, applying the previous notations,
the following formulae yield a  vector field
 $Y_H^0$ on $\cM_0^\reg$ \eqref{H31},
\bea\label{YH0}
&&Y_H^0[Q] = \cV(L)_0 Q ,\nonumber\\
&& Y_H^0[v(\alpha)] =\bigl( K(Q,L) + \cZ(Q, L, v)\bigr) v(\alpha),  \label{N8} \\
&& Y_H^0[L] = [K(Q,L)  + \cZ(Q, L, v), L], \nonumber
\eea
for any $\cZ(Q, L, v)\in \u(n)_0$ \eqref{N4}, with $L=L(Q,v)$ given by \eqref{H20}.
The vector field $Y_H^0$ admits a well-defined projection
on $\cM_\red^\reg \subset \cM_\red$ \eqref{H32}, which coincides with the corresponding restriction of the
projection of the Hamiltonian vector
 field $X_H$ \eqref{N2}
to $\cM_\red$.
\end{proposition}

\begin{proof}
As a consequence of $(\Ad_Q -\id) K(Q,L) = \cV(L)_\perp$, we have the identity
\be
\cV(L) Q + [ K(Q,L) ,Q] = \cV(L)_0 Q.
\label{N9}\ee
This shows that $Y_H^0$ is obtained by restricting $Y_H$ \eqref{N3} to $\cM_0^\reg$, where
\be
Z(Q,L,v) = K(Q,L) + \cZ(Q,L,v).
\label{N10}\ee
This choice of $Z$ guarantees that the restricted vector field is tangent to $\cM_0^\reg$.
The fact that $\cZ$ is left undetermined reflects
the residual $\cN(n)$  gauge transformations acting on $\cM_0^\reg$.
\end{proof}

\begin{remark}\label{Rem:sol}
 Only the first two relations in \eqref{N8} are essential, since
the third one follows from them
via the formula \eqref{H20} of $L=L(Q,v)$.
Now take an initial value $(Q^0, L^0, v^0)\in \cM_0^\reg$ and $\epsilon>0$ ($\epsilon=\infty$ is allowed) such that
\be
g_R(t) = \exp(t \cV(L^0)) Q^0 \in \U(n)_\reg\quad\hbox{for}\quad  -\epsilon <t  < \epsilon,
\ee
where the elements of $\U(n)_\reg$ have $n$ distinct eigenvalues.
Notice from \eqref{T36}  that $g_R(t)$  describes the unreduced solution curve, and that a small enough
$\epsilon$  will certainly do.
Then, for $- \epsilon < t < \epsilon$ there exists  a unique smooth curve $\eta(t) \in \U(n)$ for which
\be
Q(t):= \eta(t) g_R(t) \eta(t)^{-1} \in \T^n_\reg
\quad \hbox{and}\quad \eta(0)=\1_n,\quad \left(\dot{\eta}(t) \eta(t)^{-1}\right)_0 =0.
\label{Qtsol}\ee
It is easy to see that $(Q(t), L(t), v(\alpha)(t))$ given by the above $Q(t)$ and
\be
L(t) = \eta(t) L^0 \eta(t)^{-1},
\quad
v(\alpha)(t) = \eta(t) v(\alpha)^0
\ee
yields the integral curve of the vector field \eqref{YH0} with $\cZ=0$.  We here used
the property $\cV(g L g^{-1}) = g \cV(L) g^{-1}$
($\forall g\in \U(n), L\in \fP(n)$), which  follows from the definition \eqref{N1}.
The auxiliary conditions imposed in \eqref{Qtsol}   fix the ambiguity of the `diagonalizer' $\eta(t)$ of $g_R(t)$.
The reduction approach leads to this solution algorithm naturally, but we should stress that an analogous algorithm
  was found long ago by Ragnisco and Suris  \cite{RaSu} using a direct method.
\end{remark}

\begin{corollary} \label{Cor:vf}
Consider $H\in \fH$ defined by
\be
H= \Lambda_R^*(h)
 \quad\hbox{with}\quad
 h(b):= (e^{2\gamma} -1) \tr (bb^\dagger).
 \label{goodH}\ee
Then the evolution equation on $\cM_0^\reg$ corresponding to the
vector field $Y_H^0$ \eqref{N8} with $\cZ=0$
can be written explicitly as follows:
\bea
&&\frac{1}{2} \dot{q}_j := \frac{1}{2\ri} Y_H^0[Q_j] Q_j^{-1}= F_{jj},
\label{N11} \\
&&\dot{v}(\alpha)_i := Y_H^0[v(\alpha)_i]=-  \sum_{j\neq i} F_{ij} v(\alpha)_j
V\left(\frac{q_j -q_i}{2}\right),
\label{N12}\eea
where $F = \sum_\alpha v(\alpha) v(\alpha)^\dagger$ and the `potential function' $V$ reads
\be
V(x)= \cot x - \cot(x-\ri \gamma).
\label{N13}\ee
These formulae reproduce the spin RS equations of motion \eqref{I2} by
 setting $x_i = q_i/2$ and imposing the  additional reality condition \eqref{I4}.
\end{corollary}

\begin{proof}
In this case the definition \eqref{N1} gives
\be
\cV(L) = 2\ri (e^{2\gamma}-1) L.
\label{N14}\ee
Equation \eqref{N11} follows immediately from \eqref{N8} since $\cV(L)_{jj} = 2 \ri  F_{jj}$ by
 \eqref{H20} and we have $Q_j = \exp(\ri q_j)$.
Taking advantage of Hermite's cotangent identity,
\be
\cot(z-a_1) \cot(z-a_2) = -1 + \cot(a_1-a_2)\cot(z-a_1) + \cot(a_2-a_1) \cot(z-a_2),
\label{N15}\ee
it is not difficult to re-cast
 the off-diagonal matrix function $K(Q, L(Q,v))$ \eqref{N7} in the form
 \be
K_{kl} = F_{kl} \left[
\cot\left(\frac{q_k- q_l}{2}\right) - \cot\left(\frac{q_k- q_l}{2} + \ri \gamma\right)\right],
\label{N16}\ee
for all $ k\neq l$. This gives \eqref{N12} with \eqref{N13}. The validity of the last sentence
of the corollary can also be checked directly.
\end{proof}

\begin{remark}
The restriction of the Hamiltonian $H$ \eqref{goodH} to $\cM_0$ gives
\be
H(Q,v) = (e^{2\gamma} -1)  \tr( L(Q,v)) =  \sum_{j=1}^n F_{jj}.
\label{N17}\ee
We shall confirm in \S\ref{ssec:RP2} that the ensuing reduced Hamiltonian generates
the equations of motion \eqref{I8} via the reduced Poisson structure described
in coordinates using the gauge fixing condition \eqref{I7}.
\end{remark}

\section{The reduced Poisson structure} \label{sec:RP}

The main purpose of this section is to present the explicit form of the reduced Poisson structure in terms of the variables that feature
in the equations of motion \eqref{I8}. The first subsection contains a couple of auxiliary lemmae, in which we provide
explicit formulae for the Poisson brackets of the half-dressed and dressed spins, and the matrix entries of $g_R$ and $L$.
These  permit us to establish that the $\U(n)$ invariant integrals of motion \eqref{Eq:Int} form a closed polynomial Poisson algebra on
the unreduced phase space, which automatically descends to the reduced phase space. This interesting algebra is given by Proposition \ref{P:BrII}.
In the second subsection we utilize the Poisson brackets of another set of $\U(n)$ invariant functions in order to characterize the reduced Poisson structure.
We shall rely on the fact that the restriction of the Poisson brackets of $\U(n)$ invariant functions to a gauge slice in the `constraint surface'
must coincide with the Poisson brackets of the restricted functions calculated from the reduced Poisson structure.

All calculations required by this section are straightforward, but they are quite voluminous and not enlightening.
We strive to give just enough details to provide the gist of
these calculations, and so that an interested reader may reproduce them.
Some of these details are relegated to Appendix \ref{A:Glob2} and Appendix \ref{A:RedPr}.

\subsection{Some Poisson brackets before reduction} \label{ssec:RP1}

Using the results from Section \ref{sec:T}, the Poisson structure $\brM{\ , \ }$ on $\cM$ can be described in terms of the
(complex-valued) functions returning the entries of the matrices $(g_R,b_R,w^1,\ldots,w^d)$ and their complex
conjugates. Namely, we can use \eqref{T7}--\eqref{T8} with $K=g_R$ or $K=b_R$, then \eqref{T18}--\eqref{T19} to characterize
the Poisson structure restricted to functions on the Heisenberg double; for fixed $\alpha=1,\ldots,d$, the Poisson
brackets involving $w^\alpha$ are given by \eqref{T29}--\eqref{T30}.
The Poisson brackets between functions of $w^\alpha$ and functions of $g_R$ and $b_R$ vanish.
Our aim is to translate these relations to the matrices $ (g_R, L=b_R b_R^\dagger, v(1)=b_R v^1,\dots, v(d)=b_R v^d)$, which are more
convenient to understand the reduced phase space $\cM_{\red}$, see Section \ref{sec:H}. As a first step, we express the
 Poisson structure on the half-dressed spins
$v^\alpha=b_1\cdots b_{\alpha-1}w^\alpha$ defined in \eqref{H3}. We let $\br{\ , \ }:=\brM{\ , \ }$ for the rest of the section.

\begin{lemma} \label{L:halfv}
The Poisson brackets of the half-dressed spins are given by the following formulae
\bea
  &&\br{v^\alpha_i,v^\beta_k} =-\ic \, \sgn(k-i) v_k^\alpha v_i^\beta +\ic \, \sgn(\beta - \alpha) v_k^\alpha v_i^\beta\,,
  \label{Eq:Pvvhalf}\\
&&\br{v^\alpha_i,\bar v^\beta_k} =
\ic \delta_{ik} v_i^\alpha \bar{v}_k^\beta+2\ic \delta_{ik}\sum_{r>k}v_r^\alpha \bar{v}_r^\beta
 +\ic \delta_{\alpha \beta} v^\alpha_i\bar v^\beta_k + 2 \ic \delta_{\alpha \beta}\sum_{\mu<\alpha} v_i^\mu \bar{v}_k^\mu + 2 \ic \delta_{ik}\delta_{\alpha \beta}  \,,
 \label{Eq:Pvbarvhalf}
 \eea
where $1\leq i,k \leq n$ and $1\leq \alpha,\beta\leq d$. In particular, this defines a Poisson structure on  $\C^{nd} \simeq \R^{2nd}$.
\end{lemma}
This result is proved in Appendix \ref{A:Glob2}.
From the reality of the Poisson bracket, we  have
\be
\br{\bar v^\alpha_i,\bar v^\beta_k}=+\ic \, \sgn(k-i) \bar v_k^\alpha \bar v_i^\beta -\ic \, \sgn(\beta - \alpha) \bar v_k^\alpha \bar v_i^\beta \,. \label{Eq:Pbarvbarvhalf}
\ee

\begin{remark}
 If we complexify the formulae of Lemma \ref{L:halfv} by introducing $a_{i\alpha}:=(v^\alpha_i)^{\C}$ and $b_{\alpha i}:=(\bar v^\alpha_i)^{\C}$,  we get a complex holomorphic Poisson structure on $\C^{2nd}$ given by
 \bea
  &&\br{a_{i\alpha },a_{k \beta}} =-\ic \, \sgn(k-i) a_{k\alpha} a_{i\beta} +\ic \, \sgn(\beta - \alpha) a_{k\alpha} a_{i\beta}\,,
  \label{Eq:AOp1}\\
    &&\br{b_{\alpha i},b_{\beta k}} =+\ic \, \sgn(k-i) b_{\alpha k} b_{\beta i} -\ic \, \sgn(\beta - \alpha) b_{\alpha k} b_{\beta i}\,,
  \label{Eq:AOp2}\\
&&\br{a_{i\alpha },b_{\beta k}} =
\ic \delta_{ik} a_{i\alpha} b_{\beta k} +2\ic \delta_{ik}\sum_{r>k} a_{r\alpha} b_{\beta r}
 +\ic \delta_{\alpha \beta} a_{i\alpha} b_{\beta k} + 2 \ic \delta_{\alpha \beta}\sum_{\mu<\alpha} a_{i\mu} b_{\mu k} + 2 \ic \delta_{ik}\delta_{\alpha \beta} .\quad
 \label{Eq:AOp3}
 \eea
 After appropriate rescaling, this reproduces
the \emph{minus} Poisson bracket introduced by
 Arutyunov and Olivucci in their treatment of the complex holomorphic  spin RS system by Hamiltonian reduction \cite{AO}.
Considering the analogous construction with the variables $v_{+,i}^\alpha:=v_i^{d-\alpha+1}$ instead,
 we obtain the \emph{plus} Poisson bracket introduced in \cite{AO}.
\end{remark}

From now on, we let $b=b_R,g=g_R$.
Using Lemma \ref{L:halfv} and the Poisson structure of the Heisenberg double, we can easily write the Poisson brackets involving the entries $v(\alpha)_i$ of the dressed spins $v(\alpha)=b_R v^\alpha$.

\begin{lemma} \label{L:Dressv}
The Poisson brackets of the dressed spins are given by the following formulae
\bea
&&\br{v(\alpha)_i,v(\beta)_k} =-\ic \, \sgn(k-i) v(\alpha)_k v(\beta)_i +\ic \, \sgn(\beta - \alpha) v(\alpha)_k v(\beta)_i\,, \label{Eq:Pvv}\\
&& \br{v(\alpha)_i,\overline {v}(\beta)_k} =
\ic \delta_{ik} v(\alpha)_i \overline{v}(\beta)_k+2\ic \delta_{ik}\sum_{r>k}v(\alpha)_r \overline{v}(\beta)_r
 +\ic \delta_{\alpha \beta} v(\alpha)_i\overline {v}(\beta)_k \nonumber \\
&& \qquad \qquad \qquad \quad \
+ 2 \ic \delta_{\alpha \beta}\sum_{\mu<\alpha} v(\mu)_i \overline{v}(\mu)_k + 2 \ic \delta_{\alpha \beta} (b b^\dagger)_{ik}  \label{Eq:Pvbarv}\,.
\eea
The Poisson brackets of the dressed spins and the matrices $b,g$ are given by the following formulae
\bea
&& \br{v(\alpha)_i,g_{kl}} =
-\ic \delta_{ik} v(\alpha)_i g_{kl} - 2 \ic \delta_{(i<k)} v(\alpha)_k g_{il}\,,  \label{Eq:PvgR}\\
&&\br{\overline{v}(\beta)_i,g_{kl}}=
-\ic \delta_{ik} \overline{v}(\beta)_i g_{kl} - 2 \ic \delta_{ik}\sum_{r>i} \overline{v}(\beta)_r g_{rl}\,,  \label{Eq:PbarvgR}\\
&&\br{v(\alpha)_i,b_{kl}} =
2\ic \delta_{(k<i)} v(\alpha)_k b_{il} + \ic \delta_{ik} v(\alpha)_k b_{il}
- 2 \ic \sum_{s<l} b_{ks}v_s^\alpha b_{il} - \ic b_{kl} v_l^\alpha b_{il}\,, \label{Eq:PvbR} \\
&&\br{\overline{v}(\beta)_i,b_{kl}} =
-\ic \delta_{ik} \overline{v}(\beta)_i b_{kl} - 2 \ic \delta_{ik}\sum_{r>k} \overline{v}(\beta)_r b_{rl}
+ \ic \bar b_{il} \bar v_l^\beta b_{kl} + 2 \ic \sum_{s<l} \bar b_{is}\bar v_l^\beta b_{ks}\,. \label{Eq:PbarvbR}
\eea
\end{lemma}

Finally, we can express the Poisson structure in terms of the elements $ (g, L, v(1),\dots, v(d))$ where $L=b b^\dagger$. This is a direct application of Lemma \ref{L:Dressv} and the relations \eqref{T7}--\eqref{T8}, \eqref{T18}--\eqref{T19} of the Heisenberg double by using that $L_{kl}=\sum_r b_{ks}\bar b_{ls}$. Alternatively, one may use the formula \eqref{+PB2} to derive equations \eqref{Eq:gL}--\eqref{Eq:LL} below.
\begin{lemma} \label{L:vL}
The Poisson brackets involving $L$ are given by the following formulae
\bea
&& \br{v(\alpha)_i,L_{kl} }=
-\ic (2\delta_{(k>i)}+\delta_{ik}) v(\alpha)_k L_{il} + \ic \delta_{il} v(\alpha)_i L_{kl} + 2 \ic \delta_{il} \sum_{r>l} v(\alpha)_r L_{kr} \,, \label{Eq:vL} \\
&& \br{g_{ij},L_{kl} }= \ic (\delta_{ik}+\delta_{il}) g_{ij} L_{kl} + 2 \ic \delta_{(k<i)} g_{kj} L_{il} + 2 \ic \delta_{il} \sum_{r>i} L_{kr}g_{rj}\,, \label{Eq:gL} \\
&&\br{L_{ij},L_{kl} }= \ic [2 \delta_{(i>k)}+\delta_{ik}-2\delta_{(j>l)}-\delta_{lj}] L_{il}L_{kj} \nonumber \\
&& \qquad \qquad  \quad \ + \ic (\delta_{il}-\delta_{jk})L_{ij}L_{kl} + 2 \ic \delta_{il} \sum_{r>i} L_{kr} L_{rj} - 2 \ic \delta_{jk} \sum_{r> k} L_{ir} L_{rl}\,. \label{Eq:LL}
\eea
\end{lemma}

Now we present an interesting application of the above auxiliary results.
Recall that our `free Hamiltonians' \eqref{T34} Poisson commute with the functions $I_{\alpha \beta}^k$ defined in \eqref{Eq:Int},
and hence they Poisson commute with the elements of the polynomial algebra
\begin{equation} \label{Eq:Ical}
 \Ical=\R[\tr L^k,\Re(I_{\alpha\beta}^k),\Im(I_{\alpha\beta}^k) \mid 1\leq \alpha,\beta\leq d,\, k \geq 0]\,.
\end{equation}
The algebra $\Ical$ is finitely generated as a consequence of the Cayley-Hamilton theorem for $L$.
We also note that for an arbitrary non-commutative polynomial $\mathrm{P}$ obtained as a linear combination of products of the matrices $L$ and $v(\alpha)v(\beta)^\dagger$, $1\leq \alpha,\beta\leq d$, we have that $\tr(\mathrm{P})\in \Ical$  in view of the identity
\begin{equation}
 \tr\left(L^{a_0}v(\alpha_0)v(\beta_1)^\dagger L^{a_1} v(\alpha_1)v(\beta_2)^\dagger\cdots L^{a_{l}}v(\alpha_l)v(\beta_0)^\dagger L^{a_{l+1}}\right)
 =I_{\alpha_0 \beta_0}^{a_0+a_{l+1}} I_{\alpha_1 \beta_1}^{a_1} \cdots I_{\alpha_l \beta_l}^{a_l}\,.
\end{equation}
A key property of $\Ical$ is that it is a Poisson subalgebra of $C^\infty(\cM)$.
This follows from the next result, which can be proved by direct calculation.
\begin{proposition}
 \label{P:BrII}
 For any $M,N\geq0$ and $1\leq \alpha,\beta,\gamma,\epsilon\leq d$,
 \begin{equation}
  \begin{aligned} \label{Eq:BrII}
\br{I_{\alpha\beta}^M,I_{\gamma\epsilon}^N}=&
\,\,2 \ic \delta_{\alpha \epsilon} I_{\gamma\beta}^{M+N+1} - 2 \ic \delta_{\gamma\beta} I_{\alpha\epsilon}^{M+N+1} \\
&+\ic (\delta_{\alpha\epsilon} - \delta_{\gamma\beta}) I_{\alpha\beta}^M I_{\gamma\epsilon}^N
+2 \ic \delta_{\alpha\epsilon} \sum_{\mu<\alpha}  I_{\gamma\mu}^N I_{\mu\beta}^M
-2\ic \delta_{\gamma\beta} \sum_{\lambda<\beta} I_{\alpha\lambda}^M I_{\lambda\epsilon}^N \\
&+\ic\, \sgn(\gamma-\alpha) I_{\gamma\beta}^M I_{\alpha \epsilon}^N
-\ic\, \sgn(\epsilon-\beta) I_{\gamma\beta}^N I_{\alpha\epsilon}^M \\
&+\ic \left(\sum_{b=0}^{M-1}+\sum_{b=0}^{N-1} \right)
\left(I_{\gamma\beta}^b I_{\alpha\epsilon}^{M+N-b} - I_{\gamma\beta}^{M+N-b} I_{\alpha\epsilon}^b  \right)\,.
\end{aligned}
 \end{equation}
\end{proposition}
Using that $\overline{I_{\alpha\beta}^M} = I_{\beta\alpha}^M$, one can verify that the complex Poisson brackets in \eqref{Eq:BrII} enjoy the property
\begin{equation}
 \br{\overline{I_{\alpha\beta}^M},\overline{I_{\gamma\epsilon}^N}}=\overline{\br{I_{\alpha\beta}^M,I_{\gamma\epsilon}^N}},
\end{equation}
which together with \eqref{Eq:BrII} implies that $\Ical$ \eqref{Eq:Ical} is indeed a real Poisson algebra.
Since the elements of $\Ical$ are invariant with respect to the $\U(n)$ action on $\cM$, they descend to the reduced phase space.
We shall further inspect these integrals of motion in Section \ref{sec:DIS}.

\subsection{The reduced Poisson bracket in local coordinates} \label{ssec:RP2}

In this subsection we shall derive explicit formulae for the reduced Poisson structure,
restricting ourselves to an open dense subset $\ccM_\red^\reg$ of the reduced phase space.
More precisely, it will be more convenient to work on a covering space  of $\ccM_\red^\reg$
 that supports residual $S_n$ gauge transformations.

We start by introducing the open dense subset $\ccM^\reg_0\subset \cM^\reg_0$ \eqref{H41},
 which is defined as
\be
\ccM^\reg_0:= \{ (Q, L(Q,v), v)\in \cM_0^{\reg} \mid \sum_{1\leq \alpha \leq d} v(\alpha)_i \neq 0 \text{ for }i=1,\ldots,n\, \}\,.
\label{RP2:1}
\ee
The corresponding open dense subset of the reduced phase space is
\be
\ccM_\red^\reg := \ccM^\reg_0/\cN(n)\,.
\label{RP2:2}
\ee
Using that $S_n=\cN(n)/\T^n$, we can take
the  quotient in two steps.  Thus, similarly to \eqref{stages1},
 we have
\be
\ccM_0^\reg/\cN(n) =\left( \ccM_0^\reg/\T^n\right)/S_n.
\ee
For our purpose, we choose to  identify $ \ccM_0^\reg/\T^n$ with the following
subset of $\ccM_0^\reg$
\be \label{RP2:1bis}
\ccM_{0,+}^\reg  :=
\{ (Q, L(Q,v), v)\in \cM_0^\reg \mid \sum_{1\leq \alpha \leq d} v(\alpha)_i >0
\ \hbox{for}\  i=1,\dots, n\}.
\ee
Indeed, if $(Q, L, v)\in\ccM^\reg_0$ we can write the vector
$\sum_{1\leq \alpha \leq d} v(\alpha)$ as $(U_1 e^{\ic \phi_1}, \ldots, U_n e^{\ic \phi_n})^T$ with $U_i >0$, $\phi_i \in \R$ for
all $i$. Acting on $(Q,L,v)$ by $\tau = \diag(e^{-\ic \phi_1}, \ldots,e^{-\ic \phi_n})\in \T^n$ yields an element of $\ccM_{0,+}^\reg$,
 and it is clear that $\tau$ is the unique element of $\T^n$ with this property.
The upshot is the identification
\be
\check\cM_\red^\reg \equiv  \check\cM_{0,+}^\reg/S_n.
\ee

The main reason for introducing the particular gauge slice $\ccM_{0,+}^\reg$  for the \emph{free} $\T^n$ action
on $\ccM_0^\reg$ is that $S_n$ still acts on it in the obvious manner, by permuting
the $n$ entries of $Q$ and the components of each column vector $v(\alpha)\in \C^n$.
Similar `democratic gauge fixing' was employed in the previous papers dealing with
holomorphic systems \cite{AF,AO,CF2}.  The relation between the spaces just defined and those
given in Section \ref{sec:H} is summarized in Figure \ref{FigEmb}.

\begin{figure}
 \begin{diagram}
\ccM^\reg_{0,+}&& \rInto && \ccM^\reg_0 && \rInto && \cM^{\reg}_0 && \rInto && \cM_0 && \rInto && \Lambda^{-1}(e^\gamma \1_n)  \\
&\rdTo(4,2)_{^{-/S_n}} &&&\dTo_{^{-/ \cN(n)}} &&&&\dTo_{^{-/ \cN(n)}}  &&&& \dTo &&  & \ldTo(4,2)_{^{-/\U(n)}} &\\
&&&&\ccM_\red^\reg && \rInto && \cM^{\reg}_\red && \rInto && \cM_\red &&&&
\end{diagram}
\caption{(From right to left.) $\cM_0$ is the subspace \eqref{H18} of the constraint surface $\Lambda^{-1}(e^\gamma \1_n)$ where each point $(Q,b_R,v)$ satisfies that $Q \in \T^n$.
$\cM_0^{\reg}\subset \cM_0$ is the subspace \eqref{H31} where $Q\in \T^n_{\reg}$,
$\ccM^\reg_0\subset \cM_0^{\reg}$ is the subspace \eqref{RP2:1} where the vector $\sum_{\alpha=1}^d v(\alpha)$ has only nonzero entries,
while $\ccM^\reg_{0,+}\subset \ccM^\reg_0$ is the subspace \eqref{RP2:1bis} obtained by imposing to the vector $\sum_{\alpha=1}^d v(\alpha)$ to have positive entries.
The spaces appearing on the second line are the sets corresponding to the $\U(n)$-orbits inside $\Lambda^{-1}(e^\gamma \1_n)$.
\label{FigEmb}}
\end{figure}

Let
\be
\xi: \ccM_{0,+}^\reg \to \cM
\ee
be the tautological inclusion.  General principles of  reduction theory \cite{HT,OR}
ensure that the pull-back $\xi^* \Omega_\cM$ is
symplectic and satisfies $\xi^*\Omega_\cM = {\check \pi}^* (\Omega_\red)$, where
${\check \pi}: \ccM_{0,+}^\reg \to \cM_\red$ is the canonical projection and $\Omega_\red$ is
the reduced symplectic form.
We let $\{\ ,\ \}_\red$ denote the Poisson bracket  on $C^\infty(\ccM_{0,+}^\reg)$ that corresponds to
$\xi^* \Omega_\cM$ \eqref{OmcM}, and
note that it possesses the key property
\be\label{Eq:PBrel}
\xi^* \{ F_1, F_2\} = \{ \xi^* F_1, \xi^* F_2\}_\red,
\qquad
\forall F_1, F_2 \in C^\infty(\cM)^{\U(n)},
\ee
where $\{F_1,F_2\} := \{ F_1, F_2\}_\cM$ is the Poisson bracket associated with $\Omega_\cM$ \eqref{OmcM}.
We shall determine the form of this reduced Poisson bracket by applying the identity \eqref{Eq:PBrel} to
a judiciously chosen set of invariant functions.

\begin{remark}\label{Rem:Dirac}
The bracket $\{\ ,\ \}_\red$ is also known as the  Dirac bracket \cite{HT} associated with the gauge slice
$\ccM_{0,+}^\reg$.
To avoid any potential confusion,
we stress that our notation $\{\ ,\ \}_\red$ involves a slight abuse of terminology,
since not all elements of $C^\infty(\ccM_{0,+}^\reg)$ arise as
restrictions of elements of $C^\infty(\cM)^{\U(n)}$ (which carries the reduced Poisson algebra
in the strict sense).
For example, all those restricted $\U(n)$ invariants are $S_n$ invariant function on $\ccM_{0,+}^\reg$.
However, the Poisson algebra $(C^\infty(\ccM_{0,+}^\reg), \{\ ,\ \}_\red)$
encodes all information about $(C^\infty(\cM)^{\U(n)}, \{\ ,\ \})$,
since $\ccM_{0,+}^\reg$ projects onto a dense open subset of $\cM_\red$. We shall see shortly that
 it underlies the Hamiltonian interpretation of the spin RS equations of motion given by \eqref{I8}.
\end{remark}

In order to implement the above ideas, now we introduce the following $\U(n)$ invariant
elements of $C^\infty(\cM)$
\be \label{Eqf}
 f^{\alpha \beta}_m:= \tr (v(\alpha) v(\beta)^\dagger g_R ^m) = v(\beta)^\dagger  g_R^m v(\alpha)\,, \quad  f_m:= \tr (g_R^m)\,, \quad
 m \in \N,\,\, 1\leq \alpha,\beta \leq d\,.
\ee
\begin{lemma} \label{Lem:PBff}
 For any $M,N \in \N$ and $1\leq \alpha,\beta \leq d$ we have the following Poisson bracket
 relations in $C^\infty(\cM)^{\U(n)}$:
\bea
&&  \br{f_M,f_N}=0\,, \quad \br{f_M, \bar f_N}=0\,, \quad \br{\bar f_M, \bar f_N}=0\,, \label{Eq:ff} \\
&& \br{f_M^{\alpha \beta},f_N}=-2\ic N f_{M+N}^{\alpha \beta}\,. \label{Eq:fspf}
\eea
Furthermore, letting $\phi^{\mu \nu}(a,c):=\tr[v(\mu)v(\nu)^\dagger g_R^a L g_R^c]$ for $a,b \in \N$, we have
 \begin{equation}
 \begin{aligned} \label{Eq:fspfsp}
    \br{f_M^{\alpha \beta},f_N^{\gamma \epsilon}}=&
    \,\,2\ic \left(\sum_{a=1}^M - \sum_{a=1}^N \right)f_a^{\alpha \epsilon} f_{M+N-a}^{\gamma \beta}
    -\ic f_M^{\alpha \epsilon} f_N^{\gamma \beta} + \ic f_N^{\alpha \epsilon} f_M^{\gamma \beta}\\
    &+\ic \sgn(\gamma-\alpha) f_N^{\alpha \epsilon} f_M^{\gamma \beta} -\ic \sgn(\epsilon-\beta) f_M^{\alpha \epsilon} f_N^{\gamma \beta}
    +\ic (\delta_{\alpha \epsilon} - \delta_{\gamma \beta}) f_M^{\alpha \beta} f_N^{\gamma \epsilon} \\
    &+2\ic \delta_{\alpha \epsilon}\sum_{\mu<\alpha} f_N^{\gamma \mu} f_M^{\mu \beta}
    - 2 \ic \delta_{\gamma \beta} \sum_{\lambda<\beta} f_M^{\alpha \lambda} f_N^{\lambda \epsilon} \\
    &+2\ic \delta_{\alpha \epsilon} \phi^{\gamma \beta}(M,N) - 2 \ic \delta_{\gamma \beta} \phi^{\alpha \epsilon}(N,M)\,.
 \end{aligned}
 \end{equation}
\end{lemma}
\begin{proof}
The identities \eqref{Eq:ff} are well-known.
To establish \eqref{Eq:fspf} we use the decomposition (here $g=g_R$)
  \begin{equation}
  \begin{aligned}
  \br{f_M^{\alpha \beta},f_N}=&N\sum_{ijkl}\br{v(\alpha)_i ,g_{kl}} \overline{v}(\beta)_j g^M_{ji}g^{N-1}_{lk}
 + N\sum_{ijkl}\br{\overline{v}(\beta)_j,g_{kl}} g^M_{ji}v(\alpha)_i g^{N-1}_{lk}\\
 & +N \sum_{ijkl}\sum_{b=0}^{M-1}\br{g_{ij},g_{kl}} (g^{M-b-1}v(\alpha)v(\beta)^\dagger g^b)_{ji}g^{N-1}_{lk}\,,
  \end{aligned}
 \end{equation}
 then use for these three terms \eqref{Eq:PvgR}, \eqref{Eq:PbarvgR} and \eqref{T7} respectively. Some obvious cancellations yield \eqref{Eq:fspf}.
Finally, \eqref{Eq:fspfsp} only requires some of the Poisson brackets gathered in \S\ref{ssec:RP1} and it can be proved in a way similar to \eqref{Eq:fspf}.
\end{proof}

Convenient variables on $\ccM_{0,+}^\reg$ are provided by the evaluation functions $Q_j=e^{\ic q_j} \in \mathrm{U}(1)$
and the real and imaginary parts of the $v(\alpha)_j\in \C$. The latter are not all independent,
since they obey the gauge fixing conditions
\be  \label{Eq:cU}
\cU_j=\Re(\cU_j)>0\,, \quad \text{ with } \quad \cU_j:=\sum_{1\leq \alpha \leq d}v(\alpha)_j\,, \quad 1\leq j \leq n\,.
\ee
It is clear that all these functions belong to $C^\infty(\ccM_{0,+}^\reg)$ and their mutual Poisson brackets
completely determine $\{\ ,\ \}_\red$.

The pull-backs of the functions \eqref{Eqf} can be written in the local variables on $\ccM_{0,+}^\reg$ as
\be \label{Eqfbi}
\xi^* f^{\alpha \beta}_m =\sum_{i=1}^n v(\alpha)_i Q_i^m \overline{v}(\beta)_i\,, \quad
\xi^* f_m = \sum_{i=1}^n Q_i^m\,,
\ee
and  we note that
\be
\sum_{\beta}  \xi^* f^{\alpha \beta}_m =\sum_{i=1}^n \cU_i v(\alpha)_i Q_i^m\,, \quad
\sum_{\alpha} \xi^*  f^{\alpha \beta}_m  =\sum_{i=1}^n \cU_i \overline{v}(\beta)_i Q_i^m\,, \quad
\sum_{\alpha,\beta}  \xi^* f^{\alpha \beta}_m  =\sum_{i=1}^n \cU_i^2 Q_i^m\,.
\ee
In conjunction with Lemma \ref{Lem:PBff} and equation \eqref{Eq:PBrel}, these expressions can be used to determine the reduced Poisson
brackets of the variables $Q_j$ and $v(\alpha)$.
To state the result, we introduce the $n\times n$ matrix-valued functions $S^0$ and $R^\alpha$, $1\leq \alpha \leq d$, whose entries are given by
\bea
&&  S^0_{ij}=\frac14\sum_{\mu,\nu} \sgn(\nu-\mu) v(\nu)_i v(\mu)_j  \label{Eq:redSij}
-\frac14 \sum_{\mu} v(\mu)_i \overline{v}(\mu)_j
-\frac12 \sum_{\nu}\sum_{\mu<\nu} v(\mu)_i \overline{v}(\mu)_j  - \frac{d}{2} L_{ij},\quad \qquad  \\
&&  R^\alpha_{ij}=L_{ij}-\frac12  \sum_{\kappa} \sgn(\kappa-\alpha)v(\kappa)_i  v(\alpha)_j
  +\frac12 v(\alpha)_i \overline{v}(\alpha)_j + \sum_{\kappa< \alpha} v(\kappa)_i  \overline{v}(\kappa)_j.\quad \label{Eq:redRij}
\eea
We also define the matrix $S$ with entries $S_{ij}=S^0_{ij}-\overline{S}^0_{ij}$.

\begin{theorem} \label{Thm:redPB}
In terms of the functions $(Q_j=e^{\ic q_j},v(\alpha)_j)$ defined on $\ccM_{0,+}^\reg$, and using the formulae \eqref{H21} for $L$
and \eqref{Eq:cU} for $\cU_j$, we can write the
reduced Poisson bracket as
\be
\br{q_i,q_j}_\red=0\,, \quad  \label{Eq:red1}
\br{v(\alpha)_i,q_j}_\red=-\delta_{ij} v(\alpha)_i\,,
\ee
\bea
   & \br{v(\alpha)_i,v(\gamma)_j}_\red=  \ic\, \sgn(\gamma-\alpha) v(\alpha)_j v(\gamma)_i
   + \ic\frac{v(\alpha)_i}{\cU_i}\frac{v(\gamma)_j}{\cU_j} S_{ij}
  + \ic \frac{v(\gamma)_j}{\cU_j} R^\alpha_{ij}
  - \ic \frac{v(\alpha)_i}{\cU_i} R^\gamma_{ji}
  \nonumber \\
  &\quad +\frac12 \ic \delta_{(i \neq j)} \frac{Q_i+Q_j}{Q_i-Q_j}\left[2 v(\alpha)_j v(\gamma)_i+ v(\alpha)_iv(\gamma)_j -\frac{\cU_i}{\cU_j}v(\alpha)_j v(\gamma)_j -
  \frac{\cU_j}{\cU_i} v(\alpha)_i v(\gamma)_i \right]\,,
 \label{Eq:red2}
 \eea
 \bea
    &\br{v(\alpha)_i,\overline{v}(\epsilon)_j}_\red= \ic
    \delta_{\alpha \epsilon} \left(v(\alpha)_i \overline{v}(\epsilon)_j + 2 \sum_{\kappa< \alpha} v(\kappa)_i
    \overline{v}(\kappa)_{j}+2L_{ij}\right) \nonumber \\
  &\qquad +\frac12 \ic \delta_{(i \neq j)} \frac{Q_i+Q_j}{Q_i-Q_j}\left[- v(\alpha)_i
  \overline{v}(\epsilon)_j +\frac{\cU_i}{\cU_j}v(\alpha)_j \overline{v}(\epsilon)_j +
  \frac{\cU_j}{\cU_i} v(\alpha)_i \overline{v}(\epsilon)_i \right]    \label{Eq:red3}\\
  &- \ic\frac{v(\alpha)_i}{\cU_i}\frac{\overline{v}(\epsilon)_j}{\cU_j} S_{ij}
  -\ic \frac{\overline{v}(\epsilon)_j}{\cU_j} R^\alpha_{ij}
  - \ic \frac{v(\alpha)_i}{\cU_i} \overline{R}^\epsilon_{ji}\,. \qquad\qquad\qquad \qquad \nonumber
 \eea
The bracket $\br{-,-}_\red$ is invariant under  simultaneous permutations of the $n$ components
of the variables $q$ and $v(\alpha)$ for $\alpha=1,\dots, d$.
\end{theorem}

The proof of this result is the subject of Appendix \ref{A:RedPr}.
Let us already mention that the reader can check the  reality condition
$\br{\overline{v}(\alpha)_i,v(\epsilon)_j}_\red  = \overline{\br{v(\alpha)_i,\overline{v}(\epsilon)_j}}_\red$.

\medskip

We know from Corollary \ref{Cor:vf} that the projection of the Hamiltonian vector field of $H=  (e^{2\gamma}-1)\, \tr (L)$
onto the gauge slice $\ccM_{0,+}^\reg$
leads to the equations of motion \eqref{I8}.
Of course, the corresponding reduced Hamiltonian must generate the same evolution equations
via the reduced Poisson bracket.  The reduced Hamiltonian is encoded
by the pull-back $\cH := \xi^*H$ on $\ccM_{0,+}^\reg$.
Thus, the next result shows the consistency of the computations performed in Section \ref{sec:N} and Section \ref{sec:RP}.

 \begin{corollary}\label{Cor:vfredPB}
 Consider the reduced Hamiltonian
 \be
 \cH(Q,v)=(e^{2\gamma}-1)\, \tr (L(Q,v)) = \sum_{k=1}^n F_{kk},
 \qquad F_{kk} = \sum_{\alpha=1}^d v(\alpha)_k \overline{v}(\alpha)_k,
 \label{redHam}\ee
  on the gauge slice $\ccM_{0,+}^\reg$.
Then the Hamiltonian vector field generated by $\cH$ via the reduced Poisson bracket of Theorem \ref{Thm:redPB}
reproduces the equations of motion \eqref{I8}--\eqref{I9}.
 \end{corollary}
\begin{proof}
 We get from \eqref{Eq:red1} that
 \be
\dot q_j := \br{q_j,\cH }_\red = \sum_k \br{q_j, F_{kk}}_\red  =2 F_{jj}\,,
 \ee
 which is just \eqref{I8}. To compute $\dot{v}(\alpha)_i$, we use that $F_{kl}=\sum_{\alpha=1}^d v(\alpha)_k \overline{v}(\alpha)_l$
 together with Theorem \ref{Thm:redPB} in order to obtain
  \begin{equation}
  \begin{aligned}
&\br{v(\alpha)_i,F_{kk}}_\red=\ic\delta_{(i \neq k)} v(\alpha)_k \left[F_{ik} + \frac{Q_i+Q_k}{Q_i-Q_k}F_{ik}+2L_{ik} \right] \\
&\quad-\frac12\ic\delta_{(i \neq k)} v(\alpha)_i \frac{\cU_k}{\cU_i} \left[\left(F_{ik} + \frac{Q_i+Q_k}{Q_i-Q_k}F_{ik}+2L_{ik}\right) + \left(F_{ki} + \frac{Q_k+Q_i}{Q_k-Q_i}F_{ki}+2L_{ki}\right)\right].
\end{aligned}
 \end{equation}
Noticing the identity
\begin{equation} \label{Eq:FQQ}
 \ic\left(F_{ik} + \frac{Q_i+Q_k}{Q_i-Q_k}F_{ik}+2 L_{ik}\right)=-F_{ik} V\left(\frac{q_k-q_i}{2}\right)\,,
\end{equation}
  where $V(x)$ is the potential \eqref{I5}, this allows us to write
   \begin{equation}
  \begin{aligned}
\br{v(\alpha)_i,F_{kk}}_\red =&-\delta_{(i \neq k)} v(\alpha)_k F_{ik}V\left(\frac{q_k-q_i}{2}\right) \\
&+\frac12\delta_{(i \neq k)} v(\alpha)_i \frac{\cU_k}{\cU_i} \left[F_{ik}V\left(\frac{q_k-q_i}{2}\right) + F_{ki} V\left(\frac{q_i-q_k}{2}\right) \right] \,.
\end{aligned}
 \end{equation}
 Summing  over $k$ precisely gives  $\dot{v}(\alpha)_i$ in  \eqref{I8} with \eqref{I9}.
\end{proof}

As a second consequence of Theorem \ref{Thm:redPB}, we can write down the reduced Poisson brackets of the
 `collective spins' $(F_{ij})$, which can be found in Appendix \ref{ssA:Coll2}.
 By using equation \eqref{H21}, then we can obtain the formula for the Poisson brackets of the entries of the
 Lax matrix on  $\ccM_{0,+}^\reg$, which implies that the symmetric functions of $L$ are in involution.
 This is in agreement with the fact that $\fH$ given in \eqref{T34} is an Abelian Poisson subalgebra of $C^\infty(\cM)$.
To present the desired formula, we use the matrix $S$ defined before Theorem \ref{Thm:redPB}. We also define
\bea
&&\! r_{12}:=\sum_{a\neq b} \frac{\ic Q_b}{Q_a-Q_b} E_{aa} \otimes \left(E_{bb} - \frac{\cU_b}{\cU_a} E_{ba} \right)
  -\sum_{a\neq b} \frac{\ic Q_a}{Q_a-Q_b} E_{ab} \otimes \left(\frac{\cU_a}{\cU_b} E_{bb}-2E_{ba} \right) \nonumber \\
&&\qquad +\sum_{a,b} \ic \frac{S_{ab}}{\cU_a \cU_b}\, E_{aa}\otimes E_{bb} + \ic \sum_a E_{aa}\otimes E_{aa}\,, \label{r12}
\eea
and
\bea
&&\! s_{12}:=\sum_{a\neq b} \frac{\ic Q_a}{Q_a-Q_b} E_{aa} \otimes \left(\frac{\cU_b}{\cU_a} E_{ab}-E_{bb} \right)
  + \sum_{a\neq b} \frac{\ic Q_a}{Q_a-Q_b}\frac{\cU_a}{\cU_b} E_{ab} \otimes E_{bb} \nonumber \\
&&\qquad -\sum_{a,b} \ic \frac{S_{ab}}{\cU_a \cU_b}\, E_{aa}\otimes E_{bb} + \frac12 \ic \sum_a E_{aa}\otimes E_{aa}\,, \label{s12}
\eea
where $E_{ab}$ is the  $n\times n$ elementary matrix  with only nonzero entry equal to $+1$ in position $(a,b)$.
\begin{proposition} \label{Pr:RUform}
 On the gauge slice  $\ccM_{0,+}^\reg$   \eqref{RP2:1bis}, the entries of the Lax matrix $L$ \eqref{H21} satisfy
 \begin{equation} \label{RUform}
 \br{L_1,L_2}_\red=r_{12} L_1L_2 + L_1L_2 t_{12} - L_1 s_{21} L_2 + L_2 s_{12} L_1 \,,
\end{equation}
where $t_{12}=-s_{12}+s_{21}-r_{12}$. This relation implies  that the functions $\tr (L^k)$ are in involution.
\end{proposition}
In \eqref{RUform}, we used the standard notations $L_1=L \otimes \1_n$, $L_2=\1_n \otimes L$, and
$\br{L_1,L_2}_\red=\sum_{ijkl} \br{L_{ij},L_{kl}}_{\red} E_{ij}\otimes E_{kl}$, where the entries of $L$ are seen as
evaluation functions on $\ccM_{0,+}^\reg$.
\begin{remark}
The formulae of Theorem \ref{Thm:redPB} exhibit an interesting two-body structure in the
sense that the Poisson brackets of the basic variables with particle labels $i$ and $j$ close
on this subset of the variables. This is consistent with the
fact that the Hamiltonian \eqref{I10} is the sum of one-body terms, while the equations of motion
\eqref{I8}--\eqref{I9}
reflect two-body interactions. It should be stressed that this interpretation is based on viewing
$q_i$ and the dressed spin $v(-)_i$ as degrees of freedom belonging to particle $i$.
The same features hold in the complex holomorphic spin RS models as well \cite{AF,AO,CF2}.
It is also worth noting that
 the formulae of Theorem \ref{Thm:redPB} enjoy a  nice homogeneity property. Namely,
let us define
a $\Z^n$-valued weight
 $\wt{-}$   by setting
\begin{equation}
\wt{1} = \wt{ q_j} =\wt{ Q_j } = 0\,, \quad
 \wt{ \overline{v}(\alpha)_j } = \wt{ v(\alpha)_j} =\BE_j\,,
  \quad \text{for }  1\leq j \leq n \,,
\end{equation}
where $\BE_j\in \Z^n$ is $+1$ in its $j$-th entry and zero everywhere else. Extending this weight by
$ \wt{ f g } = \wt{ f } +
\wt{ g } $ for homogeneous elements $f,g$, we easily get that
 \be
  \wt{ \cU_j^{\pm1} } =\pm\BE_j\,,
  \quad  \wt{ F_{ij} } =\BE_i+\BE_j\,, \quad
  \wt{ L_{ij} } =\BE_i+\BE_j\,.
 \ee
We can then observe from  \eqref{Eq:red1}--\eqref{Eq:red3}
that the reduced Poisson bracket preserves this weight.
\end{remark}

\section{Degenerate integrability of the reduced system} \label{sec:DIS}

We discussed the degenerate integrability of the unreduced free system in \S\ref{ss:unredIS}, and now wish to show that this
property is inherited by the reduced system.
This is expected to hold not only  in view of the earlier
results on  holomorphic spin RS systems \cite{AO,CF2} and related  models  \cite{Res1,Res2,Res3}, but also on account
of a general result in reduction theory.
In fact, it is known (Theorem 2.16 in \cite{Zung}, see also \cite{J})
that the integrability of invariant Hamiltonians on a manifold
descends generically to the reduced space of Poisson reduction.
However, the pertinent spaces of group orbits  are typically not smooth manifolds.
The existing results provide strong motivation, but  do not help us directly
 to establish integrability in our concrete case.

Our goal is
to prove the  degenerate integrability of the reduced system in the
real-analytic category by explicitly displaying the required integrals of motion.
Specifically, we wish to show that the $n$ reduced Hamiltonians arising from the functions
\be
\tr(L^k), \qquad k=1,\dots, n,
\label{S1}\ee
are functionally independent, and that one can complement them to $(2nd -n)$ independent functions
using suitable
reduced integrals of motion that arise from the real and imaginary parts of the $\U(n)$ invariant functions
\be
I^k_{\alpha \beta} = v(\beta)^\dagger L^k v(\alpha).
\label{S2}\ee
These integrals of motion appeared before in Proposition \ref{P:BrII}.
As throughout the paper, we assume that $d\geq 2$.

The independence of functions means linear independence of their exterior derivatives
at generic points, and this can be translated into the
non-vanishing of a suitable Jacobian determinant.
For real-analytic functions,
the determinant at issue is also real-analytic, and hence it is generically non-zero if
it is non-zero at a single point.
Thus, by patching together analytic charts, one sees that on a connected real-analytic manifold
independence of real analytic functions follows from the linear independence of their derivatives at
a single point. We can use this observation since we know (see Remark \ref{Rem:EmbComplex}
 and Corollary \ref{Cor:Dense}) that $\cM_\red$ is a connected
real-analytic manifold.

\subsection{Construction of local coordinates}

Our first goal below is to construct local coordinates around certain points
of the reduced phase space in which the formulae of the integrals of motion
become simple. The coordinates will involve the eigenvalues of $L$, whereby
the Hamiltonians $\tr(L^k)$ acquire a trivial form.
We start by noting that the moment map constraint admits solutions for which
only a single one of the vectors $v(\alpha)$ is non-zero.  Concerning those elements
of $\Lambda^{-1}(e^\gamma \1_n)$,  the following useful result can be
obtained from (the proof of) Lemma 5.2 of \cite{FK2}.

\begin{lemma} \label{Lem:d1}
Consider any  $y \in \R^n$ whose components $ y_1,\dots, y_n$ satisfy the inequalities
\begin{equation} \label{Eq:yL0}
 y_i> e^{2\gamma}y_{i+1}
 \quad\forall i=1,\dots, n\quad \hbox{with}\quad y_{n+1}:= 0.
 \end{equation}
Then there exists $(g^0,L^0,v^0)\in \Lambda^{-1}(e^\gamma \1_n)$ such that $L^0=\diag(y_1,\ldots,y_n)$ and
$v(\alpha)^0=0$ for each $1\leq \alpha <d$ (where $d\geq 2$ and $\gamma >0$). For such elements all components
of the vector $v(d)^0$ are non-zero.
\end{lemma}
\begin{proof}
Given $L^0=\diag(y_1,\ldots,y_n)$ and $v(1)^0=\ldots=v(d-1)^0=0$, we have to find $g^0\in \U(n)$ and $v(d)^0\in \C^n$ such that the moment map
constraint \eqref{momconst} holds. Using \eqref{momconstdag}, this means that
 \begin{equation}
  e^{2\gamma} (g^0)^{-1} L^0 g^0   = L^0 + v(d)^0 (v(d)^0)^\dagger.
\label{momconstd1}
 \end{equation}
 This is equivalent to the requirement that there exists $v(d)^0\in\C^n$
 such that $L^0 + v(d)^0 (v(d)^0)^\dagger$ and $e^{2\gamma}L^0$ have the same spectrum.
 But this holds if and only if we have the equality of polynomials in $\lambda$
 \begin{equation} \label{Eq:LL0spec}
  \det(L^0 + v(d)^0 (v(d)^0)^\dagger-\lambda \1_n)=\det( e^{2\gamma}L^0 -\lambda \1_n) = \prod_{k=1}^n (e^{2\gamma} y_k-\lambda)\,.
 \end{equation}
 We can expand the left-hand side as follows :
 \begin{equation}
  \begin{aligned}
    \det(L^0 + v(d)^0 (v(d)^0)^\dagger-\lambda \1_n) =& \det(L^0-\lambda \1_n) [1+(v(d)^0)^\dagger (L^0-\lambda \1_n)^{-1} v(d)^0] \\
    =&\prod_{k=1}^n (y_k-\lambda) + \sum_{j=1}^{n}|v(d)_j^0|^2 \prod_{k\neq j} (y_k-\lambda)\,.
  \end{aligned}
 \end{equation}
 Thus, we seek $v(d)^0\in \C^n$ such that
 \begin{equation} \label{Eq:LL0spec2}
  \prod_{k=1}^n (e^{2\gamma} y_k-\lambda)-\prod_{k=1}^n (y_k-\lambda) = \sum_{j=1}^{n}|v(d)_j^0|^2 \prod_{k\neq j} (y_k-\lambda)\,.
 \end{equation}
 Evaluating this identity at $\lambda=y_l$ yields
  \begin{equation} \label{Eq:LL0spec3}
  |v(d)_l^0|^2= (e^{2\gamma}-1)y_l   \prod_{k\neq l} \frac{e^{2\gamma} y_k-y_l}{y_k-y_l}\,,
 \end{equation}
 which is \emph{positive due to \eqref{Eq:yL0}}.
It now suffices to pick $v(d)^0$ whose components have moduli given by \eqref{Eq:LL0spec3}, while we pick for $g^0$ any unitary matrix diagonalizing $L^0 + v(d)^0 (v(d)^0)^\dagger$ into
$e^{2\gamma} L^0$.
\end{proof}

\begin{remark}
Notice that a completely gauge fixed normal form of the elements appearing in Lemma \ref{Lem:d1} can be obtained
by requiring all components of the vector $v(d)^0$ to be positive.  We also note in passing  that in the $d=1$ case
the set of possible (ordered) eigenvalues of $L$ in $(g_R,L,v)\in \Lambda^{-1}(e^\gamma \1_n)$ is given \cite{FK2} by the polyhedron in $\R^n$ specified by the
conditions  $y_i\geq e^{2\gamma}y_{i+1}$
 for all $i=1,\dots, n-1$ and $y_n>0$.
\end{remark}

Now we introduce two subsets of the inverse image of the `constraint surface'.
\begin{definition} \label{def:S1}
Denote
 \begin{equation} \label{Eq:Sspace}
 \cS=\{(g_R,L,v) \in \Lambda^{-1}(e^\gamma \1_n) \mid L=\diag(y_1,\ldots,y_n),\, y_i>y_{i+1},\, v(1)_i>0 \,\, \forall i \}.
\end{equation}
The open subset $\cS_1\subset \cS$ is defined by imposing the
further condition  that the matrix
\be
L_1:= L+\sum_{\alpha=1}^{d-1} v(\alpha)v(\alpha)^\dagger
\label{defL1}\ee
is conjugate to $\diag(\mu_1,\dots, \mu_n)$, where the $\mu_i$ satisfy the inequalities
\be
e^{2\gamma} y_i > \mu_i > e^{2\gamma}y_{i+1}
 \quad\forall i=1,\dots, n-1\quad \hbox{and}\quad e^{2\gamma} y_n > \mu_n.
\label{Eq:ymuL}\ee
\end{definition}
Note that $\cS$ is non-empty since we can apply the analogue of Lemma \ref{Lem:d1} to obtain elements of $\Lambda^{-1}(e^\gamma \1_n)$  for which only $v(1)$ is non-zero,
and $\cS_1$ is non-empty since for those elements $L_1=L$.
It is clear that $\cS$ can serve as a model of an open dense subset of the reduced phase space.
Below, we provide a full characterization of the elements of $\cS_1$.

Taking $y$ and $\mu$ subject to the inequalities in \eqref{Eq:ymuL}, define $\bV(y,\mu) \in \R^n$ by
 \begin{equation} \label{Eq:LL0plusspec3}
 \bV_l(y,\mu):= \left[ (e^{2\gamma}y_l-\mu_l)   \prod_{k\neq l} \frac{e^{2\gamma} y_k-\mu_l}{\mu_k-\mu_l}\right]^{\frac{1}{2}}\, \quad \forall l=1,\dots, n.
 \end{equation}
Observe that the function under the square root is positive; and the positive root is taken.

\begin{lemma} \label{Lem:Laci1}
For any $(g_R, L, v)\in \cS_1$ pick a matrix $g_1 \in \U(n)$ for which
\be
g_1 L_1 g_1^{-1} = \diag(\mu_1,\dots, \mu_n)
\label{L1g1}\ee
with $\mu$ satisfying \eqref{Eq:ymuL}.
Then $v(d)$ is of the form
\be
v(d) = g_1^{-1} \diag(\tau_1,\dots, \tau_n) \bV(y,\mu) \quad \hbox{with some}\quad \tau\in \T^n.
\label{vdeq}\ee
Furthermore, $g_R$ is of the form
\be
g_R = \diag(\Gamma_1,\dots, \Gamma_n) g_R^0 \quad\hbox{with some}\quad \Gamma \in \T^n,
\label{gr}\ee
where $g_R^0 \in \U(n)$ is a fixed solution of the constraint equation
\be
g_R^{-1} e^{2\gamma} L g_R = L_1 + v(d) v(d)^\dagger.
\label{gReq}\ee
Conversely, take any  positive definite $L=\diag(y_1,\dots, y_n)$ and $\C^n$ vectors $v(1),\dots, v(d-1)$ such that $L$ and $L_1$ given by \eqref{defL1} satisfy the spectral
conditions   \eqref{Eq:ymuL}, and all components of $v(1)$ are positive.
Choose a diagonalizer $g_1$ according to \eqref{L1g1} and define  $v(d)\in \C^n$  by the formula \eqref{vdeq} using an arbitrary $\tau\in \T^n$.
Then equation \eqref{gReq} admits solutions for $g_R$,  the general solution has the form \eqref{gr} with arbitrary $\Gamma \in \T^n$,
and all so
obtained triples $(g_R, L, v)$ belong to $\cS_1$.
\end{lemma}
\begin{proof}
By using the definitions  \eqref{H17} of  $L$  and \eqref{defL1} of $L_1$, we can always recast the moment map constraint \eqref{momconstdag} in the form \eqref{gReq},
which implies the equality of characteristic polynomials
\begin{equation} \label{Eq:LLplusspec}
 \det( e^{2\gamma} L -\lambda \1_n)= \det(L_1 + v(d) v(d)^\dagger-\lambda \1_n).
 \end{equation}
 Since $L$ is diagonal for $(g,L,v)\in \cS_1$,
 we have
 \be
  \det(e^{2\gamma} L-\lambda \1_n) = \prod_j (e^{2\gamma} y_j - \lambda).
  \ee
  By using  \eqref{L1g1} and introducing
  \be
\tilde u:=g_1 v(d),
\label{tildu}\ee
we can write the polynomial on the right-hand side of \eqref{Eq:LLplusspec} as
\begin{equation}
 \begin{aligned}
\det(L_1-\lambda \1_n) [1+v(d)^\dagger (L_1-\lambda \1_n)^{-1} v(d)]
   =&\prod_{k=1}^n (\mu_k-\lambda) \, \left[ 1+ \sum_{j=1}^n \tilde{u}^\dagger_j \frac{1}{\mu_j-\lambda}\tilde{u}_j \right] \\
   =&\prod_{k=1}^n (\mu_k-\lambda) + \sum_{j=1}^{n}|\tilde{u}_j|^2 \prod_{k\neq j} (\mu_k-\lambda)\,.
 \end{aligned}
\end{equation}
 Thus \eqref{Eq:LLplusspec} evaluated at $\lambda=\mu_l$ yields
  \begin{equation} \label{Eq:tildeu}
  |\tilde{u}_l|^2= (e^{2\gamma}y_l-\mu_l)   \prod_{k\neq l} \frac{e^{2\gamma} y_k-\mu_l}{\mu_k-\mu_l} = \bV_l(y,\mu)^2,
 \end{equation}
 which is positive due to \eqref{Eq:ymuL}.
 We conclude from this and equation \eqref{tildu} that $v(d)$ has the form \eqref{vdeq}. The claim \eqref{gr} about the form of $g_R$ follows from \eqref{gReq}
  since $L$ is  diagonal and has distinct eigenvalues.

 The converse statement is proved by utilizing that the equality of the polynomials in $\lambda$ \eqref{Eq:LLplusspec} is \emph{equivalent}
 to the existence of a unitary matrix $g_R$  that solves the constraint equation \eqref{gReq}.
 Then we simply turn the above arguments backwards. The crux is that the spectral assumption  \eqref{Eq:ymuL} ensures the positivity
 of the expression in \eqref{Eq:LL0plusspec3}, whence $v(d)$ can be constructed starting from the vector $\tilde u = \diag(\tau_1,\dots, \tau_n)\bV(y,\mu)$.
\end{proof}

 From now on we  write
\be
v(\alpha)_j = v(\alpha)_j^\Re + \ri v(\alpha)_j^\Im
\quad\hbox{for}\quad
\alpha=2,\dots, d-1,
\label{vdec}\ee
with real-valued $v(\alpha)_j^\Re$, $v(\alpha)_j^\Im$.
In the next statement we summarize how Lemma \ref{Lem:Laci1}  gives us coordinates on $\cS_1$.

\begin{corollary}\label{C:coord}
Via the formulae of Lemma \ref{Lem:Laci1}  for $v(d)$ and $g_R$, the elements of $\cS_1$ are uniquely parametrized by the $2n(d-1)$ variables
\be\label{Eq:coord1}
y_j, v(1)_j, v(\alpha)_j^\Re,  v(\alpha)_j^\Im, \quad j=1,\dots, n, \quad \alpha=2,\dots, d-1\,,
\ee
together with the $2n$ variables
\be\label{Eq:coord2}
\tau_j \in \U(1), \quad \Gamma_j \in \U(1), \quad j=1,\dots, n\,.
\ee
The variables \eqref{Eq:coord1}  take values in an open subset of $\R^{2n(d-1)}$.
The matrix  elements of $g_1$ \eqref{L1g1} can be chosen to be real-analytic functions of the $2n(d-1)$ variables  \eqref{Eq:coord1},
and then the components of $v(d)$ \eqref{vdeq} are also real-analytic functions of these variables and the $\tau_j$.
Likewise, the matrix elements of $g_R^0$ can be chosen to be real-analytic functions of the variables  \eqref{Eq:coord1} and the $\tau_j$.
Consequently, the variables  \eqref{Eq:coord1}  together with $t_j$ and $\gamma_j$ in $\tau_j=e^{\ri t_j}$ and $\Gamma_j = e^{\ri \gamma_j}$
define a coordinate system on the open submanifold of the reduced phase space corresponding to $\cS_1$.
\end{corollary}
\begin{proof}
The variables  \eqref{Eq:coord1} run over an open set simply because the eigenvalues of $L_1$ depend continuously on them.
This dependence is actually analytic since those eigenvalues are all distinct.
Regarding the dependence of $g_1$ and $g_R^0$ on the variables, we use the well-known fact that the eigenvectors of
regular Hermitian matrices can be chosen
as analytic functions of the independent parameters of the matrix elements.
\end{proof}

\subsection{Degenerate integrability}

The reduced integrals of motion arising from \eqref{S1} and \eqref{S2} take a simple form in terms of our coordinates on $\cS_1$.
Relying on this, we shall inspect the following $2n(d-1)$ reduced integrals of motion:
\begin{equation}
\begin{aligned} \label{Eq:FctDim}
  \tr (L^k)=\sum_j y_j^k\,, \quad  &I_{1,1}^k=\sum_j v(1)_j^2 y^{k}_j\,, \\
 \Re[I_{\alpha,1}^k]=\sum_j v(1)_j y^{k}_j v(\alpha)_j^\Re,\quad  &\Im[I_{\alpha,1}^k]=\sum_j v(1)_j y^{k}_j v(\alpha)_j^\Im\,,
\end{aligned}
\end{equation}
where $k=1,\dots, n$ and $\alpha=2,\dots, d-1$, and the additional $2n$ integrals of motion supplied by the real and imaginary parts of
\be\label{Ikd}
I^k_{d,1}= \sum_j v(1)_j y_j^k v(d)_j
\quad\hbox{with}\quad
v(d) = g_1^{-1} \diag(\tau_1,\ldots,\tau_n) \bV.
\ee

\begin{proposition} \label{Pr:indpdt}
The $2n (d-1)$ reduced integrals of motion \eqref{Eq:FctDim}, which include the $n$ reduced Hamiltonians $\tr(L^k)$,
are functionally independent on $\cS_1$.
On each connected component of $\cS_1$, $n$ further integrals of motion may be selected  from the real and
imaginary parts of the functions \eqref{Ikd}
in such a way that together with \eqref{Eq:FctDim} they provide a set of
$2nd - n$ independent functions.
\end{proposition}
\begin{proof}
We are going to prove functional independence by inspection of Jacobian determinants using the coordinates on
$\cS_1$ exhibited in Corollary \ref{C:coord}.
Let us first consider the functions given by \eqref{Eq:FctDim}.
If we order the $2n(d-1)$ functions as written in \eqref{Eq:FctDim} and also order the $2n(d-1)$  coordinates
as written in \eqref{Eq:coord1}, then the
corresponding  Jacobian matrix $J$ takes a block lower-triangular form, with $n\times n$ blocks.
The first diagonal block, $(\partial \tr L^k/\partial y_j)$,  is given by
$Y\in\Mat_{n\times n}(\R)$ with $Y_{kj}= k y_j^{k-1}$, while
all other diagonal blocks are given by
$XD_1$ with $X_{kj}=y^k_j$ and $D_1=\diag(v(1)_1,\ldots, v(1)_n)$, except the second one,
$(\partial I^k_{1,1}/\partial v(1)_j)$, which equals  $2 X D_1$.
 By the definition of $\cS_1$, the coordinates $y_j$ are positive and distinct while the $v(1)_j$
are positive, so that $X$, $Y$ and $D_1$ are invertible. Hence $J$ has rank $2n(d-1)$.

To continue, consider the $2n$ functions
\be
\Re(I_{d,1}^k),\, \Im(I_{d,1}^k) \quad \hbox{with}\quad  1\leq k \leq n.
\label{2nI}\ee
It is clear that any function $G$ taken from \eqref{Eq:FctDim} satisfies $\partial G/\partial t_j=0$.
So our claim will follow if there exists a subset of
$n$ functions  $F_1,\ldots,F_n$ from those in \eqref{2nI} for which the
Jacobian matrix $\left(\frac{\partial F_k}{\partial t_l}\right)_{kl}$ is invertible.

Note from Lemma \ref{Lem:Laci1} and Corollary \ref{C:coord} that
 \begin{equation}
 \frac{\partial v(d)_j}{\partial t_l}=\ic (g_1^{-1})_{jl} e^{\ic t_l} \bV_l\,,
\end{equation}
since $\bV$ and $g_1$ depend only on the variables \eqref{Eq:coord1}.
In particular, the matrix
\begin{equation}
 \frac{\partial v(d)}{\partial t} :=
 \left(\frac{\partial v(d)_j}{\partial t_l}\right)_{1\leq j,l\leq n}
\end{equation}
is invertible, because so are $g_1$, $\diag( e^{\ic t_1},\ldots, e^{\ic t_n})$ and $\diag(\bV_{1},\ldots,\bV_{n})$.

 If the $2n\times n$ real matrix
\begin{equation} \label{Eq:Mat}
 \left(\frac{\partial \Big(\Re(I_{d,1}^k), \Im(I_{d,1}^k)\Big) }{\partial t_l}\right)_{1\leq k,l\leq n}
\end{equation}
has rank $n$, then we are done. Assume by contradiction that  this matrix has rank less than $n$.  From \eqref{Ikd} we have
\begin{equation}
\frac{\partial I_{d,1}^k}{\partial t_l} = \sum_j v(1)_j y_j^k   \frac{\partial v(d)_j}{\partial t_l}\,,
\end{equation}
and therefore we can write the following equality of complex matrices
\begin{equation}
\frac{\partial I_{d,1}}{\partial t} := \left(\frac{\partial I_{d,1}^k }{\partial t_l}\right)_{1\leq k,l\leq n}
 = X \diag(v(1)_1,\ldots,v(1)_n) \frac{\partial v(d)}{\partial t}\,,
\end{equation}
where $X$ is given by $X_{kj}=y_j^k$ as before.  We have already established that all three factors in the above product of matrices
are invertible.
Thus $\partial I_{d,1}/\partial t $ is invertible, hence has rank $n$.

To finish the proof,
it suffices to remark that the complex matrix $\partial I_{d,1}/\partial t$ is a complex linear combination of the rows
of the matrix  given in \eqref{Eq:Mat}.
If the latter matrix has rank strictly less than $n$, then so does $\partial I_{d,1}/\partial t$, which gives a contradiction.
\end{proof}

Let us recall from Proposition \ref{P:BrII} that the unreduced phase space supports the
polynomial Poisson algebra $\Ical$ \eqref{Eq:Ical}, whose Poisson center contains
the polynomial algebra
\begin{equation} \label{Eq:Htr}
 \fH_{\mathrm{tr}}:=\R[\tr L^k, k \geq 0]\,.
\end{equation}
Since these Poisson algebras consist of $\U(n)$ invariant functions, they engender corresponding Poisson algebras
$\Ical^\red$ and  $\fH_{\mathrm{tr}}^\red$ over the reduced phase space $\cM_\red$.
Our final result is a direct consequence of Proposition \ref{Pr:indpdt}.

\begin{theorem}\label{Th:last}
The reduced polynomial algebras of functions $\fH_{\mathrm{tr}}^\red$ and $\Ical^\red$
inherited from $\fH_{\mathrm{tr}}$ \eqref{Eq:Htr} and $\Ical$ \eqref{Eq:Ical}
have functional dimension $n$ and $2nd - n$, respectively.
In particular,  on the phase space  $\cM_\red$ of dimension $2nd$, the Abelian Poisson algebra $\fH_{\mathrm{tr}}^\red$
yields a real-analytic,
degenerate integrable system with integrals of motion  $\Ical^\red$.
\end{theorem}
\begin{proof}
Let us consider $\Ical^\red$ and its Poisson center $\cZ(\Ical^\red)$.
Denote $r$ and $r_0$ the functional dimensions of these polynomial algebras of functions.
Observe from Proposition \ref{Pr:indpdt} that
\be
r\geq (2nd - n) \quad \hbox{and}\quad r_0 \geq n.
\label{ineq1}\ee
The second inequality holds since  $\fH_{\mathrm{tr}}^\red$ is contained in
$\cZ(\Ical^\red)$, and Proposition \ref{Pr:indpdt} implies that the functional dimension of $\fH_{\mathrm{tr}}^\red$ is $n$.

In a neighbourhood $U_0$ of a generic point of $\cM_\red$, we can choose a system of coordinates given by
$2nd$ functions $F_1,\dots, F_{2nd}$ such that
the first $r$ functions belong to $\Ical^\red$, of which the  first $r_0$ belong to $\cZ(\Ical^\red)$.
  In terms of such coordinates,  the Poisson matrix $P=(\br{F_i,F_j})_{i,j}$ can be decomposed into blocks as
\begin{equation}
P=\left( \begin{array}{ccc}
  0_{r_0\times r_0}& 0_{r_0\times(r-r_0)}& B \\
  0_{(r-r_0)\times r_0}& A& \ast \\
 -B^T&\,\,\ast\,\, & \ast
 \end{array} \right).
\end{equation}
This matrix must be non-degenerate since the reduced phase space is a symplectic manifold.
In particular, this implies that the $r_0$ rows of $B$ must be independent.
Then the number of independent columns of $B$ must be also $r_0$, which cannot be bigger than the number of columns.
This gives $r_0 \leq (2nd -r)$, or equivalently
\be
r_0 + r \leq 2nd.
\label{ineq2}\ee
By combining \eqref{ineq1} with \eqref{ineq2}, we obtain that $r_0 = n$  and $r = (2nd -n)$.
\end{proof}

We see from the above proof that $\cZ(\Ical^\red)$ and $\fH_{\mathrm{tr}}^\red$ have the same functional
dimension. Since $\fH_{\mathrm{tr}}^\red \subseteq\cZ(\Ical^\red)$, we expect that these polynomial algebras of functions
actually coincide.

\begin{remark}\label{Rem:actang}
Let us explain that our coordinates on $\cS_1$ are very close to action-angle variables.
To start, we recall that
the joint level surfaces of the integrals of motion of the unreduced free system are compact,
because (with the help of the variables $(g_R, L, W)$) they can be identified with
closed subsets of $\U(n)$. This compactness property is inherited by the reduced system.
If we restrict ourselves to the open subset of the reduced phase space parametrized by $\cS_1$,
then the connected components of the joint level surfaces of the elements of $\Ical^\red$
\eqref{Eq:Ical} are the $n$-dimensional `$\Gamma$-tori' obtained
by fixing all variables in \eqref{Eq:coord1} and \eqref{Eq:coord2} except the $\Gamma_j$.
Both the gauge slice $\cS$ \eqref{Eq:Sspace} and its subset
$\cS_1$ are invariant under the flow \eqref{T36}  of the Hamiltonian $H_k:= \frac{1}{2k} \tr(L^k)$,
for every $k=1,\dots, n$,  which gives
the following linear flow on the $\Gamma$-torus:
\be
\Gamma_j(t) =  \exp( \ri y_j^k t) \Gamma_j^0,\quad \forall j=1,\dots, n,
\label{linflow}\ee
where $\Gamma_j^0$ refers to the initial value.
This statement holds since for $H_k(L) \equiv h_k(b_R)$ one has
 $D h_k(b_R)  = \ri L^k$.
The flow \eqref{linflow} entails that on $\cS_1$ the variables
${\hat p}_j:=\frac{1}{2} \log y_j$  are canonical  conjugates to the angles $\gamma_j$ in
$\Gamma_j = e^{\ri \gamma_j}$, i.e.,
they satisfy  $\{ \gamma_j, {\hat p}_l\}_\red =  \delta_{jl}$.
\end{remark}

\section{Conclusion} \label{sec:Con}

In this paper we investigated a trigonometric real form of the spin
RS system \eqref{I2} introduced originally by  Krichever and Zabrodin \cite{KZ}
and studied subsequently in \cite{AO,CF2} in the complex holomorphic setting.
We have shown that this real form arises from Hamiltonian reduction
of a free system on a spin extended Heisenberg double of the $\U(n)$ Poisson--Lie group, and
exploited the reduction approach for obtaining a detailed characterization of its
main features.
In particular,
we presented two models of dense open subsets of the reduced phase space
where the system lives.  The model developed in \S\ref{ssec:prim} led to
an elegant description of the reduced symplectic form (Theorem \ref{Th:redsymp}), while
the equations of motion and the corresponding Hamiltonian are complicated in the pertinent
variables based on the
`primary spins'. On the other hand, the model studied in \S\ref{ssec:Hbis} and in Sections \ref{sec:N} and  \ref{sec:RP}
allowed us to recover the spin RS equations of motion \eqref{I8} from the projection
of a free flow (Corollary \ref{Cor:vf}), but the reduced Poisson brackets (Theorem \ref{Thm:redPB}) take  a relatively
complicated form in the underlying  `dressed spin' variables.
In our framework the solvability of the evolution equations by linear
algebraic manipulations emerges naturally (Remark \ref{Rem:sol}), and we also proved their
degenerate integrability by explicitly exhibiting the
required number of constants of motion (Theorem \ref{Th:last}).

A basic ingredient of the unreduced phase space that we started with was a $\U(n)$ covariant Poisson structure  on
$\C^n \simeq \R^{2n}$ that goes back to Zakrzewski \cite{Z},
for which we found the corresponding moment map
(Proposition \ref{Pr:A3}) and symplectic form (Proposition \ref{Pr:A5}).

We finish by highlighting a few open problems related to our current research.
As always in the reduction treatment of an integrable Hamiltonian system, one should
gain as complete an understanding of the global structure of the reduced phase space as possible.
The basic point is that the projections of free flows are automatically complete, but only on the full reduced phase space.
In the present case,
one should actually construct two global models of the reduced phase space: one fitted to the
system that we have studied, and another one that should be associated with its action-angle
dual.   Without going into details,  we refer to the literature \cite{FM, JHEP,Res1, RBanff} where it
is explained that  the integrable many-body systems  usually
come in dual pairs, and the same holds for their several spin extensions.
In our case, the commuting Hamiltonians of the dual system are expected to arise from the reduction
of the Abelian Poisson algebra $\hat \fH = \Xi_R^*(C^\infty(\U(n)))$, which is in some sense
dual to $\fH$ \eqref{T34} on which our system was built.

It could be interesting to explore generalizations of the construction employed in our study.
For example, one may obtain new variants of the trigonometric
spin RS model by
replacing some or all of the primary spins $w^\alpha$ by $z^\alpha$ subject to the Poisson
bracket described at the end of  Appendix \ref{sec:A} (Remark \ref{Rem:minZak}).
Generalization of our reduction in which the Heisenberg double is replaced by
a quasi-Hamiltonian double of the form $\U(n)\times \U(n)$  \cite{AMM}, and the primary spins are also modified suitably,
should
lead to compactified spin RS systems.
It should be possible to uncover a reduction picture behind the hyperbolic real form
of the spinless and spin RS models, too.
All these issues, as well as the questions of quantization and the reduction approach
to elliptic spin RS models, pose challenging problems for future work.

\bigskip
\noindent
{\bf Acknowledgements} We wish to thank J.~Balog for useful remarks on the manuscript.
We are also grateful to  A.~Alekseev and R.~Sjamaar
for  advise regarding proper moment maps.
The research of M.F. was partly supported by the travel grant ECR-1819-01 of the London Mathematical Society,
and a Rankin-Sneddon Research Fellowship of the University of Glasgow.
The work of L.F. was supported in part by the NKFIH research grant K134946.

\bigskip

\newpage
 \appendix
\section{Properties of the primary spin variables}
\label{sec:A}

In this appendix  we first elaborate the properties of the primary spin variables that
were summarized in Proposition \ref{Pr:Zak}. As was already mentioned, the pertinent
Poisson structure on $\C^n \simeq \R^{2n}$ is a special case of the $\U(n)$ covariant Poisson structures
due to Zakrzewski \cite{Z}.  Nevertheless, to make our text self-contained, we shall also
verify its Jacobi identity and covariance property. Then we present the corresponding moment
map and symplectic form, which have not been considered in previous work.

For any real function $F\in C^\infty(\C^n)$, we define its $\C^n$-valued `gradient' $\nabla F$ by
the equality\footnote{This is a \emph{symplectic} gradient associated with
the standard symplectic form,
$\omega(\xi,\eta)=\Im(\xi^\dagger \eta)$,  on
$\C^n \simeq \R^{2n}$.}
\be
\Im\left( (\nabla F(w))^\dagger V\right):= \dt F(w + t V),
\qquad
\forall w, V\in \C^n,
\ee
where the elements of $\C^n$ are viewed as column vectors.
We note that any real linear function on the real vector space $\C^n$ is of the form
\be
F_\xi(w):= \Im (\xi^\dagger w),
\ee
for some $\xi \in \C^n$, and for such function $\nabla F_\xi = \xi$.
Next we give a convenient presentation of Zakrzewski's Poisson bracket.

\begin{proposition} \label{Pr:A1}
For real functions $F, H\in C^\infty(\C^n)$,  let $\xi(w):= \nabla F(w)$ and $\eta(w):= \nabla H(w)$.
Then the following formula
\be\label{pb1}
\{F, H\}(w) = \Im\left( \xi(w)^\dag(w \eta(w)^\dag)_{\u(n)} w
 - \half\xi(w)^\dag w\eta(w)^\dag w - \half\xi(w)^\dag w w^\dag\eta(w) -
\xi(w)^\dag\eta(w) \right),
\ee
where the notation \eqref{T4} is used,
defines a Poisson bracket on $C^\infty(\C^n)$.
Equivalently to the formula \eqref{pb1}, the Hamiltonian vector field $V_H$ associated with $H\in C^\infty(\C^n)$
is given by
\be
V_H(w) = (w \eta(w)^\dag)_{\u(n)} w - \eta(w) -\half (\eta(w)^\dag w + w^\dag\eta(w))w,
\qquad
\eta(w)= \nabla H(w).
\label{HamZak}\ee
Extending the real Poisson bracket to complex functions by complex bilinearity,
the Poisson brackets of the component functions $w\mapsto w_i$ satisfy the explicit formulae \eqref{T29} and \eqref{T30}.
\end{proposition}

\begin{proof}
The antisymmetry of the  last two terms of \eqref{pb1} is obvious, while
the antisymmetry of the sum of the first and second terms is seen from the identity
\be
\Im\left( \xi^\dag(w \eta^\dag)_{\u} w
 -  \half\xi^\dag w\eta^\dag w\right)=
  \half \Im \,\tr \left(   (w \eta^\dag)_{\u} (w \xi^\dag)_\b -(w \xi^\dag)_{\u} (w \eta^\dag)_\b \right),
 \ee
where we used constant $\xi$ and $\eta$ for simplicity.
Here and below, the subscripts    $\u$ and $\b$ stand for $\u(n)$ and $\b(n)$.

Regarding the Jacobi identity, it is enough to verify it for
linear functions $F_\xi$, $F_\eta$ and $F_\zeta$  for arbitrary $\xi,\eta,\zeta \in \C^n$.
In this verification we may use the formula \eqref{HamZak},
since  this
expresses the identity $\{F,H\}(w) = \Im (\xi(w)^\dagger V_H(w))$, and does not rely on the Jacobi identity.

We start by calculating  the gradient of $\{ F_\xi, F_\eta\}$ from \eqref{pb1}, and find
\be
\left(\nabla \{ F_\xi, F_\eta\}(w)\right)^\dagger =
\xi^\dag(w\eta^\dag)_\k - \eta^\dag(w\xi^\dag)_\k + \half\eta^\dag w\xi^\dag - \half\xi^\dag w \eta^\dag -
\half (w^\dag\eta)\xi^\dag + \half (w^\dag\xi)\eta^\dag.
\ee
Combining this with $V_{{F_\zeta}}(w)$
from \eqref{HamZak}, we have to inspect
\be
\ba
{\mathcal{J}}(w): &= \{\{ F_\xi , F_\eta\}, F_\zeta\}(w) + \hbox{cycl. perm.} \\
&=\Im \left[ \left(\nabla \{ F_\xi, F_\eta\}(w)\right)^\dagger
\left( (w\zeta^\dag)_\k w - \zeta - \half(\zeta^\dag w)w - \half(w^\dag\zeta)w\right)\right] + \hbox{c.p.}
\ea
\ee
By spelling this out, we obtain
$$
\ba
{\mathcal{J}(w)}&=
\Im\Bigl[ \eta^\dag(w\xi^\dag)_\k\zeta - \xi^\dag(w\eta^\dag)_\k\zeta + \half(\xi^\dag w)\eta^\dag\zeta -
\half(\eta^\dag w)\xi^\dag\zeta + \half(w^\dag\eta)\xi^\dag\zeta - \half(w^\dag\xi)\eta^\dag\zeta + \hbox{c.p.}\Bigr] \\
&+
\Im\Bigl[\xi^\dag(w\eta^\dag)_\k(w\zeta^\dag)_\k w - \eta^\dag(w\xi^\dag)_\k(w\zeta^\dag)_\k w
+ \half\eta^\dag w\xi^\dag(w\zeta^\dag)_\k w - \half\xi^\dag w\eta^\dag(w\zeta^\dag)_\k w \\
&\qquad\qquad\qquad
+ \half(w^\dag\xi)\eta^\dag(w\zeta^\dag)_\k w  -  \half(w^\dag\eta)\xi^\dag(w\zeta^\dag)_\k w +\hbox{ c.p.} \\
&\qquad -
\half(\zeta^\dag w)\xi^\dag(w\eta^\dag)_\k w + \half(\zeta^\dag w)\eta^\dag(w\xi^\dag)_\k w
- \quarter (\eta^\dag w)(\xi^\dag w)(\zeta^\dag w) + \quarter (\xi^\dag w)(\eta^\dag w)(\zeta^\dag w) \\
&\qquad\qquad\qquad
+ \quarter (w^\dag\eta)(\xi^\dag w)(\zeta^\dag w) - \quarter (w^\dag\xi)(\eta^\dag w)(\zeta^\dag w)
+ \hbox{c.p.} \\
&\qquad -
\half(w^\dag\zeta)\xi^\dag(w\eta^\dag)_\k w + \half(w^\dag\zeta)\eta^\dag(w\xi^\dag)_\k w
- \quarter (w^\dag\zeta)(\eta^\dag w)(\xi^\dag w) + \quarter (w^\dag\zeta)(\xi^\dag w)(\eta^\dag w) \\
&\qquad\qquad\qquad
+ \quarter (w^\dag\eta)(w^\dag\zeta)(\xi^\dag w) - \quarter (w^\dag\xi)(w^\dag\zeta)(\eta^\dag w)
+ \hbox{c.p.}\Bigr]
\ea
$$
After making several self-evident cancellations, and using cyclic permutations to reorganize terms
in a convenient way, we get
\[
\ba
{\mathcal{J}(w)}=&\Im
\Bigl(\zeta^\dag(w\eta^\dag)_\k\xi - \xi^\dag(w\eta^\dag)_\k\zeta + \half(\xi^\dag w - w^\dag\xi)\eta^\dag\zeta -
\half(\eta^\dag w - w^\dag\eta)\xi^\dag\zeta + \hbox{c.p.}\Bigr)\\
+&\Im\,\tr
\Bigl(w\xi^\dag\bigl[(w\eta^\dag)_\k,(w\zeta^\dag)_\k\bigr] + (\eta^\dag w) w\xi^\dag (w\zeta^\dag)_\k  -
(\xi^\dag w) w\eta^\dag (w\zeta^\dag)_\k  + \hbox{c.p.}\Bigr).
\ea
\]
It is not difficult to see that the first line gives zero. Rearranging the second line, we have
\[
\ba
{\mathcal{J}(w)}&=\Im\,\tr
\Bigl(w\xi^\dag\bigl[(w\eta^\dag)_\k,(w\zeta^\dag)_\k\bigr] + w\eta^\dag w\xi^\dag (w\zeta^\dag)_\k  -
 w\xi^\dag w\eta^\dag (w\zeta^\dag)_\k  + \hbox{c.p.}\Bigr) \\
& =
\Im\,\tr
\Bigl(-w\xi^\dag\bigl[(w\eta^\dag)_\b,(w\zeta^\dag)_\k\bigr] +  \hbox{c.p.}\Bigr)\\
& =
- \Im\,\tr\Bigl((w\xi^\dag)_\k\bigl( [(w\eta^\dag)_\b,(w\zeta^\dag)_\k] + [w\eta^\dag,(w\zeta^\dag)_\b]\bigr)
\\
& \qquad\qquad\qquad
+
(w\xi^\dag)_\b\bigl( [(w\eta^\dag)_\b,(w\zeta^\dag)_\k] + [(w\eta^\dag)_\k, w\zeta^\dag]\bigr)\Bigr)\\
& =
-\Im\,\tr\Bigl(
(w\xi^\dag)_\k\bigl( [(w\eta^\dag)_\b,(w\zeta^\dag)_\k] + [(w\eta^\dag)_\k,(w\zeta^\dag)_\b] +
[(w\eta^\dag)_\b,(w\zeta^\dag)_\b]\bigr)\\
& \qquad\qquad\qquad
+
(w\xi^\dag)_\b\bigl( [(w\eta^\dag)_\b,(w\zeta^\dag)_\k] + [(w\eta^\dag)_\k,(w\zeta^\dag)_\k] +
 [(w\eta^\dag)_\k,(w\zeta^\dag)_\b]\bigr)\Bigr)\\
& =
\Im\,\tr \bigl(w\xi^\dag [ w\eta^\dag, w\zeta^\dag]\bigr) +\hbox{cycl. perm.} =0.\\
\ea
\]

Having verified the Jacobi identity, it remains to calculate the Poisson brackets of the components of $w$ and their
complex conjugates.
Let $e_k$ $(k=1,\dots, n)$ denote the canonical basis of $\C^n$.
One obtains by  tedious calculation that the Hamiltonian vector fields of the linear
functions given by the real and imaginary parts of the components $w_k$ have the following form:
\[
V_{\Re w_k}(w) =
\left\{\ba
&
\ri \Re(w_k) w_k \mk + \ri \sum_{r>k}(w_kw_r\mr + |w_r|^2\mk) + \ri\mk - \half\ri(w_k-\overline w_k)w\qquad k<n,\\
& \ri \Re(w_n)w_n\mn + \ri\mn - \half\ri(w_n-\overline w_n)w
\quad \qquad\qquad\qquad\qquad\qquad\qquad k=n.
\ea\right.
\]
and
\[
V_{\Im w_k}(w) =
\left\{\ba
&
\ri \Im(w_k) w_k\mk + \sum_{r>k}(w_kw_r\mr - |w_r|^2\mk) - \mk - \half (w_k+\overline w_k)w\qquad k<n,\\
& \ri \Im( w_n) w_n\mn - \mn - \half (w_n+\overline w_n)w
\,\,\,\, \qquad\qquad\qquad\qquad\qquad\qquad k=n.
\ea\right.
\]
By using these, one can  check that the formulae \eqref{T29} and \eqref{T30} follow.
If desired, the reader can supply the details.
\end{proof}

\medskip
The bracket \eqref{pb1} has the nice property that the natural action of
$\U(n)$ on $\C^n$ is Poisson \cite{Z}, and
this can also be checked using linear functions $F_\xi$.
To this end, for any $g\in \U(n)$ and $w\in \C^n$ we define the
functions $F_\xi(g\,\cdot\,)\in C^\infty(\C^n)$ and $F_\xi(\,\cdot\,w)\in C^\infty(\U(n))$ by
\be
F_\xi(g\,\cdot\,)(w) = F_\xi(gw) = F_\xi(\,\cdot\,w)(g).
\ee
Then an easy calculation gives that
\be
\{F_\xi,F_\eta\}(gw) - \{F_\xi(g\,\cdot\, ),F_\eta(g\,\cdot\, )\}(w)
\ee
is equal to
\be
\Im\tr\bigl(gw\xi^\dag(gw\eta^\dag)_{\u(n)} - w\xi^\dag g (w\eta^\dag g)_{\u(n)}\bigr),
\ee
which in turn is equal to the value at $g$
of the Poisson bracket \eqref{T10} of the functions $F_\xi(\,\cdot\, w)$ and $F_\eta(\,\cdot\,w)$ on $\U(n)$.
The last equality follows using $D F_\xi(\,\cdot\,w)(g)= (g w \xi^\dag)_{\b(n)}$ and elementary manipulations.
Thus, we have
\be
\{F_\xi,F_\eta\}(gw) = \{F_\xi(g\,\cdot\, ),F_\eta(g\,\cdot\, )\}(w) + \{F_\xi(\,\cdot\,w),F_\eta(\,\cdot\,w)\}_U(g),
\ee
which means that the map $\U(n) \times \C^n \ni (g,w) \mapsto gw \in \C^n$ is indeed a Poisson map.

Let us recall the diffeomorphism
\be
b \mapsto b b^\dagger
\label{BtoP}\ee
from the group $\B(n)$ to the space $\fP(n)$ of positive definite Hermitian matrices.
By this diffeomorphism, the Poisson structure \eqref{T11} on $\B(n)$ is converted into a Poisson structure
on $\fP(n)$, which is given by the first term of \eqref{+PB2}, i.e.
\be
\{f, h\}_\fP (L)=4 \langle L d f(L), \left(L d h (L)\right)_{\u(n)} \rangle
\label{pb2}\ee
for all $f,h \in C^\infty(\fP(n))$.
Here,   the $\u(n)$-valued
derivatives $d f$ and $dh$ are defined by \eqref{T22}.

\begin{proposition} \label{Pr:A2}
With respect to the brackets \eqref{pb1} and \eqref{pb2}, the map
\be
\Phi:w\mapsto \1_n+ ww^\dag
\label{Phimap}\ee
from $\C^n$ to $\fP(n)$ is Poisson.
\end{proposition}

\begin{proof}
Let $X,Y\in\u(n)$ and consider the pull-backs $\Phi^*(f_X)$ and $\Phi^*(f_Y)$ of the functions
$f_X(L):=\langle X,L\rangle$ and $f_Y(L) := \langle Y,L\rangle$. We have
\be
\Phi^*(f_X)(w) = \Im(w^\dag Xw) + \Im\tr(X)\quad\hbox{and}\quad \Phi^*(f_Y)(w) = \Im(w^\dag Yw) + \Im\tr(Y).
\ee
Using the formula \eqref{pb1} with $(\nabla \Phi^* (f_X))(w) = - 2 Xw$ and  similar for $f_Y$,
we can compute
\bea
&\{\Phi^*(f_X),\Phi^*(f_Y)\}(w) = 4 \Im \Bigl( w^\dag X(ww^\dag Y)_\u w + w^\dag X Yw - \half w^\dag Xww^\dag Yw +
\half w^\dag X ww^\dag Y w\Bigr) \nonumber \\
 &=
4\Im\tr \Bigl( (\1_n+ww^\dag)X\bigl((\1_n +ww^\dag)Y\bigr)_{\u(n)} \Bigr) = \{ f_X, f_Y\}_\fP(\Phi(w)).
\eea
Here, we have taken into account that, for example, $\Im\tr(XY) =0$ for $X,Y\in \u(n)$.
The statement follows since the linear functions of the form $f_X$ can serve as coordinates on $\fP(n)$.
\end{proof}

Let $\bb: \C^n \to \B(n)$ be the map determined by the condition
\be
\Phi = \bb \bb^\dagger.
\label{Phibb}\ee
It follows from Proposition \ref{Pr:A2} that this is a Poisson map with respect to the Poisson brackets \eqref{pb1} on $\C^n$
and \eqref{T11} on $\B(n)$.

\begin{proposition} \label{Pr:A3}
The map $\bb$ defined by \eqref{Phibb} with \eqref{Phimap}
is the moment map for the Poisson action \eqref{T26} of $\U(n)$ on $\C^n$.
According to \eqref{T27}, this means that we have
\be
 \Im\left( (\nabla F(w))^\dagger Xw\right) = \Im \tr \left(X \{F, \bb \}(w) \bb(w)^{-1}\right),
\quad \forall X\in \u(n),\,w\in \C^n,\, F\in C^\infty(\C^n).
\label{propA3}\ee
\end{proposition}

\begin{proof}
For ease of notation, we verify the relation for linear functions $F_\xi$ on $\C^n$, which is sufficient.
For this, we have to calculate the $\b(n)$-valued function
\be
\beta_F := \{ \bb, F\} \bb^{-1}\,, \quad F:=F_\xi.
\ee
Since \eqref{BtoP} is a diffeomorphism, $\beta_F$ is uniquely determined by
\be
\{ \Phi, F\} = \beta_F \Phi + \Phi \beta_F^\dagger,
\ee
and this can be calculated as follows.
First, we rearrange  the expression \eqref{HamZak} of the Hamiltonian vector field
in the form
\be
V_F(w)= \half (\xi^\dag w - w^\dag\xi)w -\xi - (w\xi^\dag)_{\b(n)} w.
\ee
Then, as $(\xi^\dag w - w^\dag\xi)\in\ri\R$, we obtain
\be
\ba
\{ \Phi, F\}(w)  &= V_F(w) w^\dag + w (V_F(w))^\dag \\
&=-(w\xi^\dag)_{\b(n)} ww^\dag - ww^\dag(w\xi^\dag)_{\b(n)}^\dag - \xi w^\dag - w\xi^\dag\\
&=
-(w\xi^\dag)_{\b(n)}\Phi(w) - \Phi(w)(w\xi^\dag)_{\b(n)}^\dag + \bigl((w\xi^\dag)_{\b(n)}-w\xi^\dag\bigr) +
\bigl((w\xi^\dag)_{\b(n)}^\dag -\xi w^\dag\bigr).
\ea
\ee
But the last two terms cancel,  and hence we see that
\be
\beta_F(w) = - \left( w \xi^\dagger\right)_{\b(n)}.
\ee
By using this,  the right-hand-side of \eqref{propA3} becomes
\be
- \Im \tr ( X \beta_F(w)) = \Im \tr (X w \xi^\dagger) = \Im (\xi^\dagger X w),
\ee
whereby the proof is complete.
\end{proof}

\begin{remark}
We had no need for the explicit formula of $\bb(w)$ in the above, but in some other calculations
it is needed. The reader can verify directly that it obeys equation \eqref{bb}.
\end{remark}

\begin{remark} \label{Rem:Tact}
The maximal torus $\T^n < \U(n)$ is a Poisson subgroup with vanishing Poisson bracket, and therefore
the restriction
of the $\U(n)$ action to $\T^n$ gives an ordinary Hamiltonian action. One can identify the dual
Poisson--Lie group
of $\T^n$ with $\B(n)_0$, the group of positive diagonal matrices, with zero Poisson bracket.
Then the corresponding group valued moment map is provided by $w \mapsto \bb(w)_0$,
which is the diagonal part of $\bb(w)$. Writing
\be
\bb(w)_0 = \exp (\phi(w)),
 \ee
 we get the ordinary moment map $w \mapsto \phi(w) \in \b(n)_0$,
 where $\b(n)_0$ (the space of real diagonal matrices) is identified with the linear dual of $\u(n)_0$.
\end{remark}

The following proposition represents one of the  side results of the paper.

\begin{proposition} \label{Pr:A5}
The Poisson bracket \eqref{pb1} is symplectic and,
 with $\cG_j= 1 + \sum_{k=j}^n \vert w_j\vert^2$,
 the corresponding symplectic form on $\C^n$ is given by
\be
\Omega_{\C^n} =
  \frac\ri2\sum_{k=1}^n\frac1{\cG_k} dw_k\wedge d\overline w_k
+ \frac\ri4\sum_{k=1}^{n-1}\frac1{\cG_k\cG_{k+1}}
d\cG_{k+1}
  \wedge
 \left(\overline w_k dw_k - w_k d\overline w_k\right).
 \ee
\end{proposition}

 \begin{proof}
 We start from the coordinate form of the Poisson bracket, copied here for convenience:
 \be\label{wPB}
 \ba
 \{w_i,w_k\} &= \ri\,\sgn(i-k)w_iw_k\\
 \{w_i,\overline w_l\} &= \ri\delta_{il}(2+|w|^2) + \ri w_i\overline w_l + \ri\delta_{il}\sum_{r=1}^n\sgn(r-l) |w_r|^2.
 \ea
 \ee
  We shall first invert the Poisson tensor
 on the dense open submanifold on which all $\vert w_j \vert >0$, where we can use
  the parametrization $w_j = e^{\ri \varphi_j} \vert w_j\vert$.

 Let us consider
\be
\{w_i, |w_k|^2\} = \ri\,\sgn(i-k)|w_k|^2w_i +\ri|w_k|^2w_i +\ri\delta_{ik}(2+|w|^2)w_i +\ri\delta_{ik}w_i\sum_{r=1}^n\sgn(r-k)|w_r|^2,
\label{A26}\ee
from which we easily obtain
\be
\{|w_i|^2,|w_k|^2\}=0.
\label{A27}\ee
Using this, and restricting now to our submanifold, the relation \eqref{A26} implies
\be
\{e^{\ri\varphi_i},|w_k|^2\} = \{ \frac{w_i}{|w_i|} , |w_k|^2 \} =
\ri[1-\delta_{ik}+\sgn(i-k)]|w_k|^2 e^{\ri\varphi_i} + 2\ri\delta_{ik} \cG_k e^{\ri\varphi_i}.
\label{A28}\ee
Plainly, we have the identity
\be
\{ w_j , w_k \} + e^{2\ri \varphi_j} e^{2\ri \varphi_k} \{ \overline w_j , \overline w_k \} =
2 \vert w_j w_k \vert \{  e^{\ri \varphi_j} ,  e^{\ri \varphi_k} \}.
\label{A29}\ee
The left-hand side can be checked to vanish, and thus we get
\be
\{  e^{\ri \varphi_j} ,  e^{\ri \varphi_k} \}=0.
\label{A30}\ee

 It is convenient to change variables, noting that the map $(|w_1|^2,\dots,|w_n|^2) \mapsto (\cG_1,\dots,\cG_n)$ is invertible. Then,
 it is elementary to derive from \eqref{A28} the relation
 \be
 \{e^{\ri\varphi_i}, \cG_k\} =
 \left\{\ba &2\ri \cG_k e^{\ri\varphi_i},\quad k\leq i\\ &0,
\ \ \ \quad\qquad k>i\ea\right.
 \ee
 that can be also written  as
 \be
 \{ \varphi_i, \ln \cG_k\} = \left\{\ba &2,\quad k\leq i\\ &0,\quad k>i \ea \right.
 \ee
 This means that the matrix of Poisson brackets, in the variables $ \varphi_i, \ln \cG_k $ has the form
 \be
 P = 2\left(\begin{array}{cc}0 & A \\ -A^T & 0\end{array}\right)
 \ee
 with
 \be
 A = \1_n + B + B^2 + \dots+B^{n-1},
 \ee
 where $B$ is the nilpotent matrix having the entries $B_{ik} = \delta_{i,k+1}$.
Both $A$ and $P$ are invertible, and their inverses are
\be
   A^{-1}=\1_n-B \quad\hbox{and}\quad
   P^{-1}  = \frac{1}{2} \left(\begin{array}{cc}0 & -(A^{-1})^T \\ A^{-1} & 0\end{array}\right).
\ee
Consequently, we obtain the symplectic form\footnote{In our convention  the wedge  does not contain $\half$ and
$dH = \Omega(\,\cdot \,, V_H)$ with the Hamiltonian vector field $V_H$.}
 ($x^\alpha$ represent the local variables  $\varphi_i$ and $\ln \cG_k$)
\be
\Omega = \frac{1}{2} \sum_{\alpha,\beta=1}^{2n}(P^{-1})_{\alpha\beta} dx^\alpha\wedge dx^\beta\\
=
\half\sum_{k=1}^{n-1} [d\ln \cG_k-d\ln \cG_{k+1}]\wedge d\varphi_k + \half d\ln \cG_n\wedge d\varphi_n.
\label{Om1}\ee

If we now substitute the  identities
 \be
 d \ln \cG_k - d \ln \cG_{k+1}= \frac{ \cG_{k+1} d \vert w_k \vert^2 - \vert w_k\vert^2 d \cG_{k+1}}{\cG_k \cG_{k+1}}
 \ee
 and
 \be
 d\varphi_k = (2\ri |w_k|^2)^{-1}(\overline w_k dw_k - w_kd\overline w_k),
 \ee
 then $\Omega$ \eqref{Om1} takes the form
\be
\Omega =
  \frac\ri2\sum_{k=1}^n\frac1{\cG_k} dw_k\wedge d\overline w_k
+ \frac\ri4\sum_{k=1}^{n-1}\frac1{\cG_k\cG_{k+1}} d \cG_{k+1}
  \wedge
 \left(\overline w_k dw_k - w_k d\overline w_k\right).
 \label{Om2}\ee
 It is clear that both the original Poisson tensor corresponding to \eqref{wPB}
 and $\Omega$ \eqref{Om2} are regular over the whole of $\C^n$.
 As a result, their inverse relationship extends from the dense open submanifold (where $|w_j|>0$ for all $j$) to the full phase space.
 \end{proof}

\begin{remark}\label{Rem:minZak}
The image of the map $w\mapsto \1_n + w w^\dagger$ is the union of the $\U(n)$ orbits in $\fP(n)$ passing through
the degenerate diagonal matrices
\be
\diag (1 + R^2, 1\dots, 1), \qquad R\geq 0.
\ee
For any fixed $R>0$, the orbit is a symplectic leaf in $\fP(n)$ of dimension $2(n-1)$;
$R=0$ corresponds to a trivial symplectic leaf.
The union of the orbits consisting of the conjugates of the matrices
\be
\diag (1 - r^2, 1,\dots, 1),
\qquad
0 \leq  r < 1
\ee
is the image of the  map
\be
z \mapsto \1_n - z z^\dagger
\label{map-}\ee
from
\be
\cB(1):= \{ z \in\C^n \mid |z|^2 <1 \}
\ee
to $\fP(n)$. In fact, the open ball $\cB(1)$, identified as a subset of $\R^{2n}$, can be equipped with the Poisson bracket
\be\label{pb-}
 \ba
 \{z_i,z_k\} &= \ri\,\sgn(i-k)z_iz_k\\
 \{z_i,\overline z_l\} &= \ri(|z|^2 -2)\delta_{il} + \ri z_i\overline z_l + \ri\delta_{il}\sum_{r=1}^n\sgn(r-l)|z_r|^2
 \ea
 \ee
with respect to which the map \eqref{map-} is Poisson. This is also a special case of the Poisson structures
found in \cite{Z}.
 The analogue of Proposition \ref{Pr:A3}  holds for the map
$\bb_-: \cB(1) \to \B(n)$ defined by
\be
\1_n - z z^\dagger = \bb_-(z) \bb_-(z)^\dagger.
\label{bb-}\ee
The Poisson map $\bb_-$ can be used to introduce variants of our reduction.
Concretely, one may replace one or more of the $\bb$ factors in \eqref{H1} by $\bb_-$, and study the reduced system.
The restriction $\gamma >0$ in the moment map constraint \eqref{momconst} then might not be necessary.
Let us also note that one obtains a Poisson pencil on $\C^n$ if one
replaces the last term of \eqref{pb1} by $-\lambda \Im \left(\xi(w)^\dagger \eta(w)\right)$ for any real parameter $\lambda$,
and the formula  \eqref{pb-} corresponds to $\lambda=-1$.
\end{remark}

\section{Proof of Lemma \ref{L:halfv}} \label{A:Glob2}

In this section, we work over $(\C^{n\times d}, \{\ ,\ \}_\cW)$ with the primary spins $(w^\alpha)$, see \S\ref{ssec:Prim}.
We set $\br{\ , \ }:=\{\ ,\ \}_\cW$ to simplify notations.

As noted in \S\ref{ssec:H}, the half-dressed spins $v^\alpha$ can be defined in $\C^{n\times d}$ in terms of the primary spins.
It is convenient to introduce the matrices $b_\alpha=\bb(w^\alpha)$ and $B^\alpha=b_1\cdots b_\alpha$, so that
\be \label{Eq:hvA}
v^\alpha = B^{\alpha-1} w^\alpha\,.
\ee
Remark that $B^\alpha$ is related to the matrix $B_\alpha$ introduced in \eqref{H2} by $B_\alpha = b_R B^\alpha$.
We also note the following lemma, which follows from Proposition \ref{Pr:Zak} by straightforward computations.
\begin{lemma} \label{Lemb1}
 For any $1\leq \alpha \leq d$, $1\leq i , j , l \leq n$,
 \bea
  &&  \br{w^\alpha_i,(b_\alpha)_{jl}}=\,\ic \,[\delta_{ij}+2 \delta_{(i>j)}]\,w^\alpha_j (b_\alpha)_{il}  \,,
  \label{Eq:wb1}\\
&& \br{\overline{w}^\alpha_i,(b_\alpha)_{jl}}=-\ic \delta_{ij} \overline{w}^\alpha_i (b_\alpha)_{il}
-2\ic \delta_{ij} \sum_{k=j+1}^l \overline{w}^\alpha_k (b_\alpha)_{kl} \,.
 \label{Eq:wb2}
 \eea
 Furthermore, the Poisson bracket evaluated on $((b_\alpha)_{ij},(\overline{b}_\alpha)_{ij})$ is given by \eqref{T7}--\eqref{T8}.
\end{lemma}

Next, we need to describe the Poisson brackets between the matrix entries of $(B^\alpha,w^\alpha)$,
which appear in the decomposition \eqref{Eq:hvA}.
To write them down, we introduce the matrices
\be
B^{\alpha;\gamma}=b_\alpha \cdots b_\gamma\,, \quad 1\leq \alpha \leq \gamma \leq d\,,
\ee
which are such that $B^{1;\alpha}=B^\alpha$ and $B^{\alpha;\alpha}=b_\alpha$. We also set $B^{\alpha+1;\alpha}:=\1_n$ and $B^0:=\1_n$.

\begin{lemma} \label{LemwB1}
For any $1\leq \alpha,\beta \leq d$, $1\leq i ,k,l \leq n$,
\bea
&& \br{w^\alpha_i,B^\beta_{kl}}= - \ic  \delta_{(\alpha \leq \beta)}  B^{\alpha-1}_{ki} w_i^\alpha B^{\alpha;\beta}_{il}
   +2\ic  \delta_{(\alpha \leq \beta)}  \sum_{k' \leq i} B^{\alpha-1}_{kk'} w_{k'}^\alpha B^{\alpha;\beta}_{il} \,, \\
&& \br{\overline{w}^\alpha_i,B^\beta_{kl}}=
  +\ic \delta_{(\alpha \leq \beta)}  B^{\alpha-1}_{ki} \overline{w}_i^\alpha B^{\alpha;\beta}_{il}
- 2\ic  \delta_{(\alpha \leq \beta)}   B^{\alpha-1}_{ki} \sum_{i\leq u} \overline{w}_u^\alpha B^{\alpha;\beta}_{ul} \,.
\eea
\end{lemma}
\begin{proof}
 By construction, for $\beta \neq \alpha$ we have $\br{w^\alpha_i,w^\beta_k}=0$, hence $\br{w^\alpha_i,(b_\beta)_{kl}}=0$. We get that
\be
 \br{w^\alpha_i,B^\beta_{kl}}=0,\,\, \alpha<\beta\,; \quad
 \br{w^\alpha_i,B^\beta_{kl}}=\sum_{k \leq l' \leq l} \br{w^\alpha_i,B^\alpha_{kl'}}B^{\alpha+1;\beta}_{l'l}\,, \,\, \beta > \alpha\,.
\ee
When $\beta=\alpha$, $\br{w^\alpha_i,(b_\alpha)_{kl}}$ is given by \eqref{Eq:wb1} and we get
\be
\br{w^\alpha_i,B^\alpha_{kl}}=
 -\ic B^{\alpha-1}_{ki} w_i^\alpha (b_\alpha)_{il}
  +2\ic \sum_{k' \leq i} B^{\alpha-1}_{kk'} w_{k'}^\alpha (b_\alpha)_{il}\,,
\ee
from which the first identity can be obtained. The second case is proved in the same way.
\end{proof}

\begin{lemma} \label{LemBB1}
For any $1\leq \alpha,\beta \leq d$,  $1\leq i,j,k,l\leq n$,
\begin{equation*}
 \begin{aligned}
\br{B^\alpha_{ij},B^\beta_{kl}}\stackrel{\alpha \leqslant \beta}{=}&
- 2 \ic  B^{\alpha}_{kj}\sum_{r > j}B^{\alpha}_{ir}  B^{\alpha+1;\beta}_{rl}
- \ic B^{\alpha}_{kj} B^{\alpha}_{ij}  B^{\alpha+1;\beta}_{jl}
+ 2 \ic \delta_{(i>k)} B^\alpha_{kj} B^\beta_{il}+ \ic  \delta_{ik} B^\alpha_{kj} B^\beta_{il}\,, \\
   \br{B^\alpha_{ij},\overline{B}^\beta_{kl}}\stackrel{\alpha \leqslant \beta}{=}&
-\ic B_{ij}^\alpha \overline{B}_{kj}^\alpha \overline{B}_{jl}^{\alpha+1;\beta}
- 2 \ic \sum_{s<j} B_{is}^\alpha \overline{B}_{ks}^\alpha \overline{B}_{jl}^{\alpha+1;\beta}
+\ic \delta_{ik} B_{ij}^\alpha \overline{B}_{kl}^\beta + 2 \ic \delta_{ik} \sum_{r>k} B_{rj}^\alpha \overline{B}_{rl}^\beta\,, \\
   \br{B^\alpha_{ij},\overline{B}^\beta_{kl}}\stackrel{\alpha \geqslant \beta}{=}&
-\ic B_{il}^\beta \overline{B}_{kl}^\beta B_{lj}^{\beta+1;\alpha}
- 2 \ic \sum_{s<l} B_{is}^\beta \overline{B}_{ks}^\beta B_{lj}^{\beta+1;\alpha}
+\ic \delta_{ik} B_{ij}^\alpha \overline{B}_{kl}^\beta + 2 \ic \delta_{ik} \sum_{r>k}  B_{rj}^\alpha \overline{B}_{rl}^\beta\,.
 \end{aligned}
\end{equation*}
\end{lemma}
\begin{proof}
 For the first equality, we have for $\alpha \leq \beta$ that
\be
 \br{B^\alpha_{ij},B^\beta_{kl}}
 =\sum_{1\leq \gamma\leq \alpha} \sum_{i',j',k',l'}(B^{\gamma-1})_{ii'} (B^{\gamma-1})_{kk'}
 \br{(b_\gamma)_{i'j'},(b_\gamma)_{k'l'}} B^{\gamma+1;\alpha}_{j'j} B^{\gamma+1;\beta}_{l'l}\,.
\ee
A similar expansion holds for $\br{B^\alpha_{ij},\overline{B}^\beta_{kl}}$. It then suffices to use Lemma \ref{Lemb1}.
\end{proof}
Note that in the case $\beta=\alpha$, the Poisson brackets from Lemma \ref{LemBB1} take
the usual form \eqref{T7}--\eqref{T8} on $\B(n)$.
We can also see that we can use $\beta=0$ in Lemma \ref{LemwB1} and $\alpha,\beta=0$ in Lemma \ref{LemBB1}, since
in such cases the Poisson bracket vanishes on $B^0=\1_n$.

\medskip

We can now  prove Lemma \ref{L:halfv} using Lemmae \ref{LemwB1} and \ref{LemBB1}. We will use
the definition of the half-dressed spins  given by \eqref{Eq:hvA}.
To show \eqref{Eq:Pvvhalf} we need to write
\begin{equation}
\begin{aligned} \label{Eq:vv1}
  \br{v^\alpha_i,v^\beta_k}=&\sum_{j,l}\br{B^{\alpha-1}_{ij} w^\alpha_j, B^{\beta-1}_{kl} w^\beta_l} \\
  =&
  \sum_{j,l}\br{B^{\alpha-1}_{ij} , B^{\beta-1}_{kl} } w^\alpha_j w^\beta_l
+\sum_{j,l}\br{B^{\alpha-1}_{ij} ,  w^\beta_l}  w^\alpha_j B^{\beta-1}_{kl} \\
&+\sum_{j,l}\br{w^\alpha_j, B^{\beta-1}_{kl} }  B^{\alpha-1}_{ij}  w^\beta_l
+\sum_{j,l}\br{w^\alpha_j, w^\beta_l}  B^{\alpha-1}_{ij} B^{\beta-1}_{kl}
  \end{aligned}
\end{equation}
where we assume $\alpha \leq \beta$ without loss of generality. We can then use Lemmae \ref{LemwB1}  and \ref{LemBB1} to show that
\be \label{Eq:vv2}
\br{v^\alpha_i,v^\alpha_k}=-\ic \, \sgn(k-i) v_{k}^\alpha v_i^\alpha\,; \quad
\br{v^\alpha_i,v^\beta_k} =-\ic \, \sgn(k-i) v_k^\alpha v_i^\beta +\ic v_k^\alpha v_i^\beta\,,\,\, \alpha < \beta\,.
\ee
By antisymmetry, \eqref{Eq:vv2} implies that \eqref{Eq:Pvvhalf} holds.

The Poisson bracket \eqref{Eq:Pvbarvhalf} is computed in the same way, and requires to remark in the case $\alpha=\beta$ that
\be
\sum_{s}  B^{\gamma}_{is} \overline{B}^{\gamma}_{ks} =\sum_{\mu=1}^\gamma v^\mu_i \bar{v}^\mu_k + \delta_{ik}\,.
\ee
This identity is equivalent to $B^\gamma (B^\gamma)^\dagger=\sum_{\mu=1}^\gamma v^\mu (v^\mu)^\dagger + \1_n$,
which is obtained by induction on $\gamma$ using   \eqref{bbdag}; it becomes \eqref{H7} when $\gamma=d$.

\section{Proof of Theorem \ref{Thm:redPB}} \label{A:RedPr}

Recall that we work over the gauge slice  $\ccM_{0,+}^\reg$ \eqref{RP2:1bis} and wish to compute the reduced
 Poisson brackets $\{\ ,\ \}_\red$
of the basic evaluation functions  $Q_j=e^{\ic q_j} \in \U(1)$ and $v(\alpha)_j\in \C$,
where the latter obey the relations \eqref{Eq:cU}.
Our fundamental tool will be the identity \eqref{Eq:PBrel}, which concerns  $\U(n)$ invariant functions
on $\cM$ and their pull-backs on  $\ccM_{0,+}^\reg$.
Knowing the left-hand side of \eqref{Eq:PBrel}, we will be able to determine the reduced Poisson brackets.
In the particular case at hand, we consider  the invariant functions
$f_m,f_m^{\alpha \beta}\in C^\infty(\cM)$ defined by \eqref{Eqf}.
Their Poisson brackets on $\cM$ are given by Lemma \ref{Lem:PBff}, and their restrictions (pull-backs) to $\ccM_{0,+}^\reg$
are displayed in \eqref{Eqfbi}.
The point is that the
right-hand side of \eqref{Eq:PBrel} can be also expressed through the reduced Poisson brackets
of the basic variables on $\ccM_{0,+}^\reg$, which enables us to
derive the explicit formulae of Theorem \ref{Thm:redPB}.

We begin by giving an auxiliary lemma, which will be used below.

\begin{lemma} \label{L:Cinv}
 The $n\times n$ matrices $\cE,\tilde{\cE}$ given by
 \begin{equation}
  \cE_{kl}=Q_l^k \quad \text{and} \quad \tilde{\cE}_{kl}=Q_l^k \cU_l
 \end{equation}
are invertible on $\ccM_{0,+}^\reg$.
\end{lemma}
\begin{proof}
 We can write that $\cE=V Q$ with $Q=\diag(Q_1,\ldots,Q_n)$ and $V=(V_{kl})$, $V_{kl}=Q_k^{l-1}$,
which is a Vandermonde matrix. Since $Q\in \T^n_{\reg}$ on $\ccM_{0,+}^\reg$, both $V$ and $Q$ are invertible.
We also have that $\tilde{\cE}=\cE D$ where $D=\diag(\cU_1,\ldots,\cU_n)$. As $\cU_j>0$ on $\ccM_{0,+}^\reg$,
$\tilde{\cE}$ is also invertible.
\end{proof}

\subsubsection*{Deriving \eqref{Eq:red1}} \label{ssA:Red1}

\begin{lemma} \label{L:invQQ}
 For any $i,j=1,\ldots,n$, $\br{q_i,q_j}_{\red}=0$.
\end{lemma}
\begin{proof}
 From \eqref{Eq:ff} and \eqref{Eqfbi} we get for any $M,N \in \N$,
 $$0=\xi^\ast \br{f_M,f_N}
 = \br{\xi^\ast f_M,\xi^\ast f_N}_{\red}=-MN\,\sum_{i,j=1}^n e^{\ic M q_i}e^{\ic N q_j} \br{q_i,q_j}_{\red}\,.$$
 Considering this equality for $M,N=1,\ldots,n$, this is equivalent to
 \begin{equation*}
  \cE \, \hat{U}^{(0)} \, \cE^T=0_{n \times n}\,,
 \end{equation*}
where $\hat{U}^{(0)} \in \Mat_{n \times n}(\C)$ is given by $\hat{U}^{(0)}_{kl}=\br{q_k,q_l}_{\red}$.
By Lemma \ref{L:Cinv}, $\cE$ is invertible on $\ccM_{0,+}^{\reg}$ so that $\hat{U}^{(0)}$ is the zero matrix.
\end{proof}

\begin{lemma} \label{L:invQV}
  For any $i,j=1,\ldots,n$,
  \begin{equation}
\br{\cU_i,q_j}_{\red}=-\delta_{ij}\cU_i\,, \quad
\br{v(\alpha)_i,q_j}_{\red}=-\delta_{ij} v(\alpha)_i\,, \quad
\br{\overline{v}(\alpha)_i,q_j}_{\red}=-\delta_{ij} \overline{v}(\alpha)_i\,.
  \end{equation}
\end{lemma}
\begin{proof}
 From \eqref{Eq:fspf}, after summing over all $\alpha,\beta$ we get for any $M,N \in \N$
 $$\sum_{i,j}\br{\cU_i^2 e^{\ic M q_i}, e^{\ic N q_j}}_{\red} =
 \sum_{\alpha,\beta}\br{\xi^\ast f_M^{\alpha \beta},\xi^\ast f_N}_{\red} =
 -2\ic N \sum_{\alpha,\beta} \xi^\ast f_{M+N}^{\alpha \beta} =
 -2\ic N \sum_i \cU_i^2 e^{\ic (M+N) q_i}\,.$$
 Using Lemma \ref{L:invQQ}, we obtain
  $$\sum_{i,j} e^{\ic M q_i}\cU_i e^{\ic N q_j}\br{\cU_i , q_j}_{\red} =
 -\sum_i \cU_i^2 e^{\ic (M+N) q_i}\,.$$
 We can rewrite this for $N,M=1,\ldots,n$ as
  \begin{equation*}
  \tilde{\cE} \, \hat{U}^{(1)} \, \cE^T=\tilde{\cE} \, {U}^{(1)} \, \cE^T\,,
 \end{equation*}
 where the $n\times n$ matrices are given by
 $\hat{U}^{(1)}_{kl}=\br{\cU_k,q_l}_{\red}$, ${U}^{(1)}_{kl}=-\delta_{kl}\cU_k$.
By Lemma \ref{L:Cinv}, both $\cE$ and $\tilde{\cE}$ are invertible. Hence $\hat{U}^{(1)}={U}^{(1)}$.

 For the second identity, we use \eqref{Eq:fspf} with summation over all $\beta$, and we get for any $M,N \in \N$
  $$\sum_{i,j}\br{\cU_i v(\alpha)_i e^{\ic M q_i}, e^{\ic N q_j}}_{\red} =
 -2\ic N \sum_i v(\alpha)_i\cU_i e^{\ic (M+N) q_i}\,.$$
 Now that the first identity is proved, we can use it to get
   $$\sum_{i,j}e^{\ic M q_i}\cU_i e^{\ic N q_j} \br{v(\alpha)_i , q_j}_{\red} =
 -\sum_i v(\alpha)_i \cU_i e^{\ic (M+N) q_i}\,.$$
 As before, we write this for $N,M=1,\ldots,n$ as
  \begin{equation*}
 \tilde{\cE} \, \hat{U}^{(2)} \, \cE^T=\tilde{\cE} \,  {U}^{(2)} \, \cE^T\,,
 \end{equation*}
 where the $n\times n$ matrices are given by
 $\hat{U}^{(2)}_{kl}=\br{v(\alpha)_k,q_l}_{\red}$, ${U}^{(2)}_{kl}=-\delta_{kl}v(\alpha)_k$.
 Again by invertibility of $\cE$ and $\tilde{\cE}$, we get $\hat{U}^{(2)}={U}^{(2)}$.

 The last identity follows from the second one by complex conjugation.
\end{proof}

From now on, we do not provide complete proofs of the different results that are stated.
They can be successively obtained by direct computations in the same way as we got Lemmae \ref{L:invQQ} and \ref{L:invQV}.

\subsubsection*{Deriving \eqref{Eq:red2}} \label{ssA:Red2}

We first need two preliminary lemmae.

\begin{lemma} \label{L:invVV} For any $i,j=1,\ldots,n$,
 \begin{equation}
 \begin{aligned} \label{Eq:LinvVV}
    \br{\cU_i,\cU_j}_{\red}=& \,\,\frac12 \ic \delta_{(i\neq j)} \frac{Q_i+Q_j}{Q_i-Q_j}\cU_i\cU_j
  +\frac14 \ic \sum_{\mu,\nu} \sgn(\nu-\mu) \left[v(\nu)_i v(\mu)_j - \overline{v}(\nu)_i \overline{v}(\mu)_j   \right] \\
  &+\frac14 \ic \sum_{\mu}\left[ v(\mu)_i \overline{v}(\mu)_j - v(\mu)_j \overline{v}(\mu)_i \right] +  \frac{d}{2} \ic (L_{ij}-L_{ji})  \\
  &+\frac12 \ic \sum_{\nu}\sum_{\mu<\nu}\left[ v(\mu)_i \overline{v}(\mu)_j - v(\mu)_j \overline{v}(\mu)_i \right]\,.
 \end{aligned}
 \end{equation}
\end{lemma}
\begin{proof}
 It suffices to use \eqref{Eq:fspfsp} where we sum over all $\alpha,\beta,\gamma,\epsilon$.
 After elementary manipulations, we arrive at
   \begin{equation}
\begin{aligned}
  \sum_{i,j} Q_i^M\cU_i Q_j^N \cU_j \br{\cU_i,\cU_j}_{\red}
=&\sum_{i,j}Q_i^M\cU_i Q_j^N \cU_j \, {U}^{(3)}_{ij} \,,
\end{aligned}
\end{equation}
where ${U}^{(3)}_{ij}$ is  the right-hand side of \eqref{Eq:LinvVV}.
We can then write the equalities with $N,M=1,\ldots,n$ as
  \begin{equation}
   \tilde{\cE} \, \hat{U}^{(3)} \, \tilde{\cE}^T=\tilde{\cE} \,  {U}^{(3)} \, \tilde{\cE}^T\,,
 \end{equation}
 where the $n\times n$ matrix $\hat{U}^{(3)}$ is given by  $\hat{U}^{(3)}_{kl}=\br{\cU_k,\cU_l}_{\red}$.
 By invertibility of $\tilde  \cE$, this proves the claim \eqref{Eq:LinvVV}.
\end{proof}

\begin{lemma} \label{L:invValV} For any $i,j=1,\ldots,n$,
 \begin{equation}
 \begin{aligned} \label{Eq:ValV}
    \br{v(\alpha)_i,\cU_j}_{\red}=& \,\,\frac12 \ic \delta_{(i\neq j)} \frac{Q_i+Q_j}{Q_i-Q_j}v(\alpha)_j \cU_i
  +\frac12 \ic \sum_\kappa \sgn(\kappa-\alpha) v(\alpha)_j v(\kappa)_i \\
  &-\frac14 \ic\frac{v(\alpha)_i}{\cU_i} \sum_{\mu,\nu} \sgn(\nu-\mu) \left[v(\nu)_i v(\mu)_j + \overline{v}(\nu)_i \overline{v}(\mu)_j   \right] \\
  &+\frac12 \ic v(\alpha)_i \overline{v}(\alpha)_{j}
  -\frac14 \ic \frac{v(\alpha)_i}{\cU_i}\sum_{\mu}\left[ v(\mu)_i \overline{v}(\mu)_j + v(\mu)_j \overline{v}(\mu)_i \right]   \\
  & + \ic \sum_{\kappa< \alpha} v(\kappa)_i \overline{v}(\kappa)_j
  -\frac12 \ic \frac{v(\alpha)_i}{\cU_i}\sum_{\nu}\sum_{\mu<\nu}\left[ v(\mu)_i \overline{v}(\mu)_j + v(\mu)_j \overline{v}(\mu)_i \right] \\
  &+\frac12 \ic \left[ 2L_{ij} - d \frac{v(\alpha)_i}{\cU_i} (L_{ij}+L_{ji})\right]\,.
 \end{aligned}
 \end{equation}
\end{lemma}
\begin{proof}
 It suffices to use \eqref{Eq:fspfsp} after summing over $\beta,\gamma,\epsilon$. We arrive at
      \begin{equation}
 \begin{aligned} \label{Eq:ValV4}
 \sum_{i,j}  Q_i^M \cU_i Q_j^N  \cU_j   \br{v(\alpha)_i , \cU_j}_{\red}
=\sum_{i,j}Q_i^M\cU_i Q_j^N \cU_j \, {U}^{(4)}_{ij}\,,
 \end{aligned}
 \end{equation}
where ${U}^{(4)}_{ij}$ is  the right-hand side of \eqref{Eq:ValV}.
We can then write the equalities with $N,M=1,\ldots,n$ as
  \begin{equation*}
  \tilde{\cE} \, \hat{U}^{(4)} \, \tilde{\cE}^T= \tilde{\cE} \,  {U}^{(4)} \, \tilde{\cE}^T\,,
 \end{equation*}
 where the $n\times n$ matrix $\hat{U}^{(4)}$ is given by  $\hat{U}^{(4)}_{kl}=\br{v(\alpha)_k,\cU_l}_{\red}$.
 By invertibility of $\tilde \cE$, we obtain the equality \eqref{L:invValV}.
\end{proof}

Summing over $\beta,\epsilon$ in \eqref{Eq:fspfsp} and using the previous results, we can get
 \begin{equation} \label{Eq:ValVgamLast}
\sum_{i,j} Q_i^M \cU_i Q_j^N \cU_j \br{ v(\alpha)_i  , v(\gamma)_j }_{\red}
=\sum_{i,j}Q_i^M\cU_i Q_j^N \cU_j \,{U}^{(5)}_{ij}\,,
\end{equation}
 where ${U}^{(5)}_{ij}$ is  the right-hand side of \eqref{Eq:red2}.
We can then write the equalities \eqref{Eq:ValVgamLast} with $N,M=1,\ldots,n$ as
  \begin{equation*}
  \tilde{\cE} \, \hat{U}^{(5)} \, \tilde{\cE}^T=\tilde{\cE} \,  {U}^{(5)} \, \tilde{\cE}^T\,,
 \end{equation*}
 where the $n\times n$ matrix $\hat{U}^{(5)}$ is given by  $\hat{U}^{(5)}_{kl}=\br{v(\alpha)_k,v(\gamma)_l}_{\red}$.
 By invertibility of $\tilde \cE$, this implies that \eqref{Eq:red2} holds.

\subsubsection*{Deriving \eqref{Eq:red3}} \label{ssA:Red3}

By antisymmetry and complex conjugation, we get $\br{\cU_i,\overline{v}(\epsilon)_j}_{\red}$ from Lemma \ref{L:invValV}.
 We can then use the previous results as well as \eqref{Eq:fspfsp} after summing over $\beta,\gamma$ in order to get
    \begin{equation} \label{Eq:ValVepsFin}
\sum_{i,j} Q_i^M \cU_i Q_j^N  \cU_j  \br{ v(\alpha)_i  , \overline{v}(\epsilon)_j }_{\red}
=\sum_{i,j}Q_i^M\cU_i Q_j^N \cU_j {U}^{(6)}_{ij}\,,
 \end{equation}
 where ${U}^{(6)}_{ij}$ is  the right-hand side of \eqref{Eq:red3}.
We can then write the equalities \eqref{Eq:ValVepsFin} with $N,M=1,\ldots,n$ as
  \begin{equation*}
  \tilde{\cE} \, \hat{U}^{(6)} \, \tilde{\cE}^T=\tilde{\cE} \,  {U}^{(6)} \, \tilde{\cE}^T\,,
 \end{equation*}
 where the $n\times n$ matrix $\hat{U}^{(6)}$ is given by
 $\hat{U}^{(6)}_{kl}=\br{v(\alpha)_k,\overline{v}(\epsilon)_l}_{\red}$.
 By invertibility of $\tilde{\cE}$, we can conclude that \eqref{Eq:red3} holds.

\section{Poisson brackets of collective spins}  \label{ssA:Coll2}

Recall the matrix $(S_{ij})$ defined before Theorem \ref{Thm:redPB}.
The reduced Poisson brackets of the so-called collective spins $F$ \eqref{H21} can be computed in the following form.
  \begin{lemma} \label{L:invFF}
Denoting $q_{ab}:=q_a-q_b$, the following identity holds on $\ccM_{0,+}^\reg$
 \begin{equation*}
  \begin{aligned}
\br{&F_{ij},F_{kl}}_{\red} =
\,\,\ic \left(\frac{S_{ik}}{\cU_i \cU_k}-\frac{S_{lj}}{\cU_l \cU_j} + \frac{S_{kj}}{\cU_k \cU_j} -
\frac{S_{il}}{\cU_i \cU_l} \right) F_{ij} F_{kl} \\
&+\frac12 \left[\delta_{(i \neq k)} \cot(\frac{q_{ik}}{2}) +
\delta_{(j\neq l)}  \cot(\frac{q_{jl}}{2}) + \delta_{(k \neq j)}  \cot(\frac{q_{kj}}{2}) +
\delta_{(i \neq l)}  \cot(\frac{q_{li}}{2})\right] \, F_{ij}F_{kl}\\
&+\left[ \delta_{(i \neq k)} \cot(\frac{q_{ik}}{2}) +
\delta_{(j \neq l)}  \cot(\frac{q_{jl}}{2}) -  \cot(\frac{q_{jk}}{2}-\ic \gamma) +
\cot(\frac{q_{li}}{2} - \ic \gamma) \right] \, F_{il}F_{kj} \\
&+\frac12 \left[\delta_{(k \neq i)} \cot(\frac{q_{ki}}{2}) -  \cot(\frac{q_{li}}{2}-\ic \gamma) \right]\frac{\cU_k}{\cU_i} F_{ij}F_{il}
+\frac12 \left[\delta_{(j \neq k)} \cot( \frac{q_{jk}}{2} ) + \cot( \frac{q_{lj}}{2}-\ic \gamma ) \right] \frac{\cU_k}{\cU_j} F_{ij} F_{jl} \\
& +\frac12 \left[\delta_{(i\neq k)}\cot( \frac{q_{ki}}{2} ) + \cot( \frac{q_{jk}}{2}-\ic \gamma ) \right] \frac{\cU_i}{\cU_k} F_{kj} F_{kl}
+\frac12 \left[\delta_{(i \neq l)} \cot( \frac{q_{il}}{2} ) - \cot( \frac{q_{jl}}{2} - \ic \gamma) \right] \frac{\cU_i}{\cU_l} F_{lj}F_{kl} \\
&+\frac12 \left[ \delta_{(i \neq l)} \cot( \frac{q_{il}}{2} ) - \cot( \frac{q_{ik}}{2} - \ic \gamma )\right] \frac{\cU_l}{\cU_i} F_{ij} F_{ki}
+\frac12 \left[ \delta_{(l \neq j)} \cot( \frac{q_{lj}}{2} ) + \cot( \frac{q_{jk}}{2} - \ic \gamma) \right] \frac{\cU_l}{\cU_j} F_{ij} F_{kj} \\
&+\frac12 \left[ \delta_{(j \neq k)} \cot( \frac{q_{jk}}{2} ) + \cot( \frac{q_{ki}}{2} - \ic \gamma )\right] \frac{\cU_j}{\cU_k} F_{ik} F_{kl}
+\frac12 \left[ \delta_{(j \neq l)} \cot( \frac{q_{lj}}{2} ) - \cot( \frac{q_{li}}{2} - \ic \gamma ) \right] \frac{\cU_j}{\cU_l} F_{il} F_{kl}
\end{aligned}
 \end{equation*}
 \end{lemma}
 This follows from Theorem \ref{Thm:redPB} by direct calculation.
The reader can easily check  the reality condition
$\br{F_{ji},F_{lk}}_{\red}=\br{\overline{F}_{ij},\overline{F}_{kl}}_{\red}=\overline{ \br{F_{ij},F_{kl}} }_{\red} $.
Taking $i=j$ and $k=l$ in Lemma \ref{L:invFF}, everything cancels out except for the third line, which can be rewritten as follows:
\be
\br{ F_{jj},F_{kk}}_{\red} =
 F_{jk} F_{kj}  \frac{2\cot(\frac{q_{jk}}{2})}{1+\sinh^{-2}(\gamma)\,\sin^2(\frac{q_{jk}}{2})}\,, \quad \text{for }j\neq k\,.
\label{r1}\ee
Let us now assume that $d=1$, so that $F_{jk} F_{kj} = F_{jj} F_{kk}$.
Note that the formula of $L$ \eqref{H21} shows that $F_{jj} >0$.
Motivated by the form of the equations of motion \eqref{I8} and the spinless Hamiltonian \eqref{I:Hplus}, we
 make the change of variables
 \be
 F_{jj} = e^{2\theta_j}\prod_{i\neq j} \left[1+ \frac{\sinh^2\gamma}
{1 + \sin^2\frac{q_i - q_j}{2}}\right]^{\frac{1}{2}}.
\label{r3}\ee
Using \eqref{Eq:red1} and \eqref{r1}, it turns out that $(q_j, \theta_j)$ are Darboux variables, and we recover the standard chiral RS Hamiltonian \eqref{I:Hplus} for $\cH = \sum_j F_{jj}$.


\begin{thebibliography}{99}

\addcontentsline{toc}{section}{References}

    \setlength{\parskip}{0em}

 \bibitem{AM}
 A.Yu.~Alekseev and  A.Z.~Malkin,
 {\it Symplectic structures associated to Lie--Poisson groups},
   Comm. Math. Phys. {\bf 162} (1994) 147-173;
 \href{https://arxiv.org/abs/hep-th/9303038}{\tt arXiv:hep-th/9303038}

\bibitem{AMM}
A.~Alekseev, A.~Malkin and E.~Meinrenken,
{\it Lie group valued moment maps},
J. Differential Geom. {\bf 48} (1998) 445-495;
 \href{https://arxiv.org/abs/dg-ga/9707021}{\tt arXiv:dg-ga/9707021}

 \bibitem{A}
 G.~Arutyunov,
Elements of Classical and Quantum Integrable Systems, Springer, 2019

 \bibitem{AF}
 G.E.~Arutyunov and  S.A.~Frolov,
 {\it On the Hamiltonian structure of the spin Ruijsenaars--Schneider model},
  J. Phys. A {\bf 31} (1998) 4203-4216;
  \href{https://arxiv.org/abs/hep-th/9703119}{\tt arXiv:hep-th/9703119}


\bibitem{AO}
G.~Arutyunov and E.~Olivucci,
{\it Hyperbolic spin Ruijsenaars--Schneider model from Poisson reduction},
Proc. Steklov Inst. Math. {\bf 309} (2020) 31-45;
 \href{https://arxiv.org/abs/1906.02619}{\tt  arXiv:1906.02619}


\bibitem{AR}
S.~Arthamonov and N.~Reshetikhin,
{\it Superintegrable systems on moduli spaces of flat connections},
 \href{https://arxiv.org/abs/1909.08682}{\tt  arXiv:1909.08682}

 \bibitem{BH}
 H.W.~Braden and N.W.~Hone,
 {\it Affine Toda solitons and systems of Calogero--Moser type},
Phys. Lett. B {\bf 380} (1996) 296-302;
  \href{https://arxiv.org/abs/hep-th/9603178}{\tt arXiv:hep-th/9603178}


\bibitem{Cal}
 F.~Calogero,
{\it Solution of the one-dimensional N-body problem with quadratic and/or
inversely quadratic pair potentials},
J. Math. Phys. {\bf 12} (1971) 419-436

\bibitem{CF1}
O.~Chalykh and M.~Fairon,
{\it Multiplicative quiver varieties and generalised Ruijsenaars-–Schneider models},
J. Geom. Phys. {\bf 121} (2017) 413-437;
 \href{https://arxiv.org/abs/1704.05814}{\tt arXiv:1704.05814}

 \bibitem{CF2}
O.~Chalykh and M.~Fairon,
{\it On the Hamiltonian formulation of the trigonometric spin Ruijsenaars--Schneider system},
 Lett. Math. Phys. (2020);
 \href{https://arxiv.org/abs/1811.08727}{\tt arXiv:1811.08727}


\bibitem{CBS}
W.~Crawley-Boevey and P.~Shaw,
{\it Multiplicative preprojective algebras, middle convolution and the Deligne--Simpson problem},
Adv. Math. {\bf 201} (2006) 180-208;
 \href{https://arxiv.org/abs/math/0404186}{\tt arXiv:math/0404186}

\bibitem{E}
P.~Etingof, Calogero--Moser Systems and Representation Theory,
 European Mathematical Society, 2007



\bibitem{EFK}
 P.I.~Etingof, I.B.~Frenkel and A.A.~Kirillov Jr.,
{\it Spherical functions on affine Lie groups},
Duke Math. J. {\bf 80} (1995) 59-90;
 \href{https://arxiv.org/abs/hep-th/9407047}{\tt  arXiv:hep-th/9407047}

\bibitem{Fai}
M.~Fairon, {\it Spin versions of the complex trigonometric Ruijsenaars--Schneider model from cyclic quivers},
 J. Integrable Syst. {\bf 4} (2019) xyz008, 55 pp;
  \href{https://arxiv.org/abs/1811.08717}{\tt arXiv:1811.08717}


\bibitem{F1}
L.~Feh\'er, {\it Poisson--Lie analogues of spin Sutherland models},
Nucl. Phys. B {\bf 949} (2019) 114807, 26 pp;
 \href{https://arxiv.org/abs/1809.01529}{\tt arXiv:1809.01529}


\bibitem{F3}
L.~Feh\'er,
{\it Reduction of a bi-Hamiltonian hierarchy on {$T^*{\rm U}(n)$} to spin Ruijsenaars-Sutherland models},
 Lett. Math. Phys. {\bf 110} (2020) 1057-1079;
 \href{https://arxiv.org/abs/1908.02467}{\tt arXiv:1908.02467}

\bibitem{FA}
L.~Feh\'er and V.~Ayadi,
{\it Trigonometric Sutherland systems and their Ruijsenaars duals from symplectic reduction},
  J. Math. Phys. {\bf 51} (2010), 103511, 30 pp.;
 \href{https://arxiv.org/abs/1005.4531}{\tt arXiv:1005.4531}

\bibitem{FG}
L.~Feh\'er and T.~G\"{o}rbe,
{\it On a Poisson--Lie deformation of the {${\rm BC}_n$} Sutherland system},
 Nucl. Phys. B {\bf 901} (2015) 85-114;
 \href{https://arxiv.org/abs/1508.04991}{\tt arXiv:1508.04991}



\bibitem{FK2}
L.~Feh\'er and C.~Klim\v c\'\i k,
{\it Poisson--Lie interpretation of trigonometric Ruijsenaars duality},
Commun. Math. Phys. {\bf 301} (2011) 55-104;
  \href{https://arxiv.org/abs/0906.4198}{\tt arXiv:0906.4198}

\bibitem{FM}
 L.~Feh\'er and I.~Marshall,
 {\it Global description of action-angle duality
  for a Poisson--Lie deformation of the trigonometric
  {${\rm BC}_n$} Sutherland system}, Annales Henri Poincar\'{e} {\bf 20} (2019) 1217-1262;
 \href{https://arxiv.org/abs/1710.08760}{\tt arXiv:1710.08760}


\bibitem{FP2}
L.~Feh\'er and B.G.~Pusztai,
{\it A class of Calogero type reductions of free motion on a simple Lie group},
 Lett. Math. Phys. {\bf 79} (2007) 263-277;
 \href{https://arxiv.org/abs/math-ph/0609085}{\tt arXiv:math-ph/0609085}

\bibitem{FP3}
L.~Feh\'er and B.G.~Pusztai,
{\it Hamiltonian reductions of free particles under polar actions of compact Lie groups},
 Theor. Math. Phys. {\bf 155} (2008) 646-658;
 \href{https://arxiv.org/abs/0705.1998}{\tt arXiv:0705.1998}

\bibitem{JHEP}
V.~Fock,  A.~Gorsky, N.~Nekrasov and V.~Rubtsov,
{\it Duality in integrable systems and gauge theories},
JHEP 07(2000)028;  \href{https://arxiv.org/abs/hep-th/9906235}{\tt arXiv:hep-th/9906235}


\bibitem{GH}
J.~Gibbons and T.~Hermsen,
{\it A generalisation of the Calogero--Moser system},
Physica D {\bf 11} (1984) 337-348

\bibitem{HT}
M.~Henneaux and C.~Teitelboim,
Quantization of Gauge Systems, Princeton University Press, 1992

\bibitem{Hi}
 J.~Hilgert, K.-H.~Neeb and W.~Plank,
{\it Symplectic convexity theorems and coadjoint orbits},
Compositio Mathematica {\bf 94} (1994) 129-180


\bibitem{J}
B.~Jovanovic,
{\it Symmetries and integrability},
Publ. Institut Math. {\bf 49} (2008) 1-36;
 \href{https://arxiv.org/abs/0812.4398}{\tt arXiv:0812.4398}

\bibitem{KLOZ}
 S.~Kharchev, A.~Levin, M.~Olshanetsky and  A.~Zotov,
 {\it Quasi-compact Higgs bundles and Calogero--Sutherland systems with two types spins},
 J. Math. Phys. {\bf 59} (2018) 103509, 36~pp;
  \href{https://arxiv.org/abs/1712.08851}{\tt arXiv:1712.08851}

\bibitem{Kli}
C.~Klim\v c\'\i k,
{\it On moment maps associated to a twisted Heisenberg double},
 Rev. Math. Phys. {\bf 18} (2006) 781-821;
  \href{https://arxiv.org/abs/math-ph/0602048}{\tt  arXiv:math-ph/0602048}

\bibitem{KS}
L.I.~Korogodski and Y.S.~Soibelman,
Algebras of Functions on Quantum Groups: Part I,
American Mathematical Society, 1998

\bibitem{Kri}
 I.~Krichever,
{\it Elliptic Solutions to Difference Nonlinear Equations and Nested Bethe Ansatz Equations},
pp. 249-271; in: Calogero-Moser-Sutherland Models (ed. by J.F. van Diejen, L. Vinet), CRM Series in Mathematical Physics,
Springer, New York (2000);
\href{https://arxiv.org/abs/solv-int/9804016}{\tt arXiv:solv-int/9804016}

\bibitem{KZ}
 I.~Krichever and A.~Zabrodin,
{\it Spin generalization of the Ruijsenaars--Schneider model, non-abelian 2D
Toda chain and representations of Sklyanin algebra},
Russian Math. Surveys {\bf 50} (1995) 1101-1150;
 \href{https://arxiv.org/abs/hep-th/9505039}{\tt arXiv:hep-th/9505039}


\bibitem{Li}
L.-C.~Li,
{\it Poisson involutions, spin Calogero--Moser systems associated
with symmetric Lie subalgebras and the symmetric space spin Ruijsenaars--Schneider models,}
Commun. Math. Phys. {\bf 265} (2006) 333-372;
 \href{https://arxiv.org/abs/math-ph/0506025}{\tt arXiv:math-ph/0506025}

\bibitem{LX}
L.-C.~Li and P.~Xu,
{\it A class of integrable spin Calogero--Moser systems},
Commun. Math. Phys. {\bf 231} (2002) 257-286;
 \href{https://arxiv.org/abs/math/0105162}{\tt arXiv:math/0105162}

\bibitem{Lu1}
J.-H.~Lu, Multiplicative and affine Poisson structures on Lie groups, Ph.D. Thesis,
 University of California, Berkeley (1990) 74 pp

\bibitem{Lu}
J.-H.~Lu,
{\it Momentum mappings and reduction of Poisson actions},
 pp. 209-226; in: Symplectic Geometry, Groupoids, and Integrable Systems, Springer, 1991

\bibitem{MF}
A.S.~Mischenko and A.T.~Fomenko,
{\it  Generalized Liouville method for integrating Hamiltonian systems},
Funct. Anal. Appl. {\bf 12} (1978) 113-125


 \bibitem{M}
J.~Moser,
{\it Three integrable Hamiltonian systems connected with isospectral deformations},
Adv. Math. {\bf 16} (1975) 197-220


\bibitem{Nekh}
N.N.~Nekhoroshev,
{\it  Action-angle variables and their generalizations},
Trans. Moscow Math. Soc. {\bf 26} (1972) 180-197


\bibitem{N}
N.~Nekrasov,
{\it Infinite-dimensional algebras, many-body systems and gauge theories},
pp. 263-299 in: Moscow Seminar in Mathematical Physics,
AMS Transl. Ser. 2, Vol.~191, American Mathematical Society, 1999

\bibitem{OR}
J.-P.~Ortega and T.~Ratiu,
Momentum Maps and Hamiltonian Reduction, Birkh\"auser, 2004


\bibitem{P}
M.~Penciak,
{\it Spectral description of the spin Ruijsenaars--Schneider system},
 \href{https://arxiv.org/abs/1909.08107}{\tt  arXiv:1909.08107}

\bibitem{RaSu}
O.~Ragnisco and Yu.B.~Suris,
{\it Integrable discretizations of the spin Ruijsenaars--Schneider models},
J. Math. Phys. {\bf 38} (1997) 4680-4691;
 \href{https://arxiv.org/abs/solv-int/9605001}{\tt arXiv:solv-int/9605001}

\bibitem{Res1}
N.~Reshetikhin,
{\it Degenerate integrability of spin Calogero--Moser systems and the
duality with the spin Ruijsenaars systems},
Lett. Math. Phys. {\bf 63} (2003) 55-71;
  \href{https://arxiv.org/abs/math/0202245}{\tt arXiv:math/0202245}


\bibitem{Res2}
 N.~Reshetikhin,
{\it Degenerately integrable systems},
J. Math. Sci. {\bf 213} (2016) 769-785;
 \href{https://arxiv.org/abs/1509.00730}{\tt arXiv:1509.00730}

\bibitem{Res3}
N.~Reshetikhin,
{\it Spin Calogero--Moser models on symmetric spaces},
 \href{https://arxiv.org/abs/1903.03685}{\tt  arXiv:1903.03685}

\bibitem{ReSt}
N.~Reshetikhin and J.~Stokman,
{\it N-point spherical functions and asymptotic boundary KZB equations},
 \href{https://arxiv.org/abs/2002.02251}{\tt  arXiv:2002.02251}

\bibitem{Rud}
G.~Rudolph and M.~Schmidt,
Differential Geometry and Mathematical Physics. Part I. Manifolds, Lie Groups and
Hamiltonian Systems, Springer, 2013

\bibitem{RBanff}
S.N.M.~Ruijsenaars,
{\it Systems of Calogero--Moser type}, pp. 251-352;
in: Proceedings of the 1994 CRM-Banff Summer School: Particles and Fields, Springer, 1999

\bibitem{RS}
S.N.M.~Ruijsenaars and H.~Schneider,
{\it A new class of integrable systems and its relation to solitons},
Ann. Phys. {\bf 170} (1986) 370-405

\bibitem{SS}
  V.~Schomerus and E.~Sobko,
 {\it From spinning conformal blocks to matrix Calogero--Sutherland models},
 JHEP04(2018)052;
  \href{https://arxiv.org/abs/1711.02022}{\tt arXiv:1711.02022}

\bibitem{STS}
  M.A.~Semenov-Tian-Shansky,
{\it Dressing transformations and Poisson group actions},
Publ. RIMS {\bf 21} (1985) 1237-1260

\bibitem{Sol}  F.L.~Soloviev,
{\it On the {H}amiltonian form of the equations of the elliptic spin  {R}uijsenaars--{S}chneider  model},
 Russian Math. Surveys {\bf 64} (2009) 1142-1144;
 \href{https://arxiv.org/abs/0808.3875}{\tt arXiv:0808.3875}

\bibitem{S}
B.~Sutherland,
{\it Exact results for a quantum many-body problem in one dimension},
Phys. Rev. A {\bf 4} (1971) 2019-2021


\bibitem{vDV}
J.F.~van Diejen and L.~Vinet (Editors), Calogero--Moser--Sutherland Models, Springer, 2000


\bibitem{Z}
S.~Zakrzewski,
{\it Phase spaces related to standard classical $r$-matrices},
J. Phys. A: Math. Gen. {\bf 29} (1996) 1841-1857;
 \href{https://arxiv.org/abs/q-alg/9511002}{\tt arXiv:q-alg/9511002}


\bibitem{Zo}
A.V.~Zotov,
{\it Relativistic interacting integrable elliptic tops},
 Teoret. Mat. Fiz. {\bf 201} (2019), 175-192;
 \href{https://arxiv.org/abs/1910.08246}{\tt arXiv:1910.08246}

\bibitem{Zung}
N.T.~Zung,
{\it Torus actions and integrable systems},
pp. 289-328 in: Topological Methods in the Theory of Integrable Systems,
Camb. Sci. Publ., 2006;
 \href{https://arxiv.org/abs/math/0407455}{\tt arXiv:math/0407455}



\end{thebibliography}
 \end{document}